\documentclass[12pt,onecolumn,twoside]{IEEEtran}

\usepackage{calc}

\usepackage{multirow}
\usepackage{cite}
\usepackage{graphicx,subfigure}
\usepackage{psfrag}
\usepackage{amsmath,amssymb,amsthm}
\usepackage{color}
\usepackage{tikz} 
\usepackage{bm}
\usepackage{bbm}
\usepackage{verbatim}
\usepackage[T1]{fontenc}

\usepackage{cases}
\usepackage[noend]{algpseudocode}

\usepackage{xcolor}
\usepackage[most]{tcolorbox}
\usepackage{bm}

\usepackage{tabu}
\usepackage{mathtools}
\usepackage{xifthen}
\usepackage{booktabs}
\usepackage{algorithm}
\usepackage{algpseudocode}
\usepackage{arydshln}

\makeatletter
\def\BState{\State\hskip-\ALG@thistlm}
\makeatother

\newtheorem{theorem}{Theorem}

\newtheorem{lemma}{Lemma}
\newtheorem{definition}{Definition}

\newcommand{\rv}{\mathbf{r}}

\newcommand{\bit}{\begin{itemize}}
\newcommand{\eit}{\end{itemize}}

\newcommand{\bc}{\begin{center}}
\newcommand{\ec}{\end{center}}

\newcommand{\ba}{\begin{array}}
\newcommand{\ea}{\end{array}}

\newcommand{\beq}{\begin{equation}}
\newcommand{\eeq}{\end{equation}}

\newcommand{\beqn}{\begin{equation*}}
\newcommand{\eeqn}{\end{equation*}}

\newcommand{\bean}{\begin{eqnarray*}}
\newcommand{\eean}{\end{eqnarray*}}
\newcommand{\bea}{\begin{eqnarray}}
\newcommand{\eea}{\end{eqnarray}}

\def\av{\boldsymbol{a}}

\def\cv{\boldsymbol{c}}

\def\rv{\boldsymbol{r}}

\def\vv{\boldsymbol{v}}
\def\wv{\boldsymbol{w}}

\def\yv{\boldsymbol{y}}

\newcommand{\Ac}{{\mathcal A}}
\newcommand{\Bc}{{\mathcal B}}
\newcommand{\Cc}{{\mathcal C}}
\newcommand{\Dc}{{\mathcal D}}

\newcommand{\Fc}{{\mathcal F}}
\newcommand{\Gc}{{\mathcal G}}

\newcommand{\Ic}{{\mathcal I}}

\newcommand{\Lc}{{\mathcal L}}
\newcommand{\Nc}{{\mathcal N}}
\newcommand{\Oc}{{\mathcal O}}
\newcommand{\Pc}{{\mathcal P}}
\newcommand{\Qc}{{\mathcal Q}}
\newcommand{\Rc}{{\mathcal R}}
\newcommand{\Sc}{{\mathcal S}}
\newcommand{\Tc}{{\mathcal T}}

\newcommand{\Wc}{{\mathcal W}}
\newcommand{\Vc}{{\mathcal V}}

\newtheorem{remark}{Remark}

\algnewcommand{\IfThenElse}[3]{  \State \algorithmicif\ #1\ \algorithmicthen\ #2\ \algorithmicelse\ #3}

\newcommand{\Me}{\wv}

\newcommand{\Tx}{T}

\newcommand{\TxAccept}{\Tc_{\mathrm{TxAccept}}}

\newcommand{\TxHeightAcceptLock}{\Tc_{\mathrm{HeightAcceptLock}}}
\newcommand{\TxHeightAcceptTemp}{\Tc_{\mathrm{HeightAccept}}}

\newcommand{\MaxNumTx}{m_{\mathrm{max}}}

\newcommand{\eon}{e}

\newcommand{\eonstar}{e^{\star}}

\newcommand{\AW}{\mathrm{AW}}

\newcommand{\tx}{{tx}}     
\newcommand{\SigAccept}{\Tc_{\mathrm{SigAccept}}}

\newcommand{\idtx}{{id}\_{tx}}

\newcommand{\idtxinitial}{{id}\_{tx}\_{ini}}  
\newcommand{\siginitial}{{sig}\_{ini}}  
     
\newcommand{\sigoptuple}{{sig}\_{op}\_{tuple}}   
   
\newcommand{\sigoplist}{{sig}\_{op}\_{list}}

\newcommand{\sigop}{{sig}\_{op}}     
\newcommand{\sigvp}{{sig}\_{vp}} 
       
\newcommand{\size}{\mathrm{size}}

\newcommand{\Break}{\textbf{break}}

\newcommand{\append}{.\mathrm{append}}

\newcommand{\NumProposed}{m} 
\newcommand{\numproposedstar}{\NumProposed^{\star}} 
\newcommand{\numproposeddiamond}{\NumProposed^{\diamond}}

\newcommand{\rtuple}{)} 
\newcommand{\ltuple}{(}

\newcommand{\inset}{\in}

\newcommand{\AND}{\land}    
\newcommand{\OR}{\lor}  
\newcommand{\idoptuple}{{id}\_{op}\_{tuple}}   
\newcommand{\idoptupleprime}{{id}\_{op}\_{tuple}'}

\newcommand{\range}{\mathrm{range}}   
\newcommand{\inamount}{{in}\_{amount}}

\newcommand{\idop}{{id}\_{op}}

\newcommand{\Global}{\mathrm{global}}    
\newcommand{\tuple}{\mathrm{tuple}}

\newcommand{\pop}{.\mathrm{pop}}

\newcommand{\import}{\mathrm{import}}

\newcommand{\eqlog}{=}

\newcommand{\op}{\mathrm{op}}

\newcommand{\numsigop}{num\_{op}}

\newcommand{\Ocior}{\mathrm{Ocior}}

\newcommand{\height}{h} 
\newcommand{\heightstar}{h^\star} 
 
\newcommand{\heightdiamond}{h^{\diamond}}

\newcommand{\VP}{\mathrm{VP}}

\newcommand{\heightdict}{\mathcal{H}_\mathrm{Height}}   
\newcommand{\heightdictLock}{\mathcal{H}_{\mathrm{HeightLock}}}

\newcommand{\opindextuple}{{op}\_{index}\_{tuple}}   
\newcommand{\opindextupleprime}{{op}\_{index}\_{tuple}'}

\newcommand{\opindex}{{op}\_{index}}     
\newcommand{\opreceiver}{{op}\_{receiver}}     
\newcommand{\sender}{{sender}}   
\newcommand{\senderprime}{{sender}'}

\newcommand{\Txochildren}{\Tc_\mathrm{OChildren}}

\newcommand{\checkindicator}{{check}\_{indicator}}

\newcommand{\idtxconflict}{{id}\_{tx}\_{conflict}}

\newcommand{\eondiamond}{\eon^{\diamond}}

\newcommand{\txinitial}{{tx}\_{ini}}

\newcommand{\contenthash}{content\_h}

\newcommand{\votelts}{votel}

\newcommand{\ltsnodeindex}{j'}   
         
\newcommand{\LTSIndexBook}{\Bc_\mathrm{LTSBook}}

\newcommand{\LTSGroupAgg}{\Ac_\mathrm{lts}}  
  
\newcommand{\TSGroupAgg}{\Ac_\mathrm{ts}}

\newcommand{\txconflict}{{tx}\_{conflict}}

\newcommand{\MVBAInputMsg}{\wv}

\newcommand{\EncodedSymbol}{y}

\newcommand{\vectorcommitment}{C}                    
\newcommand{\proofpositionvc}{\omega}                    
\newcommand{\VCOpen}{\mathsf{VC.Open}}                  
\newcommand{\VCCom}{\mathsf{VC.Com}}  
\newcommand{\VCVerify}{\mathsf{VC.Verify}}

\newcommand{\VOTE}{\mathsf{VOTE}}

\newcommand{\CodedSymbols}{\mathbb{Y}_\mathrm{Symbols}}

\newcommand{\thisnodeindex}{i}         
\newcommand{\thisnodeindexstar}{\thisnodeindex^{\triangledown}}    
\newcommand{\thisnodeindexshuffle}{\thisnodeindex^{\triangledown}}     
\newcommand{\jshuffle}{j^{\triangledown}}                                                                                                                                       
                           
\newcommand{\send}{\textbf{send}}

\newcommand{\IDMVBA}{\mathrm{ID}}

\newcommand{\Output}{\textbf{output}}

\newcommand{\wait}{\textbf{wait}}                             
                        
\newcommand{\ECC}{\mathrm{ECC}}  
\newcommand{\EC}{\mathrm{EC}}  
\newcommand{\Alphabet}{\Bc}      
        
\newcommand{\MVBAOutputMsg}{\hat{\wv}}

\newcommand{\SYMBOL}{\text{``}\mathrm{SYMBOL}\text{''}}

\newcommand{\defaultvalue}{\bot}

\newcommand{\OEC}{\mathrm{OEC}}

\newcommand{\PROPOSE}{\mathsf{PROP}}

\newcommand{\Pass}{\textbf{pass}}

\newcommand{\COOL}{\mathrm{COOL}}

\newcommand{\DM}{\mathrm{DM}}

\newcommand{\networksizen}{n}                                           
\newcommand{\kencode}{k^\diamond}                                            
\newcommand{\networkfaultsizet}{t}

\newcommand{\ECCcodedsymbol}{y}

\newcommand{\ProtocolID}{\mathrm{ID}}

\newcommand{\HMDM}{\mathrm{HMDM}}          
           
\newcommand{\EONVOTE}{\mathsf{EPOCH}}

\newcommand{\nodeindexj}{j}            
\newcommand{\DoneDict}{\Dc}

\newcommand{\ABAOneSet}{\Ac_{\mathrm{ones}}}

\newcommand{\deliver}{\textbf{deliver}}                     
              
\newcommand{\alphabetsize}{q}

\newcommand{\APVA}{\mathrm{APVA}}                  
\newcommand{\Mc}{{\mathcal M}}

\newcommand{\contentop}{{content}\_{op}}     
\newcommand{\contentvp}{{content}\_{vp}} 
 
 \newcommand{\LTS}{\mathrm{LTS}}           
\newcommand{\Tproposal}{\Pc_\mathrm{Proposal}}  
\newcommand{\TMyproposal}{\Pc_\mathrm{MyProposal}}  
\newcommand{\TproposalPending}{\Pc_\mathrm{ProposalPending}}  
 
\newcommand{\heightprime}{h'}

\newcommand{\txdiamond}{{tx}^\diamond}      
\newcommand{\sigvpdiamond}{{sig}\_{vp}^\diamond}

\newcommand{\idtxdiamond}{{id}\_{tx}^\diamond}  
\newcommand{\sigdiamond}{sig^\diamond}

\newcommand{\WeightTxTwo}{\Wc_\mathrm{TxWeight2}}    
\newcommand{\WeightTxThree}{\Wc_\mathrm{TxWeight3}}

\newcommand{\idtxprime}{{id}\_{tx}'}   
\newcommand{\sigprime}{sig'}  
 
\newcommand{\idtxprimeprime}{{id}\_{tx}''}

\newcommand{\numproposedprime}{\NumProposed'} 
\newcommand{\numproposedprimeprime}{\NumProposed''} 
\newcommand{\txop}{{tx}\_{op}} 
\newcommand{\CONFLICT}{\mathsf{CONF}}  
   
\newcommand{\eonprime}{e'}

\newcommand{\contentinitial}{{content}\_{ini}}  
\newcommand{\OP}{\mathrm{OP}}

\newcommand{\vote}{vote}  
\newcommand{\contenthashinitia}{{content}\_{hash}\_{ini}}   
   
\newcommand{\EONCONFIRM}{\mathsf{EOK}}

\newcommand{\ltsgroupprime}{b'} 
\newcommand{\TS}{\mathrm{TS}}      
\newcommand{\ADKGCertificate}{\Cc_{\mathrm{adkg}}}                                                                           
   
\newcommand{\SEED}{\mathsf{SEED}}    
\newcommand{\SeedRecord}{\Sc_\mathrm{Seed}} 
\newcommand{\seed}{seed}

\newcommand{\SEEDVOTE}{\mathsf{SVOTE}}    
\newcommand{\SeedVoteRecord}{\Vc_{\mathrm{Seed}}}

\newcommand{\TnewIDSet}{\Tc_\mathrm{NewIDSet}} 
\newcommand{\TnewTxDictionary}{\Tc_\mathrm{NewTxDic}}  
\newcommand{\TnewselfQueue}{\Tc_{\mathrm{NewSelfIDQue}}} 
\newcommand{\contentoplist}{{content}\_{op}\_{list}}     
\newcommand{\contentoptuple}{content\_op\_ tuple}     
     
\newcommand{\opamount}{{op}\_{amount}}     
\newcommand{\feeamount}{{fee}}     
\newcommand{\outamount}{{out}\_{amount}}

 \newcommand{\txprime}{\mathrm{tx}'}     
       
\newcommand{\NumProposedSeed}{m_{\mathrm{seed}}}

\newcommand{\NumHeightMulticast}{h_\mathrm{dm}}

\newcommand{\IDProtocol}{\mathrm{ID}}

\newcommand{\VE}{\mathsf{VE}}       
\newcommand{\Dealer}{D}         
\newcommand{\secret}{s}         
\newcommand{\negligible}{\mathsf{negl}}          
\newcommand{\PPT}{\mathrm{PPT}}

\newcommand{\SPCSetup}{\mathsf{SHPC.Setup}}                              
\newcommand{\pp}{pp}            
\newcommand{\Group}{\mathbb{G}}

\newcommand{\SPCCommit}{\mathsf{SHPC.Commit}}                                                                                               
\newcommand{\CommitmentVector}{\vv}          
\newcommand{\CommitmentSymbol}{v}

\newcommand{\SPCOpen}{\mathsf{SHPC.Open}}

\newcommand{\PCDegCheck}{\mathsf{PC.DegCheck}}          
\newcommand{\SPCDegCheck}{\mathsf{SHPC.DegCheck}}                                                                                            
     
\newcommand{\PKIVerify}{\mathsf{PKI.Verify}}    
          
\newcommand{\SPCVerify}{\mathsf{SHPC.Verify}}                                                                                          
                                
\newcommand{\VEbEncProve}{\mathsf{VE.bEncProve}}                                                          
      
\newcommand{\CiphertextVector}{\cv}                                                    
\newcommand{\CiphertextSymbol}{c}                                                    
\newcommand{\NIZKProofVector}{\pi_{\mathrm{VE}}}                                                    
                                                    
\newcommand{\NIZK}{\mathrm{NIZK}}

\newcommand{\ECEnc}{\mathsf{ECEnc}}   
\newcommand{\ECDec}{\mathsf{ECDec}}    
\newcommand{\ECCEnc}{\mathsf{ECCEnc}}   
\newcommand{\ECCDec}{\mathsf{ECCDec}}  
\newcommand{\VEbVerify}{\mathsf{VE.bVerify}}   
\newcommand{\VEbDec}{\mathsf{VE.bDec}}

\newcommand{\partialsig}{\sigma}       
\newcommand{\finalsig}{\sigma}                                                         
\newcommand{\TSVerify}{\mathsf{TS.Verify}}

\newcommand{\sig}{sig}    
\newcommand{\Accept}{\mathsf{Accept}}
\newcommand{\CheckTx}{\mathsf{CheckTx}}   
\newcommand{\GetNewTxNoKL}{\mathsf{GetNewTx}}  
\newcommand{\GetNewTxNoKLCheck}{\mathsf{GetNewTxCheck}}                            
              
\newcommand{\content}{content}   
\newcommand{\contentprime}{content'} 
\newcommand{\contentdiamond}{content^{\diamond}} 
\newcommand{\TSCombine}{\mathsf{TS.Combine}}                                  
\newcommand{\SeedGeneration}{\mathsf{SeedGen}}        
               
\newcommand{\LTSGen}{\mathsf{LTS.DKG}}             
\newcommand{\ek}{ek}                               
\newcommand{\dk}{dk}   
\newcommand{\pkvector}{\underline{pk}}                                
\newcommand{\pk}{pk}                              
\newcommand{\sk}{sk}  
\newcommand{\true}{{true}}       
\newcommand{\false}{{false}}   
\newcommand{\pkl}{pkl}                               
\newcommand{\skl}{skl} 
\newcommand{\pklvector}{\underline{pkl}}                                 
\newcommand{\ltslayer}{\ell}    
\newcommand{\ltslayerMax}{L}    
     
\newcommand{\ltslayerTotalNodes}{u}     
\newcommand{\TSGen}{\mathsf{TS.DKG}}        
\newcommand{\Vote}{\mathsf{TS.Sign}}
\newcommand{\VoteLTS}{\mathsf{LTS.Sign}}  
\newcommand{\LTSVerify}{\mathsf{LTS.Verify}}                                  
\newcommand{\LTSCombine}{\mathsf{LTS.Combine}}                                   
\newcommand{\ltsGroupSet}{\Gc}

\newcommand{\skshare}{s}

\newcommand{\FieldZ}{\mathbb{Z}}         
\newcommand{\randomgenerator}{\mathsf{g}}             
    
\newcommand{\ltsgroup}{b}                                              
            
\newcommand{\TSthreshold}{k}              
\newcommand{\skshareLTS}{\tilde{s}}

\newcommand{\PolynomialFunction}{\phi}              
       
\newcommand{\PolynomialFunctionNew}{\varphi}              
   
\newcommand{\LTSPolynomialFunction}{\psi}

\newcommand{\PolyDegree}{d}

\newcommand{\skinput}{s}                 
\newcommand{\CommitPolySecret}{r}

\newcommand{\LTSCommitmentVector}{\tilde{\vv}}           
\newcommand{\LTSCommitmentSymbol}{\tilde{v}}    
\newcommand{\IndexInParentGroup}{\omega}                                                                                            
\newcommand{\ParentIndex}{\beta}                                                                                              
\newcommand{\SHARE}{\mathsf{SHARE}}                    
\newcommand{\checkNew}{\textbf{check}}

\newcommand{\PKISign}{\mathsf{PKI.Sign}}     
\newcommand{\digitalsig}{\sigma}        
\newcommand{\ACK}{\mathsf{ACK}}                     
\newcommand{\AVSSValidSet}{\Cc_\mathrm{ack}}      
\newcommand{\AVSSValidSetIndex}{\Ic_\mathrm{ack}}      
\newcommand{\AVSSMissSetIndex}{\Ic_\mathrm{miss}}            
\newcommand{\AVSSVDelay}{\Delta_\mathrm{delay}}       
\newcommand{\DelayPara}{\Delta_\mathrm{delay}}             
\newcommand{\RBCSEND}{\mathsf{RBC}}                      
\newcommand{\broadcast}{\textbf{broadcast}}                   
\newcommand{\return}{\textbf{return}}                    
     
\newcommand{\APBVAinputNew}{\av}          
\newcommand{\APBVAoutputNew}{\APBVAinputNew^{\star}}  
\newcommand{\RBC}{\mathsf{RBC}}   
\newcommand{\AVSSindex}{\tau}

\newcommand{\LagrangeCoefficient }{\gamma}

\newcommand{\missing}{\bot_o}                       
\newcommand{\erase}{\textbf{erase}}

\newcommand{\OciorADKG}{\mathsf{OciorADKG}}

\newcommand{\Adversary}{\Ac}        
\newcommand{\ADKG}{\mathsf{ADKG}}

\newcommand{\OciorBLS}{\mathsf{OciorBLSts}}    
                      
\newcommand{\TSSetup}{\mathsf{TS.Setup}}  
\newcommand{\LTSSetup}{\mathsf{LTS.Setup}}              
\newcommand{\PrimeOrder}{p}        
\newcommand{\FiniteFieldSize}{p}

\newcommand{\VSS}{\mathsf{VSS}}

\newcommand{\HashCommit}{h}                                
\newcommand{\HashCommitVector}{\boldsymbol{\HashCommit}}                                
\newcommand{\PolynomialFunctionAddWitness}{\check{\phi}}              
\newcommand{\skshareAddWitness}{\check{\skshare}}              
\newcommand{\Hash}{\mathsf{H}}     
\newcommand{\WitnessVectorR}{\rv}                                                        
     
\newcommand{\RandomPolyFunction}{\theta}              
\newcommand{\SPCWitnessReconstruct}{\mathsf{SHPC.WitnessReconstruct}}                                                                                          
\newcommand{\SPCPolyReconstruct}{\mathsf{SHPC.Reconstruct}}                                                                                          
\newcommand{\IndexSet}{\Tc}               
\newcommand{\CommitmentVectorOutput}{\vv^{\star}}          
               
\newcommand{\SHVSS}{\mathsf{SHVSS}}                             
\newcommand{\ASHVSS}{\mathsf{ASHVSS}}                             
\newcommand{\OciorASHVSS}{\mathsf{OciorASHVSS}}          
\newcommand{\SHPC}{\mathsf{SHPC}}                              
\newcommand{\OciorSHPC}{\mathsf{OciorSHPC}}                                 
\newcommand{\SPCWitnessCommit}{\mathsf{SHPC.WitnessCommit}}                                                                                           
\newcommand{\SPCWitnessOpen}{\mathsf{SHPC.WitnessOpen}}                                                                                           
\newcommand{\NumPoints}{M}                                                                                           
             
\newcommand{\OciorDKGtd}{\mathsf{OciorDKGtd}}          
\newcommand{\BFT}{\mathrm{BFT}}           
\newcommand{\Blocksize}{B}

\newcommand{\TX}{\mathsf{TX}}                                
                                 
\newcommand{\TXDONE}{\mathsf{APS}}
\newcommand{\Add}{.\mathrm{add}}                                                                                                                                                                                                                                                                                                                                                                               
\newcommand{\NewTXProcess}{\mathsf{NewTxProcess}} 
            
\newcommand{\RPC}{\mathrm{RPC}}             
\newcommand{\TnewIDSetOtherProposedLastEpoch}{\Tc_\mathrm{NewIDSetPOLE}}  
\newcommand{\TnewIDSetOtherProposed}{\Tc_\mathrm{NewIDSetPO}}  
\newcommand{\NumProposedIntevalRandomTxSelf}{\NumProposed_\mathrm{txself}}      
\newcommand{\NumProposedIntevalRandomTxFromOtherProp}{\NumProposed_\mathrm{txpo}}      
\newcommand{\TproposedIDSet}{\Tc_\mathrm{ProposedIDSet}} 
\newcommand{\BinaryIndicator}{indicator}  
\newcommand{\Remove}{.\mathrm{remove}}                                                                                                                                                                                                                                                                                                                                                                               
\newcommand{\Getrandom}{.\mathrm{get}\_\mathrm{random}}                                                                                                                                                                                                                                                                                                                                                                               
\newcommand{\Popleft}{.\mathrm{popleft}}                                                                                                                                                                                                                                                                                                                                                                               
\newcommand{\None}{None}                                                                                                                                                                                                                                                                                                                                                                                
     
\newcommand{\PoP}{\mathrm{PoP}}             
\newcommand{\proofofPreviousProposedTx}{proof}                                                                                                                                                                                                                                                                                                                                                                                 
                                                                                                                                                                                                                                                                                                                                                                                 
\newcommand{\myproofofPreviousProposedTx}{myproof}                                                                                                                                                                                                                                                                                                                                                                                 
\newcommand{\ProofCheck}{\mathsf{ProofCheck}} 
\newcommand{\proofcheckindicator}{{pcheck}\_{indicator}}  
\newcommand{\NumTxProposedLock}{\Mc_\mathrm{ProposedLock}}   
\newcommand{\ProofTypeConflict}{\mathsf{C}}

\newcommand{\jdiamond}{j^{\diamond}}  
\newcommand{\ConflictTxCheck}{\mathsf{ConflictTxCheck}}  
\newcommand{\conflictcheckindicator}{{ccheck}\_{indicator}}  
\newcommand{\TxConflictingDic}{\Tc_\mathrm{ConfTxDic}} 
\newcommand{\OPSetTemporary}{\Oc_\mathrm{opSetTemporary}}  
\newcommand{\acceptcheckindicator}{{acheck}\_{indicator}}

\newcommand{\allacceptcheckindicator}{{all}\_{indicator}}    
\newcommand{\ProposalProcess}{\mathsf{ProposalProcess}}  
\newcommand{\ProposalProcessHeightLock}{\mathsf{ProposalProcessHL}}  
\newcommand{\ProposalProcessHeightLockMLock}{\mathsf{ProposalProcessHLML}}  
\newcommand{\NNeedProposalVP}{\Nc_\mathrm{NeedPropVP}}  
    
\newcommand{\NNeedProposalHeightLock}{\Nc_\mathrm{NeedPropHL}}    
\newcommand{\NNeedProposalHeightLockMLock}{\Nc_\mathrm{NeedPropML}}     
\newcommand{\HeightLockUpdate}{\mathsf{HeightLockUpdate}}  
\newcommand{\inputtuple}{{input}\_{tuple}}   
\newcommand{\ProofCheckNumTxProposedLockUpdate}{\mathsf{ProofCheckNPLUpdate}}   
\newcommand{\ProofCheckNumTxProposedLockUpdateSpecicial}{\mathsf{ProofCheckNPLUpdateSpecicial}}  
\newcommand{\ProofCheckNumTxProposedLockUpdateNormal}{\mathsf{ProofCheckNPLUpdateNormal}}  
\newcommand{\PProof}{\Pc_\mathrm{proof}}

\newcommand{\PtotocolSigma}{\Sigma}

\newcommand{\belinearpairing}{\mathsf{e}}     
\newcommand{\HashZ}{\mathsf{H}_{\mathsf{z}}}     
\newcommand{\LagrangeCoefficientIntAtZero}{\lambda}          
\newcommand{\indicator}{{indicator}} 
     
\newcommand{\SigAggregationLTS}{\mathsf{SigAggregationLTS}}     
\newcommand{\numredundantnodes}{\tau} 
\newcommand{\indexinset}{\alpha}  
\newcommand{\PartialSigltsGroupSet}{\Pc}   
\newcommand{\PartialSigLayer}{\Lc}    
\newcommand{\PartialSigTree}{\Tc_{\mathrm{tree}}}    
\newcommand{\PartialSigCollection}{\Cc_{\ltslayerMax}}    
\newcommand{\PartialSigLayerSubset}{\PartialSigLayer^{\diamond}}    
\newcommand{\ResilianceCondition}{n\geq 3t+1}   
\newcommand{\tResilianceCondition}{t < \frac{n}{3}}

\newcommand{\ROM}{\mathrm{ROM}}    
\newcommand{\ACMA}{\mathrm{ACMA}}    
\newcommand{\CorruptSet}{\Cc_{\mathrm{corrupt}}}       
\newcommand{\PartialSignQuerySet}{\Qc_{\mathrm{psQuery}}}       
\newcommand{\RandomOracleQuerySet}{\Qc_{\mathrm{roQuery}}}           
\newcommand{\APS}{\mathrm{APS}}      
\newcommand{\RandomOracleH}{\mathcal{O}_{\Hash}}        
\newcommand{\LTSROBACMA}{\text{LTS-ROB-ACMA}}    
\newcommand{\LTSUNFACMA}{\text{LTS-UNF-ACMA}}      
\newcommand{\TSUNFACMA}{\text{TS-UNF-ACMA}}      
\newcommand{\TSROBACMA}{\text{TS-ROB-ACMA}}       
\newcommand{\AttestedProofofSeal}{\text{Attested Proof of Seal}}        
\newcommand{\ID}{\mathrm{ID}}         
\newcommand{\APSISent}{\Ac_{\mathrm{APSSent}}}     
\newcommand{\APSI}{\mathsf{APSI}}       
\newcommand{\TatalNumRPCnodes}{\bar{n}}

\newcommand{\AWthrehold}{\eta}

\newcommand{\SenderSet}{\Sc}    
\newcommand{\ReceiverSet}{\Rc}    
    
\newcommand{\VC}{\mathrm{VC}}   
\newcommand{\myproofProposedTx}{myprooftx}                                                                                                                                                                                                                                                                                                                                                                                  
\newcommand{\proofProposedTx}{prooftx}                                                                                                                                                                                                                                                                                                                                                                                  
\newcommand{\checkprooftxindicator}{{checkp}\_{indicator}}  
\newcommand{\CheckProofTx}{\mathsf{CheckProofTx}}      
\newcommand{\TAWOneIDSetOtherProposed}{\Tc_\mathrm{AW1IDSetPO}}  
\newcommand{\TAWOnePOProposedIDSet}{\Tc_\mathrm{AW1ProposedIDSet}}     
\newcommand{\TAWOneIDSetOtherProposedLastEpoch}{\Tc_\mathrm{AW1IDSetPOLE}}   
\newcommand{\EpochOutageThreshold}{e_\mathrm{out}}                                                                                                                                                                                                                                                                                                                                                                                   
\newcommand{\OciorHMDMHash}{\mathsf{OciorHMDMh}}      
\newcommand{\OciorHMDMIT}{\mathsf{OciorHMDMit}}      
\newcommand{\HMDMMsg}{\wv}

\newcommand{\NumTxProposedVoted}{\Mc_\mathrm{ProposedVoted}}

\newcommand{\numproposedvartriangle}{\NumProposed^{\vartriangle}}    
\newcommand{\eonvartriangle}{\eon^{\vartriangle}}   
\newcommand{\txvartriangle}{\tx^{\vartriangle}}

\newcommand{\numproposedcirc}{\NumProposed^{\circ}}    
\newcommand{\eonccirc}{\eon^{\circ}}   
\newcommand{\myLastProposal}{myproposal}                                                                                                                                                                                                                                                                                                                                                                                    
   
\newcommand{\inputtupleprimeprime}{{input}\_{tuple}''}   
\newcommand{\aux}{aux}                                                                                                                                                                                                                                                                                                                                                                                    
\newcommand{\ppinputtuple}{{pp}\_{input}\_{tuple}}    
\newcommand{\pphlinputtuple}{{in}\_{tuple}}    
\newcommand{\pkpki}{pk^{\diamond}}                              
\newcommand{\skpki}{sk^{\diamond}}
\newcommand{\PKI}{\mathrm{PKI}}    
\newcommand{\WitnessSHARE}{\mathsf{WITNESS}}                    
\newcommand{\WitnessShareValidSet}{\Wc_{wit}}                                                                                                                                                                                                                                                                                                                                                                                    
\newcommand{\MarkSet}{\Cc_{\mathrm{done}}}                                                                                                                                                                                                                                                                                                                                                                                    
\newcommand{\BOneSet}{\Bc_{\mathrm{ones}}}                                                                                                                                                                                                                                                                                                                                                                                    
\newcommand{\indexi}{i}

\algdef{SE}[SUBALG]{Indent}{EndIndent}{}{\algorithmicend\ }
\algtext*{Indent}
\algtext*{EndIndent}

\begin{document}
\sloppy
\title{Ocior: Ultra-Fast Asynchronous Leaderless Consensus with Two-Round Finality, Linear Overhead, and Adaptive Security}

\author{Jinyuan Chen \\ jinyuan@ocior.com 
}

\maketitle
\pagestyle{headings}

\begin{abstract} 
In this work, we propose \emph{Ocior}, a practical \emph{asynchronous} Byzantine fault-tolerant ($\BFT$) consensus protocol that achieves the  optimal performance in resilience, communication, computation, and round complexity. 
Unlike traditional $\BFT$ consensus protocols, Ocior processes incoming transactions individually and concurrently using parallel instances of consensus.   
While leader-based consensus protocols rely on a designated leader to propose transactions, Ocior is a \emph{leaderless} consensus protocol that guarantees \emph{stable liveness}.  
A protocol is said to satisfy the \emph{stable liveness} property if it ensures the continuous processing of incoming transactions, even in the presence of an \emph{adaptive} adversary who can dynamically choose which nodes to corrupt, provided that the total number of corrupted nodes does not exceed $t$, where $\ResilianceCondition$  is the total number of consensus nodes. 
Ocior achieves:
\begin{itemize}
    \item \textbf{\emph{Optimal resilience}}: Ocior tolerates up to $t$ faulty nodes controlled by an \emph{adaptive} adversary, for $\ResilianceCondition$.
    \item \textbf{\emph{Optimal communication complexity}}: The  total expected communication per transaction is $O(n)$. 
        \item \textbf{\emph{Optimal (or near-optimal) computation complexity}}: The total  computation per transaction is $O(n)$ in the  best case, or $O(n \log^2 n)$ in the worst case. 
        \item \textbf{\emph{Optimal round complexity}}:    A legitimate \emph{two-party} transaction can be finalized with a \emph{good-case latency} of \emph{two} asynchronous rounds, for any $\ResilianceCondition$, where each round corresponds to a single \emph{one-way} communication. The \emph{good case} in terms of latency  refers to the scenario where the transaction is proposed by any (not necessarily designated) honest node. A \emph{two-party transaction} involves the transfer of digital assets from one user (or group of users) to one or more recipients.  
\end{itemize} 
To support efficient consensus, we introduce a novel non-interactive threshold signature (\emph{$\TS$}) scheme called $\OciorBLS$. It offers fast signature aggregation, and is adaptively secure under the algebraic group model and the hardness assumption of the one-more discrete logarithm problem.  
$\OciorBLS$ achieves a computation complexity of signature aggregation of only $O(n)$ in the good cases.    
Moreover, $\OciorBLS$ supports the property of \emph{Instantaneous $\TS$ Aggregation}. A $\TS$ scheme guarantees this property if it can aggregate partial signatures immediately, without waiting for all $\TSthreshold$ signatures, where $\TSthreshold$ is the threshold required to compute the final signature. This enables real-time aggregation of partial signatures as they arrive, reducing waiting time and improving responsiveness.  
Additionally, $\OciorBLS$ supports weighted signing power or voting, where nodes may possess different signing weights, allowing for more flexible and expressive consensus policies.

\end{abstract}

\section{Introduction}

Distributed Byzantine fault-tolerant ($\BFT$) consensus is a fundamental building block of blockchains and distributed systems. 
Yet, according to CoinMarketCap \cite{CoinMarketCap:2025}, the top thirty blockchain systems by market capitalization rely on consensus protocols designed under synchronous or partially synchronous assumptions. 
When the network becomes fully asynchronous, such protocols may no longer guarantee safety or liveness, creating critical risks for blockchain infrastructure. 
This motivates the design of a \emph{practical asynchronous} consensus protocol.  

\vspace{5pt}

For many Web~3.0 applications, particularly \emph{latency-sensitive services}, the traditional guarantees of safety and liveness are insufficient.
Applications such as decentralized finance (DeFi) trading, cross-chain transfers, non-fungible token (NFT) marketplaces, supply-chain tracking, gaming platforms, and real-time blockchain systems demand \emph{fast transaction finality}, a cryptographically verifiable confirmation that a transaction will be accepted irrevocably by all consensus nodes. 
High latency degrades user experience, increases risks in financial settings, and limits scalability. 
Thus, the goal is not only to design an \emph{asynchronous} consensus protocol but also a \emph{fast} asynchronous consensus protocol.

\vspace{5pt}

Despite significant progress in $\BFT$ consensus, many existing protocols remain \emph{leader-based}, relying on a designated leader to propose transactions (e.g., \cite{solana:18, YMRGA:19}). 
This design introduces systemic vulnerabilities: if the leader is faulty or under distributed denial of service (DDoS) attacks, the entire system may stall or suffer degraded performance. 
In adversarial environments with \emph{adaptive adversaries}, capable of dynamically corrupting or targeting nodes, these vulnerabilities become even more severe. 
Leader-based protocols   fail to provide \emph{stable liveness}, which we define as the ability to continuously process incoming transactions despite an adaptive adversary corrupting up to $t$ nodes  out of $\ResilianceCondition$ participants.  

\vspace{5pt}

This raises a central research question:  
\[
  \text{Can we design a \emph{fast}, \emph{adaptively secure}, \emph{asynchronous} $\BFT$ consensus protocol with \emph{stable liveness}?}
\]
We address this challenge by introducing $\Ocior$, a fast, leaderless, adaptively secure, asynchronous $\BFT$ consensus protocol that guarantees stable liveness.
Unlike traditional batch-based designs, $\Ocior$ processes transactions \emph{individually} and \emph{concurrently}, executing parallel consensus instances to maximize throughput and responsiveness.  

\vspace{5pt}

A key performance metric is \emph{transaction latency} (also called \emph{finality time} or \emph{confirmation latency}), which measures the time from transaction submission until the client receives a cryptographic acknowledgment of acceptance. 
To support lightweight and trustless verification, we introduce the notion of an \emph{$\AttestedProofofSeal$} ($\APS$), a \emph{short} cryptographic proof that a transaction has been sealed and is guaranteed to be accepted irrevocably by all consensus nodes. 
This enables external parties, including \emph{light clients}, to efficiently verify transaction acceptance without running a full node.  

\vspace{5pt}

$\Ocior$ achieves the following asymptotically optimal guarantees:
\begin{itemize}
    \item \emph{Optimal resilience}: It tolerates up to $t$ Byzantine nodes controlled by an \emph{adaptive} adversary, for $\ResilianceCondition$.
    \item \emph{Optimal communication complexity}: The total expected communication per transaction is $O(n)$. 
    \item \emph{Optimal (or near-optimal) computation complexity}:  The total computation per transaction is $O(n)$ in the best case, or $O(n \log^2 n)$ in the worst case, measured in cryptographic operations (signing, verification, hashing, and arithmetic on signature-sized values).
    \item \emph{Optimal round complexity}:     
        A legitimate \emph{two-party} transaction can be finalized in \emph{two} asynchronous rounds (good-case latency) for any $\ResilianceCondition$, where each round is a single \emph{one-way} communication. 
        Here, the \emph{good case} refers to transactions proposed by any honest node. 
        A \emph{two-party transaction} involves the transfer of digital assets from one user (or group) to one or more recipients. 
        All other transactions can be finalized in \emph{four} asynchronous rounds, with linear communication and computation overhead. 
        Following common conventions~\cite{CL:99, YMRGA:19, SKN:25, MXCSS:16, LDZ:2020}, we do not count client-to-node communication in round complexity.   
\end{itemize}

 {\renewcommand{\arraystretch}{1.5}
\begin{table}
\footnotesize  
\begin{center}
\caption{
Comparison between the proposed Ocior protocol and some other consensus protocols.  $\Delta$ is a constant but is much larger than $4$. $\Blocksize$ denotes the number of transactions in a block.  
In this comparison, we consider the optimal resilience setting of $\ResilianceCondition$ for all protocols.   
In this comparison, we just focus on the two-party transactions. 
 } 
 \label{tb:protocols}
\begin{tabular}{||c||c|c|c|c|c|c|c|c|}
\hline
Protocols & Network   &    Security           & Rounds   &   $ \! \! \! \!\!$  Total Communication$ \! \!\!$     & $ \! \! \! \!\!\!\!$   Total  Computation$ \! \! \! \!\!$   &    $\! \! \! \! \! $ Instantaneous  $  \! \! \! \! \! $ &      $ \! \! \! $Stable$\! \! \!$ &  $ \! \! \! \!\!\!\!$  Short $ \! \! \! \!\!\!\!$ \\ 
  &     &             &  $ \! \! \! \!\!\!\!$ (Finality) $\! \!\! \! \! \!\!$ & $ \! \! \! \!\!$ Per Transaction $ \! \! \! \!\!$    & Per Transaction  &   $ \! \! \! \!\!$  $\TS$  Aggr.$  \!\! \! \!\!$      &    $  \! \! \!\!\!$   Liveness  $  \!\! \! \!\!$  &    $ \! \! \!\! \! $  $\APS$  $\!\! \! \!$  \\ 
    &     &             &  $ \! \! \! \!\!$ (Good Case) $ \!\! \! \! \!\!$ &      & (Good Case)   &     &       &   \\ 
\hline 
$ \! \! \! \!\!$ Ethereum 2.0 \cite{ethereum2}$ \! \! \!\!$ &  $\! \! \!\!$ Partially Syn. $ \! \! \! \!\!$    &    Static         &  $>\Delta$  &    $O(n)$  &  -  &  -  &    $\times$  &        $\checkmark$ \\
\hline
Solana  \cite{solana:18}  &  $\! \! \!\!$ Partially Syn. $ \! \! \! \!\!$   &     Static       & $>\Delta$     &  $O(n)$  &  -  &  -  &   $\times$ &        $\checkmark$ \\
\hline
PBFT  \cite{CL:99} &  $\! \! \!\!$ Partially Syn. $ \! \! \! \!\!$    &     Static       & 3     &  $O(n^2)$     &  $ \! \! \! \!\!$  $O(\max\{n^2/B, n\})$ $ \! \! \! \!\!$  &  -  &   $\times$  &        $\times$  \\
\hline
HotStuff  \cite{YMRGA:19} &  $\! \! \!\!$ Partially Syn. $ \! \! \! \!\!$     &     Static       & 6 or 4  &  $O(n)$  &  $ \! \! \! \!\!$  $O(\max\{n^2/B, n\})$ $ \! \! \! \!\!$  &  $\times$  &   $\times$  &        $\checkmark$ \\
\hline
Hydrangea  \cite{SKN:25} &  $\! \! \!\!$ Partially Syn. $ \! \! \! \!\!$      &     Static       &  3  or 2  &  $O(\max\{n^2/B, n)$  &  $ \! \! \! \!\!$  $O(\max\{n^2/B, n\})$ $ \! \! \! \!\!$  &   -   &   $\times$  &        $\times$  \\  
\hline 
\hline
$ \! \! \! \!\!$ Avalanche  \cite{RYS+:20} $ \!\! \! \! \!\!$ &  $ \! \! \!\!$ Synchronous$ \! \! \!\!$    &    $ \! \! \!$ Adaptive $ \! \! \!$      & $O(\log n)$    &  $O(n  \cdot \log n)$  & $ \! \! \! \!\!$   $O( n \cdot \log n)$ $ \! \! \! \!\!$   &  -   &    $\checkmark$ &        $\times$ \\   
\hline
\hline
$ \! \! \! \!\!$ HoneyBadger  \cite{MXCSS:16,LDZ:2020} $ \!\! \! \! \!\!$ &  $ \! \! \!\!$Asynchronous$ \! \! \!\!$    &     $ \! \! \!$Adaptive$ \! \! \!$       & $>\Delta$    &  $O(n)$  & $ \! \! \! \!\!$   $O(\max\{n^2/B, n\})$ $ \! \! \! \!\!$   &  $\times$   &    $\checkmark$ &        $\times$ \\
\hline
 {\color{blue}Ocior} &  {\color{blue} $ \!\! \!\!$Asynchronous$ \! \! \! \!$}  &   {\color{blue} $ \! \! \!$Adaptive$ \! \! \!$}      &  {\color{blue}2}   & {\color{blue}  $O(n)$} &   {\color{blue}   $O(n)$}&   {\color{blue}   $\checkmark$}  &    {\color{blue} $\checkmark$} &    {\color{blue} $\checkmark$} \\
\hline 
\end{tabular}
\end{center}
Notes: 1. HoneyBadger relies on threshold cryptography that is \emph{statically} secure, and is therefore secure only against static adversaries~\cite{LDZ:2020}. The work in~\cite{LDZ:2020} improves HoneyBadger to achieve adaptive security by using adaptively secure threshold cryptography.   
2. Following common conventions~\cite{CL:99,YMRGA:19,SKN:25,MXCSS:16,LDZ:2020}, when measuring round complexity, we do not include the communication cost between clients and consensus nodes.  
3.  Hydrangea provides an $\APS$ called a block certificate, but it is not short.   The size of its $\APS$ is $O(n)$. 
4. Hydrangea achieves a good-case latency of three rounds in settings where the number of Byzantine nodes is greater than $\lfloor \frac{\numredundantnodes}{2} \rfloor$ for a total of $n=3t +\numredundantnodes+1$ nodes, where $\numredundantnodes\geq 0$. In particular,  when $\numredundantnodes=0$,  Hydrangea achieves a good-case latency of three rounds when  the number of Byzantine nodes is greater than $0$.    When the number of Byzantine nodes is at most   $\lfloor \frac{\tau}{2} \rfloor$, it achieves a good-case latency of two rounds.   
\end{table}
}

To support efficient consensus and fast attestation, $\Ocior$ introduces a novel \emph{non-interactive threshold signature} ($\TS$) scheme called $\OciorBLS$. 
Unlike traditional $\TS$ schemes that require $O(n^2)$ effort for signature aggregation (or $O(n \log^2 n)$ with optimizations \cite{TCZAPGD:20}), $\OciorBLS$ reduces aggregation to $O(n)$ in the best case while ensuring \emph{adaptive security}. 
Its design is based on a new idea of non-interactive \emph{Layered Threshold Signatures} ($\LTS$), where the final signature is composed through $\ltslayerMax$ layers of partial signatures, for a predefined parameter $\ltslayerMax$. 
At each layer~$\ltslayer$, a partial signature is generated from a set of partial signatures at layer~$\ltslayer+1$, for $\ltslayer \in \{1,2,\dots,\ltslayerMax-1\}$.   
In the signing phase of $\LTS$, nodes sign messages independently, producing partial signatures that are treated as the signatures of Layer~$\ltslayerMax$. 
An example is illustrated in Fig.~\ref{fig:OciorLTSexample} (Section~\ref{sec:OciorBLS}).  

\vspace{5pt}

The $\LTS$ construction enables \emph{fast and parallelizable} aggregation. 
In particular, $\OciorBLS$ supports \emph{Instantaneous TS Aggregation}, where partial signatures can be combined immediately upon arrival, without waiting for all $\TSthreshold$ shares. 
This reduces waiting time and lowers transaction latency.  
Additionally, $\OciorBLS$ supports weighted signing power (or voting), where nodes may possess different signing weights, allowing for more flexible and expressive consensus policies. 
 
\vspace{5pt}
 
 Finally, Table~\ref{tb:protocols} compares $\Ocior$ with other consensus protocols. 
Along this line of research, some works attempt to reduce round complexity, but only under synchronous or partially synchronous settings (e.g., \cite{YMRGA:19, KADCW:10, GAGMPRSTT:19, SSV:25, SKN:25}). 
In particular, the authors of~\cite{SKN:25} proved that no consensus protocol can achieve a \emph{two-round latency} when the number of Byzantine nodes exceeds 
\(\left\lfloor \frac{\numredundantnodes + 2}{2} \right\rfloor\), for a total of 
\(n = 3t + \numredundantnodes + 1\) nodes, where \(\numredundantnodes \geq 0\).  
In contrast, we show that $\Ocior$ breaks this barrier by achieving \emph{two-round latency} even in the \emph{asynchronous} setting, when the number of Byzantine nodes is up to 
\(t = \left\lfloor \frac{n - 1 - \numredundantnodes}{3} \right\rfloor\), for any \(\numredundantnodes \geq 0\).

 \newpage

 \section{System Model \label{sec:system}}

We consider an \emph{asynchronous} Byzantine fault-tolerant  consensus problem over a network consisting of $n$ consensus nodes (servers), where up to $t$ of them may be corrupted by an \emph{adaptive} adversary, assuming the optimal resilience condition $\ResilianceCondition$.   
A key challenge in this BFT consensus problem lies in achieving agreement despite the presence of corrupted (dishonest) nodes that may arbitrarily deviate from the designed protocol. 
In this setting, the consensus nodes aim to reach agreement on transactions issued by clients. These transactions may represent digital asset transfers in cryptocurrencies or the execution of smart contracts. The system model and relevant definitions are presented below. 

\vspace{3pt}

\emph{Adaptive Adversary:} We consider an adaptive adversary that can dynamically choose which nodes to corrupt, subject to the constraint that the total number of corrupted nodes does not exceed $t$. 

\vspace{3pt}

\emph{Asynchronous Network:} We assume an asynchronous network in which any two consensus nodes are connected by reliable and authenticated point-to-point channels. Messages sent between honest nodes may experience \emph{arbitrary delays} but are guaranteed to be eventually delivered. 
 
 \vspace{3pt}

\emph{Network Structure:} 
Each consensus node is fully interconnected with all other consensus nodes and also maintains connections with Remote Procedure Call ($\RPC$) nodes, following an \emph{asymmetric outbound-inbound connection} model, in which clients may also act as $\RPC$ nodes, as shown in Fig.~\ref{fig:OciorNet}. 
Specifically, each consensus node is capable of sending low-latency outbound messages regarding the finality confirmations for transactions it processes to all $\RPC$ nodes.  To enhance resilience against distributed denial of service  attacks, each consensus node limits its inbound connections to a dynamically sampled and periodically refreshed subset of $\RPC$ nodes for receiving new transaction submissions. 
Lightweight clients submit transactions and verify their finality through interactions with $\RPC$ nodes. $\RPC$ nodes propagate pending, legitimate transactions to their connected $\RPC$ nodes and to consensus nodes with inbound connections. Upon receiving a finality confirmation, an $\RPC$ node forwards it to its connected $\RPC$ nodes and relevant clients, often leveraging real-time subscription mechanisms such as WebSockets~\cite{SolanaWebsocket}.

 \vspace{3pt}

We classify transactions into two types: Two-party (Type~I)   and third-party (Type~II) transactions, defined below.

\begin{definition}[{\bf Two-Party (Type~I) Transactions}]    
A \emph{two-party transaction} involves the transfer of digital assets from one user (or a group of users) to one or more recipients. For a Type~I transaction, both the sender and the recipients must be able to verify the finality of the transaction. However, no third party is required to verify its finality.  
Multi-signature transactions, which require multiple signatures from different parties to authorize a single transaction, and multi-recipient transactions also fall under this category.  
\end{definition}

 \begin{figure} 
\centering
\includegraphics[width=6cm]{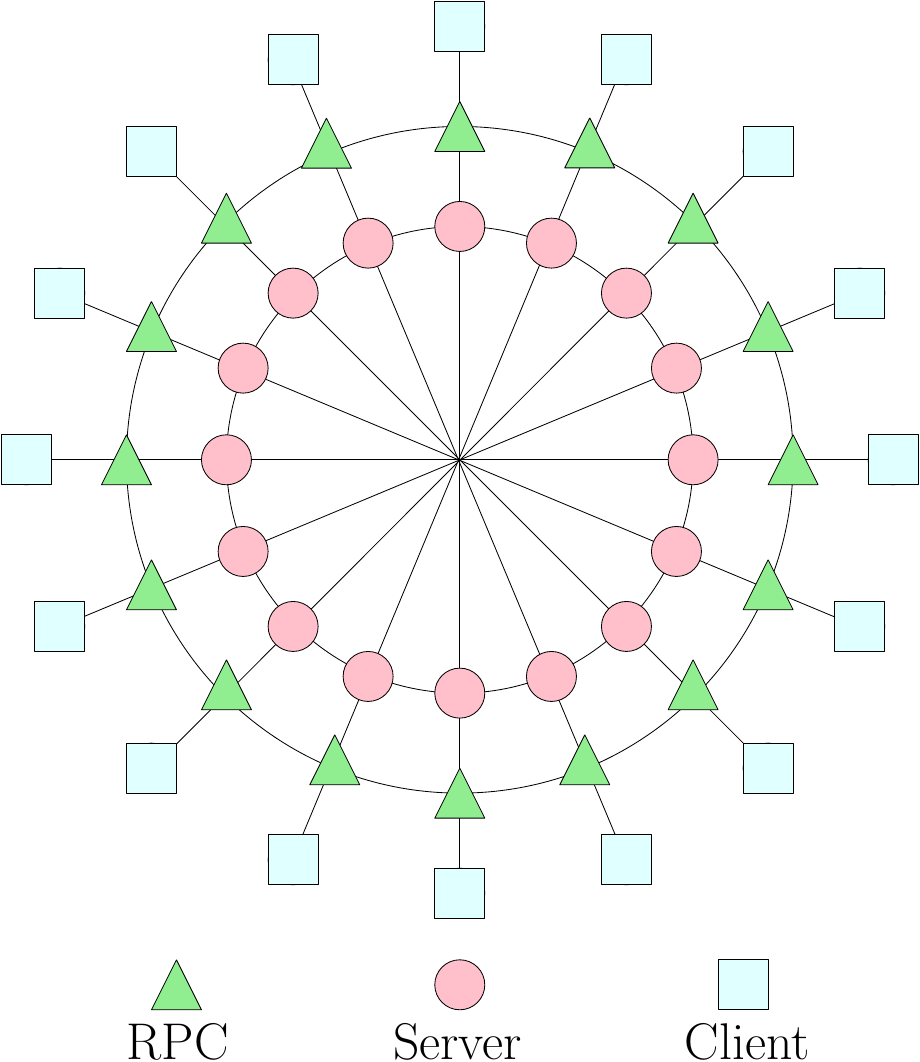}
\caption{The network architecture of Ocior, consisting of consensus nodes (servers), remote procedure call ($\RPC$) nodes and clients.}  
\label{fig:OciorNet}
\end{figure}

\begin{definition}[{\bf Third-Party (Type~II) Transactions}]    
The key distinction between Type~I and Type~II transactions is the involvement of a third party. In a Type~II transaction, it is necessary for any third party to be able to verify the finality of the transaction. 
\end{definition}

\begin{definition}[{\bf Transaction Latency and $\APS$}]    
\emph{Transaction latency}, also referred to as \emph{transaction finality time} or \emph{confirmation latency}, measures the time elapsed from the submission of a transaction to the moment when the client receives a concise cryptographic acknowledgment of its acceptance. Specifically, this acknowledgment takes the form of a short \emph{$\AttestedProofofSeal$} ($\APS$), which attests that the transaction has been sealed and is guaranteed to be irrevocably accepted by all consensus nodes, even if it has not yet appeared on the ledger. 
This proof allows any external parties, including resource-constrained \emph{light clients}, to efficiently verify the transaction's acceptance status without running a full node. 
\end{definition}

\begin{definition}[{\bf Parent Transactions}]      \label{def:parenttx}
A two-party transaction involving the transfer of digital assets from Client~$A$ to Client~$B$ is denoted by $\Tx_{A,B}$, where $A$ and $B$ represent the respective wallet addresses of Clients~$A$ and $B$.  
Once $\Tx_{A,B}$ is finalized (i.e., has received a valid $\APS$), Client~$B$ becomes the official owner of the transferred digital assets.  
If Client~$B$ subsequently initiates a new transaction $\Tx_{B,C}$ to transfer the assets received in $\Tx_{A,B}$ to Client~$C$, then $\Tx_{A,B}$ must be cited as a \emph{parent} transaction. In this case, the new transaction $\Tx_{B,C}$ is referred to as a \emph{child} transaction of $\Tx_{A,B}$. 
\end{definition}

 \vspace{1pt}    

\begin{definition}[{\bf Proof of Parents ($\PoP$)}]    
When Client~$B$ initiates a new transaction $\Tx_{B,C}$ to transfer the assets received in $\Tx_{A,B}$ to Client~$C$, the transaction $\Tx_{A,B}$ is cited as a \emph{parent} transaction.  
Our proposed consensus mechanism includes a process for proving that $\Tx_{A,B}$ is a valid parent of $\Tx_{B,C}$.  
We refer to this consensus approach as  \emph{proof of parents}. 
\end{definition}

  \vspace{1pt}

\begin{definition}[{\bf Conflicting Transactions}]    
If Client~$B$ creates two different transactions, $\Tx_{B,C}$ and $\Tx_{B,C'}$, attempting to doubly spend the asset inherited from $\Tx_{A,B}$, then $\Tx_{B,C}$ and $\Tx_{B,C'}$ are considered \emph{conflicting transactions} (i.e., an instance of \emph{double spending}).  
Any descendant of a conflicting transaction is also considered a conflicting transaction.
\end{definition}

  \vspace{1pt}    

\begin{definition}[{\bf Legitimate Transactions}]    
A transaction $\Tx_{B,C}$ that cites $\Tx_{A,B}$ as its parent is considered \emph{legitimate} if all of the following conditions are satisfied:  
\begin{itemize}
\item   \emph{Condition 1:} It must \emph{attach} a valid $\APS$ for the parent transaction $\Tx_{A,B}$, serving as a cryptographic proof of the parent transaction's finality; 
\item   \emph{Condition 2:}  $\Tx_{B,C}$ must not conflict with any other transaction that cites $\Tx_{A,B}$ as its parent;   
\item   \emph{Condition 3:}  The addresses and balances between $\Tx_{B,C}$ and its parent must be consistent; and  
\item   \emph{Condition 4:}  The signature of $\Tx_{B,C}$ must be valid, i.e., it must be correctly signed by $B$.
\end{itemize}
\end{definition}

  \vspace{1pt}

\begin{definition}[{\bf Safety and Liveness}]    
To solve the $\BFT$ consensus problem considered here, the protocol must satisfy the following conditions:
\begin{itemize}
\item {\bf \emph{Safety}}: 
Any two transactions accepted by honest nodes do not conflict. Furthermore, if two valid $\APS$s are generated for two different transactions, then those transactions must also be non-conflicting.
Finally, any node that receives a transaction together with its valid $\APS$ must accept that transaction. 
\item {\bf \emph{Liveness}}: 
If a legitimate transaction is received and proposed by at least one consensus node that remains uncorrupted throughout the protocol, and the transaction remains legitimate, then a valid $\APS$ for the transaction is eventually generated, delivered to, and accepted by all honest consensus nodes and all active $\RPC$ nodes.
\end{itemize}
\end{definition}

 \vspace{1pt}

\begin{definition}[{\bf Stable Liveness}]    
A protocol is said to have \emph{stable liveness} if liveness is guaranteed in the presence of an \emph{adaptive} adversary who can decide which nodes to control at any time, provided that the total number of corrupted nodes is bounded by $t$. 
\end{definition}

\vspace{5pt}

\begin{definition}[{\bf Type~I $\APS$:}]    
A Type~I $\APS$ for a transaction satisfies the Safety conditions: if the network generates valid $\APS$s for two transactions, then those transactions are non-conflicting. Furthermore, any node that receives a transaction together with its valid $\APS$ accepts that transaction.  For any Type~I  transaction, a Type~I $\APS$ is sufficient to verify the transaction finality. 
\end{definition}

\vspace{5pt}

\begin{definition}[{\bf Type~II $\APS$:}]    
Like a Type~I $\APS$, a Type~II $\APS$ satisfies the Safety conditions and, in addition, meets the following requirement: at least $t+1$ honest nodes have received a Type~I $\APS$ for this transaction and have accepted the transaction.  A Type~II $\APS$ is typically used for a Type~II transaction but can also serve as an $\APS$ option for a Type~I transaction. 
\end{definition}

\vspace{5pt}

\begin{remark} [{\bf Type~I   and Type~II $\APS$s:}] 
The proposed  $\Ocior$   guarantees the following two properties. 
\begin{itemize}
    \item In $\Ocior$, if a valid Type~II $\APS$ is generated for a transaction $\tx$, then eventually all honest consensus nodes and $\RPC$ nodes will receive a valid Type~II $\APS$ and accept $\tx$, even if another transaction conflicts with $\tx$ (see Theorem~\ref{thm:APSII} in Section~\ref{sec:OciorAnalysis}).    
    \item In $\Ocior$, if a valid Type~I $\APS$ is generated for a \emph{legitimate} transaction $\tx$, and $\tx$ remains legitimate, then eventually all honest consensus nodes and $\RPC$ nodes will receive a valid Type~II $\APS$ and accept $\tx$ (see Theorem~\ref{thm:APSI} in Section~\ref{sec:OciorAnalysis}).  Note that any node that receives a transaction $\tx$ together with a valid Type~I or Type~II $\APS$ will accept $\tx$, even if another transaction conflicts with it (see Theorem~\ref{thm:Safety} in Section~\ref{sec:OciorAnalysis}). 
\end{itemize}
\end{remark}

 \vspace{5pt}
 
\begin{definition}[{\bf Good-Case Round Complexity:}]    
Following common conventions~\cite{CL:99,YMRGA:19,SKN:25,MXCSS:16,LDZ:2020}, when measuring round complexity, we do not include the communication cost between clients and consensus nodes. The \emph{good-case} round complexity refers to the scenario where the transaction is proposed by any (not necessarily designated) honest node, and is measured from the round in which the node proposes the transaction to the round in which an $\APS$ of the transaction is generated.
\end{definition}

\vspace{5pt}

\begin{definition} [{\bf Honest-Majority Distributed Multicast ($\HMDM$) \cite{Chen:2020arxiv, ChenDISC:21,ChenOciorCOOL:24}}] \label{def:HMDM}
In the distributed multicast ($\DM$) problem, there are $n$ nodes in total in set $\SenderSet$ and $\TatalNumRPCnodes$ nodes in total in another set $\ReceiverSet$, where up to $t$ nodes in $\SenderSet$ may be dishonest.  In this problem, a subset of nodes in $\SenderSet$    act as senders, each multicasting an input message to the nodes in $\SenderSet$ and $\ReceiverSet$.  
A $\DM$ protocol guarantees the following property:  
\begin{itemize}
\item   {\bf Validity:} If all  honest senders input the same message $\wv$,  then every honest node in $\SenderSet$ and $\ReceiverSet$ eventually outputs $\wv$.       
\end{itemize} 
 We call a $\DM$ problem     an honest-majority $\DM$ if  at least   $t+1$ senders are honest.  
\end{definition}

\vspace{5pt}

We present two $\HMDM$ protocols: $\OciorHMDMHash$ and $\OciorHMDMIT$, given in Algorithms~\ref{algm:OciorHMDMHash} and~\ref{algm:OciorHMDMIT}, respectively. The $\OciorHMDMIT$ protocol is derived from the $\COOL$ protocol \cite{Chen:2020arxiv, ChenDISC:21, ChenOciorCOOL:24} and is information-theoretically secure and error-free; that is, it guarantees the required properties in all executions without relying on cryptographic assumptions. $\OciorHMDMIT$ achieves $\HMDM$ consensus in two rounds, with $O(n |\MVBAInputMsg| + n^2 \log \alphabetsize)$ total communication bits and $\tilde{O}(n |\MVBAInputMsg|)$ computation per node in the worst case, where $\alphabetsize$ denotes the alphabet size of the error-correcting code used.     
In contrast, the $\OciorHMDMHash$ protocol is a hash-based $\HMDM$ protocol that completes in one round, with $O(n |\MVBAInputMsg| + \kappa n^2 \log n)$ total communication bits and $\tilde{O}(|\MVBAInputMsg| + \kappa n)$ computation per node, where $\kappa$ is a security parameter. When $|\MVBAInputMsg| \geq \kappa n \log n$, the $\OciorHMDMHash$ protocol becomes an attractive option in terms of round, communication, and computation complexities.

\begin{algorithm}  
\caption{$\OciorHMDMHash$  protocol, with identifier $\ProtocolID$. Code is shown for Node~$\thisnodeindex \in [n]$.}    \label{algm:OciorHMDMHash} 
\begin{algorithmic}[1]
\vspace{5pt}       
\footnotesize
  \Statex   \emph{//   **   This asynchronous $\HMDM$  protocol  is a hash-based protocol.    **}      
\State Initially set  $\CodedSymbols \gets \{\}$; $\kencode \gets t+1$ 		 
 
\Statex
\Statex   \emph{//   ***** Code for Node~$\thisnodeindex\in \SenderSet$ for $\SenderSet:=[n]$    *****} 
\State {\bf upon} receiving input  message  $\HMDMMsg$, and if  Node~$\thisnodeindex$ is a sender {\bf do}:  
\Indent  

	\State $[\EncodedSymbol_{1}, \EncodedSymbol_{2}, \cdots, \EncodedSymbol_{n} ]   \gets \ECEnc(n, \kencode, \HMDMMsg)$       
	\State $(\vectorcommitment, \aux) \gets \VCCom([\EncodedSymbol_{1}, \EncodedSymbol_{2}, \cdots, \EncodedSymbol_{n} ] )$ 
	\State $\proofpositionvc_{\thisnodeindex} \gets \VCOpen(\vectorcommitment, \EncodedSymbol_{\thisnodeindex}, \thisnodeindex, \aux)$  
	\State $\send$ $(\SHARE,  \IDMVBA, \vectorcommitment,  \EncodedSymbol_{\thisnodeindex}, \proofpositionvc_{\thisnodeindex})$ to  all nodes
	\State  $\Output$  $\HMDMMsg$   	
\EndIndent

\Statex

 \Statex   \emph{//   ***** Code for Node~$\thisnodeindex\in \SenderSet \cup \ReceiverSet$    *****} 
\State {\bf upon} receiving   $(\SHARE, \IDMVBA, \vectorcommitment,  \EncodedSymbol, \proofpositionvc)$ from  Node~$j \in [n]$ for the first time, for some $\vectorcommitment, \EncodedSymbol, \proofpositionvc$, and if  Node~$\thisnodeindex$ is not a sender  {\bf do}:  
\Indent  
	\If {$\VCVerify(j, \vectorcommitment,  \EncodedSymbol,   \proofpositionvc) \eqlog \true$}  
	
		\IfThenElse {$\vectorcommitment\notin \CodedSymbols $}  {$\CodedSymbols[\vectorcommitment]\gets \{j: \EncodedSymbol\}$ } {$\CodedSymbols[ \vectorcommitment] \gets \CodedSymbols[ \vectorcommitment] \cup \{j: \EncodedSymbol\}$  }

		\If {$|\CodedSymbols[\vectorcommitment] |=\kencode$}  

					\State  $\hat{\MVBAInputMsg}  \gets\ECDec(n, \kencode, \CodedSymbols[\vectorcommitment])$  
					\State  $\Output$  $\hat{\MVBAInputMsg}$					  
		\EndIf

	\EndIf    
\EndIndent
\end{algorithmic}
\end{algorithm}

\begin{algorithm} 
\caption{$\OciorHMDMIT$ protocol with identifier $\ProtocolID$ .  Code is shown for    Node~$\thisnodeindex \in [n]$. }  \label{algm:OciorHMDMIT}
\begin{algorithmic}[1]
\vspace{5pt}    

\footnotesize

 \Statex   \emph{//   **   This asynchronous $\HMDM$  protocol  is information theocratic secure  and error free  **}

\State Initially set  $\CodedSymbols \gets \{\}$;  $\kencode \gets t+1$ 	  
\Statex
 
\Statex   \emph{//   ***** Code for Node~$\thisnodeindex\in \SenderSet$ for $\SenderSet:=[n]$    *****}  
\State {\bf upon} receiving input message  $\wv$, and if  Node~$\thisnodeindex$ is a sender   {\bf do}:
\Indent  
 
		\State  $[\ECCcodedsymbol_{1}, \ECCcodedsymbol_{2}, \cdots, \ECCcodedsymbol_{\networksizen}] \gets \ECCEnc(\networksizen,  \kencode, \wv)$   
		\State   $\send$ $\ltuple   \SYMBOL, \ProtocolID, \ECCcodedsymbol_{j},   \ECCcodedsymbol_{\thisnodeindex}  \rtuple $ to  Node~$j$,    $\forall j \in  [\networksizen]$
		\State   $\send$ $\ltuple   \SYMBOL, \ProtocolID, \defaultvalue,   \ECCcodedsymbol_{\thisnodeindex}  \rtuple $ to  all nodes in $\ReceiverSet$      
		\State  $\Output$  $\wv$

\EndIndent	

\Statex

\State {\bf upon} receiving  $\networkfaultsizet+1$ $\ltuple   \SYMBOL, \ProtocolID, \ECCcodedsymbol_{\thisnodeindex},  * \rtuple $  messages from  distinct nodes in $\SenderSet$, for the same  $\ECCcodedsymbol_{\thisnodeindex}$, and if  Node~$\thisnodeindex$ is not a sender    {\bf do}:
\Indent  

 		\State   $\send$ $\ltuple   \SYMBOL, \ProtocolID,  \defaultvalue,   \ECCcodedsymbol_{\thisnodeindex} \rtuple $ to all   nodes

\EndIndent

\Statex

\Statex   \emph{//   ***** Code for Node~$\thisnodeindex\in \SenderSet \cup \ReceiverSet$    *****}

\State {\bf upon} receiving message  $\ltuple   \SYMBOL, \ProtocolID,  * , \ECCcodedsymbol_{j} \rtuple $ from Node~$j \in \SenderSet$  for the first time, and if   Node~$\thisnodeindex$ is not a sender     {\bf do}:

\Indent  
	\State $\CodedSymbols[j] \gets \ECCcodedsymbol_{j}$    
 	\If  { $|\CodedSymbols|\geq  \kencode + t  $}     \label{line:OECbegin}    \quad    \quad\quad \quad \quad\quad\quad \quad    \quad\quad \quad \quad\quad\quad \quad   \quad\quad \quad \quad\quad\quad \quad   \quad\quad \quad \quad\quad\quad \quad    \emph{//   online error correcting  (OEC)  }  
			\State   $\MVBAOutputMsg  \gets \ECCDec(n,  \kencode , \CodedSymbols)$	     
			\State  $[y_{1}', y_{2}', \cdots, y_{n}'] \gets \ECCEnc (n,  \kencode, \MVBAOutputMsg)$ 
    			\If {at least $\kencode + t$ symbols in $[y_{1}', y_{2}', \cdots, y_{n}']$ match with  those in $\CodedSymbols$}
				\State  $\Output$  $\MVBAOutputMsg$ 
    			\EndIf    

	\EndIf   
		     
\EndIndent

\end{algorithmic}
\end{algorithm}

 \vspace{5pt}

\begin{definition}[{\bf Vector Commitment ($\VC$)}] 
\label{def:VC}
A vector commitment scheme allows one to commit to an entire vector while enabling efficient proofs of membership for individual positions. 
We consider a vector commitment scheme implemented using a Merkle tree based on hashing. 
It consists of the following algorithms:
\begin{itemize}
    \item $\VCCom(\yv) \to (\vectorcommitment, \aux)$:  
    Given an input vector $\yv = [\EncodedSymbol_{1}, \EncodedSymbol_{2}, \ldots, \EncodedSymbol_{n}]$ of length $n$, this algorithm generates a commitment $\vectorcommitment$ and an auxiliary string $\aux$.  
    In a Merkle-tree-based $\VC$, $\vectorcommitment$ corresponds to the Merkle root, which has size $O(\kappa)$ bits. 
    \item $\VCOpen(\vectorcommitment, \EncodedSymbol_{j}, j, \aux) \to \proofpositionvc_{j}$:  
    Given inputs $(\vectorcommitment, \EncodedSymbol_{j}, j, \aux)$, this algorithm outputs a proof $\proofpositionvc_{j}$ demonstrating that $\EncodedSymbol_{j}$ is indeed the $j$-th element of the committed vector. 
    \item $\VCVerify(j, \vectorcommitment, \EncodedSymbol_{j}, \proofpositionvc_{j}) \to \true/\false$:  
    This algorithm outputs $\true$ if and only if $\proofpositionvc_{j}$ is a valid proof showing that $\vectorcommitment$ is a commitment to a vector whose $j$-th element is $\EncodedSymbol_{j}$. 
\end{itemize} 
\end{definition}

\vspace{5pt}
 
\begin{definition}[{\bf Error Correction Code ($\ECC$)}] \label{def:ECC}  
An $(n, \kencode)$ error correction code consists of the following algorithms:
\begin{itemize}
    \item $\ECCEnc: \Alphabet^{\kencode} \to \Alphabet^{n}$: This encoding algorithm maps a message of $\kencode$ symbols to an output of $n$ symbols. Here, $\Alphabet$ denotes the alphabet of each symbol, and $\alphabetsize := |\Alphabet|$ denotes its size.  
    \item $\ECCDec: \Alphabet^{n'} \to \Alphabet^{\kencode}$: This decoding algorithm recovers the original message from an input of $n'$ symbols, for some $n'$.  
\end{itemize}
Reed-Solomon (RS) codes (cf.~\cite{RS:60}) are widely used error correction codes. An $(n, \kencode)$ RS code can correct up to $t$ Byzantine errors and simultaneously detect up to $a$ Byzantine errors from $n'$ symbol observations, provided that  
\[
2t + a + \kencode \leq n' \quad \text{and} \quad n' \leq n.
\]  
Although RS codes are popular, they impose a constraint on the alphabet size, namely $n \leq \alphabetsize - 1$. To overcome this limitation, other error correction codes with constant alphabet size, such as Expander Codes~\cite{SS:96}, can be employed.  
In the asynchronous setting, online error correction ($\OEC$) provides a natural method for decoding the message \cite{BCG:93}. A node may be unable to recover the message from only $n'$ symbol observations; in such cases, it waits for an additional symbol before attempting to decode again. This procedure continues until the node successfully reconstructs the message. 
\end{definition}
 
 \vspace{5pt}
 
\begin{definition}[{\bf Erasure Code  ($\EC$)}] \label{def:EC}  
An $(n, \kencode)$ erasure code consists of the following algorithms:
\begin{itemize}
    \item $\ECEnc: \Alphabet^{\kencode} \to  \Alphabet^{n}$: This encoding algorithm maps a message of $\kencode$ symbols to an output of $n$ symbols.  
    \item $\ECDec: \Alphabet^{\kencode} \to  \Alphabet^{\kencode}$: This decoding algorithm recovers the original message from an input of $\kencode$ symbols.  
\end{itemize}
By using an $(n, \kencode)$ erasure code, the original message can be recovered from any $\kencode$ encoded symbols.   
\end{definition}

\vspace{5pt}

\begin{definition}[\textbf{Reliable Broadcast ($\RBC$)}]
The $\RBC$ protocol allows a designated leader to broadcast an input value to a set of distributed nodes, while ensuring the following properties: 
\begin{itemize}
\item \textbf{Consistency:} If two honest nodes output values $\wv'$ and $\wv''$, then $\wv' = \wv''$.
\item \textbf{Validity:} If the leader is honest and broadcasts   $\wv$, then all honest nodes eventually output $\wv$.
\item \textbf{Totality:} If any honest node outputs a value, then all honest nodes eventually output a value.
\end{itemize}
\end{definition}

  \vspace{5pt}

\emph{Notations}:  In the asynchronous network considered here, we will use the notion of  \emph{asynchronous rounds} when counting the communication rounds, where each round does not need to be synchronous. 
The computation cost is measured in units of cryptographic operations, including signing, signature verification, hashing, and basic arithmetic operations (addition, subtraction, multiplication, and division) on values of signature size.  

\vspace{3pt}

We use $[b]$ to denote the ordered set $\{1,2,3,\dotsc, b\}$, and   $[a,b]$ to denote the ordered set  $\{a, a+1, a+2,\dotsc, b\}$, for any   integers $a$ and  $b$ such that $b >a$ and $b\geq 1$. 
The symbol $:=$ is used to mean ``is defined as.'' The symbol $\defaultvalue$ denotes a default value or an empty value.  
Let $\Group$ be a cyclic group of prime order $\PrimeOrder$, and let $\Group_T$ be a multiplicative cyclic group of the same order $\PrimeOrder$.  
Let $\FieldZ_{\PrimeOrder}$ denote the finite field of order $\PrimeOrder$.  
Let $\Hash: \{0,1\}^* \to \Group$ be a hash function that maps arbitrary-length bit strings to elements of $\Group$, modeled as a random oracle.  
Let $\HashZ: \{0,1\}^* \to \FieldZ_{\PrimeOrder}$ be a hash function that maps arbitrary-length bit strings to elements of $\FieldZ_{\PrimeOrder}$, also modeled as a random oracle.
Here $f(x)=O(g(x))$ implies that ${\lim\sup}_{x \to \infty} |f(x)|/g(x) < \infty$. 
Similarly, $f(x)=\tilde{O}(g(x))$ implies that ${\lim\sup}_{x \to \infty} |f(x)|/(\log x)^a \cdot g(x) < \infty$, for some constant $a\geq 0$.

\newpage

\section{$\OciorBLS$: A Fast  Non-Interactive Threshold Signature  Scheme} \label{sec:OciorBLS}

We propose a novel non-interactive threshold signature  scheme called $\OciorBLS$. It offers fast signature aggregation in \emph{good cases}  of partial signature collection, as defined in    Definition~\ref{def:GoodCasePSC}.   
Moreover, $\OciorBLS$  enables real-time aggregation of partial signatures as they arrive, reducing waiting time and improving responsiveness.  
In addition, $\OciorBLS$ supports weighted signing power or voting, where nodes may possess different signing weights, allowing for more flexible and expressive consensus policies.

\subsection{Non-Interactive Threshold Signature}

Let us at first provide some definitions on the threshold signature scheme.

\begin{definition}[{\bf Non-Interactive Threshold Signature ($\TS$)}]    
A non-interactive $(n, \TSthreshold)$ $\TS$ scheme allows any $\TSthreshold$ valid partial signatures  to collaboratively generate  a final signature, for some $\TSthreshold\in [t+1, n-t]$. It comprises a tuple of algorithms   $\PtotocolSigma_{\TS} = (\TSSetup, \TSGen,  \Vote, \TSVerify, \TSCombine,\TSVerify)$    satisfying the following properties.  
\begin{itemize} 
\item \textbf{$\TSSetup(1^\kappa) \to (\PrimeOrder, \Group, \Group_{T}, \belinearpairing, \randomgenerator, \Hash)$:}  This algorithm generates the public parameters ($\pp$).  All subsequent algorithms take the public parameters as input, but these are omitted in the presentation for simplicity. 
\item \textbf{$\TSGen(1^\kappa, n, \TSthreshold) \to (\pk, \pk_1, \dotsc, \pk_n, \sk_1, \sk_2, \dotsc, \sk_n)$:}
Given the security parameter $\kappa$, this algorithm generates a set of public keys $\pkvector := (\pk, \pk_1, \dotsc, \pk_n)$, which are available to all nodes, and a private key share $\sk_i$ that is available only to Node~$i$, for each $i \in [n]$.  Any $\TSthreshold$ valid private key shares are sufficient to reconstruct a secret key $ \sk$ corresponding to the public key $\pk$. 
The $\TSGen$ algorithm is implemented using an Asynchronous Distributed Key Generation ($\ADKG$) scheme. 
\item \textbf{$\Vote(\sk_i, \Hash(\wv)) \to \partialsig_i$:}  This algorithm produces the $i$-th partial signature $\partialsig_i$ on the input message $\wv$ using the private key share $\sk_i$, for $i \in [n]$. This partial signature can be interpreted as the vote from Node~$i$ on the message $\wv$.   We sometimes denote this operation as $\Vote_i(\Hash(\wv))$.  
$\Hash()$ is a hash function. 
\item \textbf{$\TSVerify(\pk_i,  \partialsig_i, \Hash(\wv)) \to \true/\false$:}  This algorithm verifies whether $\partialsig_i$ is a valid partial signature on message $\wv$ by using the public key $\pk_i$, for $i \in [n]$. It outputs $\true$ if the verification succeeds, and $\false$ otherwise.    
\item \textbf{$\TSCombine (n, \TSthreshold, \{(i, \partialsig_i)\}_{i\in \Tc\subseteq[n], |\Tc| \geq \TSthreshold}, \Hash(\wv)) \to \finalsig$:}  This algorithm combines any $\TSthreshold$ valid partial signatures $\{\partialsig_i \}_{i\in \Tc\subseteq[n], |\Tc|\geq \TSthreshold}$  on the same message $\wv$ to produce the final signature $\finalsig$, for any $\Tc\subseteq [n]$ with  $|\Tc| \geq \TSthreshold$.        

\item \textbf{$\TSVerify(\pk, \finalsig, \Hash(\wv)) \to \true/\false$:}  This algorithm verifies whether $\finalsig$ is a valid final signature on message $\wv$ using the public key $\pk$. It outputs $\true$ if the verification is successful, and $\false$ otherwise.   
\end{itemize} 
\end{definition}

In our setting, we set  $\TSthreshold=\lceil \frac{n+t+1}{2} \rceil$    for the $\TS$ scheme.  
Here, different $\TS$ schemes are used for different epochs. To express this, we extend the notation and define epoch-specific algorithms as follows:  
$\TSGen(\eon, 1^\kappa, n, \TSthreshold)$,   $\Vote (\sk_{\eon,i}, \Hash(\wv))$, $\TSVerify(\pk_{\eon, i},  \partialsig_i, \Hash(\wv))$,   
where the additional parameter $\eon$ denotes the epoch number associated with the corresponding instance of the $\TS$ scheme.  
 The $\TS$ scheme guarantees the properties of \emph{Robustness} and \emph{Unforgeability} defined below.

Robustness  of the  $\TS$ scheme ensures that an adaptive adversary, who controls up to $t$  nodes, cannot prevent the honest nodes from forming a valid final signature on a message of their choice. This definition is presented in the Random Oracle Model ($\ROM$), where the hash function $\Hash$ is modeled as a publicly available random oracle.

\begin{definition}[{\bf Robustness  of $\TS$ under Adaptive Chosen-Message Attack ($\ACMA$)}] \label{def:robustness-acma}
A non-interactive $(n, \TSthreshold)$ $\TS$ scheme $\PtotocolSigma_{\TS} = (\TSSetup, \TSGen, \Vote, \TSVerify, \TSCombine, \TSVerify)$ is robust under adaptive chosen-message attack ($\TSROBACMA$) if for any probabilistic polynomial-time ($\PPT$) adversary $\Adversary$ that corrupts at most $t$ nodes, the following properties hold with overwhelming probability:

\begin{enumerate}
    \item \textbf{Correctness of Partial Signatures:} For any message $\wv$ and any honest node $i \in [n]$, if a partial signature $\partialsig_i$ is produced by running $\partialsig_i \leftarrow \Vote(\sk_i, \Hash(\wv))$, then $\TSVerify(\pk_i, \partialsig_i, \Hash(\wv)) = \true$.  
    \item \textbf{Correctness of Final Signatures:} For any message $\wv$ and any set of partial signatures $\{\partialsig_i\}_{i \in \Tc \subseteq [n]}$ with $|\Tc| \geq \TSthreshold$, if for each $i \in \Tc$, $\TSVerify(\pk_i, \partialsig_i, \Hash(\wv)) = \true$, then combining them must yield a valid final signature. Specifically, if $\finalsig \leftarrow \TSCombine(n, \TSthreshold, \{(i, \partialsig_i)\}_{i\in \Tc}, \Hash(\wv))$, then $\TSVerify(\pk, \finalsig, \Hash(\wv)) = \true$.  
\end{enumerate}
\end{definition}

Unforgeability of the  $\TS$ scheme ensures that an adaptive adversary cannot forge a valid final signature on a new message, even with the ability to corrupt up to $t$  nodes and obtain partial signatures.  

\begin{definition}[{\bf Unforgeability of $\TS$ under Adaptive Chosen-Message Attack}] \label{def:unf-acma}
Let $\PtotocolSigma_{\TS} = (\TSSetup, \TSGen, \Vote, \TSVerify, \TSCombine, \TSVerify)$ be a non-interactive $(n, \TSthreshold)$ $\TS$ scheme. We say that $\PtotocolSigma_{\TS}$ is unforgeable under adaptive chosen-message Aattack ($\TSUNFACMA$) if for any $\PPT$ adversary $\Adversary$, the advantage of $\Adversary$ in the following game is negligible.

\begin{itemize}
    \item \textbf{Setup:} The challenger runs $(\pk, \pk_1, \dotsc, \pk_n, \sk_1, \dotsc, \sk_n) \leftarrow \TSGen(1^\kappa, n, \TSthreshold)$. It gives the public keys $\pk, \pk_1, \dotsc, \pk_n$ to the adversary $\Adversary$ and keeps the private keys. All nodes, including the adversary, have oracle access to a public random oracle $\RandomOracleH$ that models the hash function $\Hash$.
    \item \textbf{Queries:} The adversary $\Adversary$ is given oracle access to the following queries:
    \begin{itemize}
        \item \textbf{RandomOracleQuery($\wv$):} On input a message $\wv$, the challenger computes $\Hash(\wv)$ by querying the random oracle $\RandomOracleH$ and returns the result to $\Adversary$. The challenger maintains a list $\RandomOracleQuerySet$ of all $\wv$ queried to the random oracle.
        \item \textbf{Corrupt($i$):} On input an index $i \in [n]$, the challenger reveals the private key share $\sk_i$ to the adversary. The set of corrupted nodes is denoted by $\CorruptSet \subseteq [n]$. This query can be made at any time. The adversary $\Adversary$ must not be able to corrupt more than $t$ nodes.
        \item \textbf{PartialSign($i, \wv$):} On input an index $i \in [n]$ and a message $\wv$, the challenger computes $\partialsig_i \leftarrow \Vote(\sk_i, \Hash(\wv))$ and returns it to $\Adversary$. The challenger maintains a list $\PartialSignQuerySet$ of all pairs $(i, \wv)$ for which a partial signature was requested.
    \end{itemize}
    \item \textbf{Challenge:} After  querying phase, the adversary $\Adversary$ outputs a message $\wv^{\star}$ and a final signature $\finalsig^{\star}$.
    \item \textbf{Win:} The adversary $\Adversary$ wins if the following conditions are met:
    \begin{enumerate}
        \item $\TSVerify(\pk, \finalsig^{\star}, \Hash(\wv^{\star})) = \true$.
        \item The number of nodes from which the adversary has obtained a secret share or a partial signature for $\wv^{\star}$ is strictly less than the threshold, i.e., $|\CorruptSet \cup \{i \mid (i, \wv^{\star}) \in \PartialSignQuerySet\}| < \TSthreshold$.
    \end{enumerate}
\end{itemize}
The advantage of the adversary is defined as:
$$
\text{Adv}_{\Adversary}^{\TSUNFACMA}(\kappa) = \Pr[\Adversary~\text{wins}]. 
$$
$\PtotocolSigma_{\TS}$ is $\TSUNFACMA$ secure if for any $\PPT$ adversary $\Adversary$, $\text{Adv}_{\Adversary}^{\TSUNFACMA}(\kappa)$ is a negligible function of $\kappa$.
\end{definition}

\subsection{Non-Interactive Layered Threshold Signature  ($\LTS$)}
 
We here introduce  a new primitive: the Layered Threshold Signature  scheme. 
The proposed $\LTS$ scheme achieves an  aggregation computation cost of only $O(n)$ in \emph{good cases}  of partial signature collection, as defined  in Definition~\ref{def:GoodCasePSC}.      
Moreover, $\LTS$ supports the property of \emph{Instantaneous TS Aggregation}. A $\TS$ scheme guarantees this property if it can aggregate partial signatures immediately, without waiting for  $\TSthreshold$ partial signatures, where $\TSthreshold$ is the threshold required to compute the final signature.

 \begin{definition}[{\bf Non-Interactive Layered Threshold Signature ($\LTS$)}]
The $(n, \TSthreshold, \ltslayerMax, \{n_{\ltslayer}, \TSthreshold_{\ltslayer}, \ltslayerTotalNodes_{\ltslayer}\}_{\ltslayer=1}^{\ltslayerMax})$ $\LTS$ scheme generates a signature from \emph{at least} $\TSthreshold$  valid partial signatures, with the parameters constrained by 
\begin{align}
n=\prod_{\ltslayer=1}^{\ltslayerMax} n_{\ltslayer}, \quad   \prod_{\ltslayer=1}^{\ltslayerMax} \TSthreshold_{\ltslayer} \geq \TSthreshold,  \quad \text{and} \quad \ltslayerTotalNodes_{\ltslayer} := \prod_{\ltslayer'=1}^{\ltslayer} n_{\ltslayer'},  \  \forall \ltslayer \in [\ltslayerMax]  \label{eq:LTSContraints}
\end{align} 
for some $\TSthreshold\in [t+1, n-t]$.  
In this $\LTS$ scheme, multiple layers of partial signatures are generated to produce the final signature, with $\ltslayerMax$ representing the total number of layers.  
Let $\partialsig_{\ltslayer, i}$ denotes the $i$-th partial signature at Layer~$\ltslayer$, for $\ltslayer\in [\ltslayerMax]$ and $i\in [\ltslayerTotalNodes_{\ltslayer}]$. 
At Layer~$\ltslayer$, $n_{\ltslayer}$ denotes the number of partial signatures within a group, and any $\TSthreshold_{\ltslayer}$ of them can be combined to generate a valid group signature. This group signature can be considered a partial signature for a ``parent group'' at the upper layer.  
We define the  $\ltsgroup$-th group of partial signatures at Layer~$\ltslayer$ as 
\begin{align}
\PartialSigltsGroupSet_{\ltslayer, \ltsgroup} := \{\partialsig_{\ltslayer, i} \mid    i \in  \ltsGroupSet_{\ltslayer, \ltsgroup}  \}    \label{eq:LTSGroupPartialSig}
\end{align}
where $\ltsGroupSet_{\ltslayer, \ltsgroup}$ denotes the corresponding set of  indices of the $\ltsgroup$-th group of partial signatures at Layer~$\ltslayer$ 
\begin{align}
\ltsGroupSet_{\ltslayer, \ltsgroup} := [(\ltsgroup -1) n_{\ltslayer} + 1, (\ltsgroup -1) n_{\ltslayer} + n_{\ltslayer}]  \label{eq:LTSGroup}
\end{align}
 for $\ltslayer \in [\ltslayerMax]$ and $\ltsgroup \in [\ltslayerTotalNodes_{\ltslayer-1}] $,  where $\ltslayerTotalNodes_0 := 1$.  
The $i$-th  partial signature $\partialsig_{\ltslayer, i}$ at Layer~$\ltslayer$ maps  to  the $\IndexInParentGroup_{\ltslayer, i}$-th  element of the  group $\PartialSigltsGroupSet_{\ltslayer, \ParentIndex_{\ltslayer, i}}$, where 
  \begin{align}
\ParentIndex_{\ltslayer, i}:=\lceil i/ n_{\ltslayer} \rceil, \quad  \IndexInParentGroup_{\ltslayer, i}:=i - (\lceil i/ n_{\ltslayer} \rceil -1)  n_{\ltslayer}  .     \label{eq:LTSparentMap}
\end{align}
   The  $\LTS$ scheme    $\PtotocolSigma_{\LTS} = (\LTSSetup,\LTSGen,  \VoteLTS, \LTSVerify, \LTSCombine,\LTSVerify)$    satisfies the following properties.  
\begin{itemize}
\item \textbf{$\LTSSetup(1^\kappa) \to (\PrimeOrder, \Group, \Group_{T}, \belinearpairing, \randomgenerator, \Hash)$:}  This algorithm generates the public parameters. All subsequent algorithms take the public parameters as input, but these are omitted in the presentation for simplicity. 
\item \textbf{$\LTSGen(1^\kappa, n, \TSthreshold, \ltslayerMax, \{n_{\ltslayer}, \TSthreshold_{\ltslayer}, \ltslayerTotalNodes_{\ltslayer}\}_{\ltslayer=1}^{\ltslayerMax}) \to (\pkl, \pkl_1, \dotsc, \pkl_n, \skl_1, \dotsc, \skl_n)$:} 
Given the security parameter $\kappa$ and other parameters satisfying $n=\prod_{\ltslayer=1}^{\ltslayerMax} n_{\ltslayer}$, $\prod_{\ltslayer=1}^{\ltslayerMax} \TSthreshold_{\ltslayer} \geq \TSthreshold$, and $\ltslayerTotalNodes_{\ltslayer} = \prod_{\ltslayer'=1}^{\ltslayer} n_{\ltslayer'}$ for each $\ltslayer \in [\ltslayerMax]$, this algorithm generates a set of public keys $\pklvector := (\pkl, \pkl_1, \dotsc, \pkl_n)$ available to all nodes, and a private key share $\skl_i$ available only to Node~$i$, for each $i \in [n]$.  
The $\LTS$ key generation is implemented using an $\ADKG$ scheme. In our setting, it is required that $\pkl = \pk$ and $\skl =  \sk$, where $\pk$ and $ \sk$ are the public and secret keys, respectively, for the $\TS$ scheme defined above.

\item \textbf{$\VoteLTS(\skl_i, \Hash(\wv)) \to \partialsig_{\ltslayerMax, i}$:}  
Given a message $\wv$, this algorithm uses the private key share $\skl_i$ to produce the $i$-th partial signature $\partialsig_{\ltslayerMax, i}$ at Layer~$\ltslayerMax$, for $i \in [n]$. This partial signature can be interpreted as Node~$i$'s vote on $\wv$.

\item \textbf{$\LTSVerify(\pkl_{i}, \partialsig_{\ltslayerMax, i}, \Hash(\wv)) \to \true/\false$:}   
This algorithm verifies whether $\partialsig_{\ltslayerMax, i}$ is a valid partial signature on message $\wv$ using the public key $\pkl_{i}$, for $i \in [n]$. It outputs $\true$ if verification succeeds and $\false$ otherwise.

\item \textbf{$\LTSCombine(n_{\ltslayer}, \TSthreshold_{\ltslayer}, \{(\IndexInParentGroup_{\ltslayer, i} , \partialsig_{\ltslayer, i})\}_{i\in \Tc \subseteq \ltsGroupSet_{\ltslayer, \ltsgroup}, |\Tc|\geq \TSthreshold_{\ltslayer}}, \Hash(\wv)) \to \finalsig_{\ltslayer-1, \ltsgroup}$:}  
This algorithm combines any $\TSthreshold_{\ltslayer}$ valid partial signatures from the $\ltsgroup$-th group at Layer~$\ltslayer$, i.e., $\PartialSigltsGroupSet_{\ltslayer, \ltsgroup}$, on the same message $\wv$ to produce a signature $\finalsig_{\ltslayer-1, \ltsgroup}$, for $\ltslayer \in [\ltslayerMax]$ and $\ltsgroup \in [\ltslayerTotalNodes_{\ltslayer-1}]$, where $\IndexInParentGroup_{\ltslayer, i}$ is defined in \eqref{eq:LTSparentMap}. In our setting, $\finalsig_{0, 1}$ denotes the final signature.       

\item \textbf{$\LTSVerify(\pkl, \partialsig_{0, 1}, \Hash(\wv)) \to \true/\false$:}  
This algorithm verifies whether $\partialsig_{0, 1}$ is a valid final signature on message $\wv$ using the public key $\pkl$. It outputs $\true$ if verification succeeds and $\false$ otherwise.
\end{itemize}
\end{definition} 
In our setting, we set $\TSthreshold=\lceil \frac{n+t+1}{2} \rceil$. 
Here, different $\LTS$ schemes are used for different epochs. We define epoch-specific algorithms as follows:  
$\LTSGen(\eon, 1^\kappa, n, \TSthreshold, \ltslayerMax, \{n_{\ltslayer}, \TSthreshold_{\ltslayer}, \ltslayerTotalNodes_{\ltslayer}\}_{\ltslayer=1}^{\ltslayerMax})$, $\VoteLTS (\skl_{\eon, i}, \Hash(\wv))$, $\LTSVerify(\pkl_{\eon, i}, \partialsig_{\ltslayerMax, i}, \Hash(\wv))$,        and  $\LTSVerify(\pkl_{\eon}, \partialsig_{0, 1}, \Hash(\wv))$, 
where the additional parameter $\eon$ denotes the epoch number associated with the corresponding instance of the $\LTS$ scheme.  
Algorithm~\ref{algm:OciorDKGtd} presents a key generation protocol with a trusted dealer for both the $\TS$ and $\LTS$ schemes, referred to as    $\OciorDKGtd$.  
Algorithm~\ref{algm:OciorADKG} presents the proposed $\ADKG$ protocol, referred to as    $\OciorADKG$.  
Fig.~\ref{fig:OciorLTSexample} illustrates a  tree structure of partial signatures of the proposed $\LTS$ scheme for the example with parameters:  $n=1400$, $t=\lfloor \frac{1400-1}{3} \rfloor = 466$,  $\TSthreshold=\lceil \frac{n+t+1}{2} \rceil =934$, $\ltslayerMax=3$,   $n_1=14, n_2=10, n_3=10$,    $\TSthreshold_1=13$,   $\TSthreshold_2=9$,  $\TSthreshold_3=8$,   $n=n_1 n_2 n_3$,  and  $\TSthreshold=\TSthreshold_1 \TSthreshold_2 \TSthreshold_3$.

  \begin{figure}   
\centering
\includegraphics[width=17.8cm]{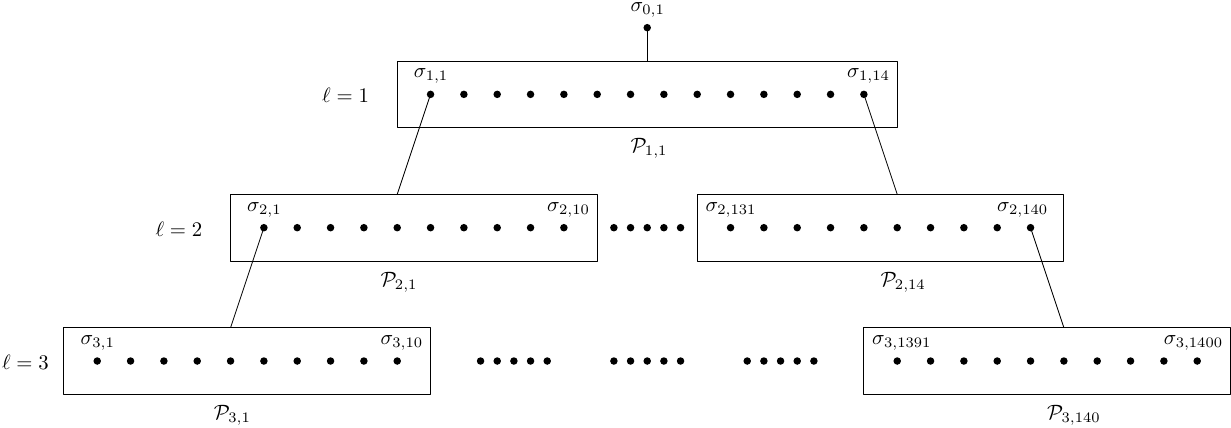}
\caption{A tree structure of partial signatures of the proposed $\LTS$ scheme for the example with parameters:  $n=1400$, $t=\lfloor \frac{1400-1}{3} \rfloor = 466$,  $\TSthreshold=\lceil \frac{n+t+1}{2} \rceil =934$, $\ltslayerMax=3$,   $n_1=14, n_2=10, n_3=10$,    $\TSthreshold_1=13$,   $\TSthreshold_2=9$,  $\TSthreshold_3=8$,   $n=n_1 n_2 n_3$,  and  $\TSthreshold=\TSthreshold_1 \TSthreshold_2 \TSthreshold_3$. In this example, the  partial signatures generated from $n$ distinct nodes at Layer~$\ltslayerMax$ are represented by $\{\partialsig_{3, i}\}_{i\in [n]}$.  The group $\PartialSigltsGroupSet_{3, 1}$  is collection of  $n_3$ corresponding partial signatures,  that is, $\PartialSigltsGroupSet_{3, 1}=\{\partialsig_{3, i} \mid    i \in [10] \}$. Any $\TSthreshold_3$ valid partial signatures in $\PartialSigltsGroupSet_{3, 1}$ can be used to generate the partial signature   $\partialsig_{2, 1}$.  We refer to $\partialsig_{2, 1}$ as the parent of the partial signatures in $\PartialSigltsGroupSet_{3, 1}$.  Similarly,  any $\TSthreshold_2$ valid partial signatures in $\PartialSigltsGroupSet_{2, 1}$ can be used to generate the partial signature   $\partialsig_{1, 1}$, and any $\TSthreshold_1$ valid partial signatures in $\PartialSigltsGroupSet_{1, 1}$ can be used to generate the final signature    $\partialsig_{0, 1}$. 
}
\label{fig:OciorLTSexample}
\end{figure}

 The $\LTS$ scheme guarantees the properties of \emph{Good-Case  Robustness} and \emph{Unforgeability} defined in Definitions~\ref{def:LTSRobustness} and \ref{def:LTSUnforgeability}.
At first, let us provide the following definitions  related to the $\LTS$ scheme.

\begin{definition}[{\bf Faulty Group, Non-Faulty (or Valid) Group, and Valid Partial Signature}]
The $\ltsgroup$-th group of partial signatures at Layer~$\ltslayer$, denoted by $\PartialSigltsGroupSet_{\ltslayer, \ltsgroup}$ as defined in~\eqref{eq:LTSGroupPartialSig}, for $\ltslayer \in [\ltslayerMax]$ and $\ltsgroup \in [\ltslayerTotalNodes_{\ltslayer-1}]$, is said to be \emph{faulty} if the number of faulty, unavailable, or invalid partial signatures in the group exceeds $n_{\ltslayer} - \TSthreshold_{\ltslayer}$. In this case, the group signature it generates is also considered faulty (or invalid). 
Conversely, $\PartialSigltsGroupSet_{\ltslayer, \ltsgroup}$ is said to be \emph{non-faulty} or \emph{valid} if the number of valid partial signatures within the group is greater than or equal to    $\TSthreshold_{\ltslayer}$. In this case, the signature it generates is considered non-faulty  or valid. 
The signature generated by $\PartialSigltsGroupSet_{\ltslayer, \ltsgroup}$ may serve as a partial signature for a ``parent group'' at Layer~$\ltslayer - 1$, and is denoted by $\partialsig_{\ltslayer-1, \ltsgroup}$. Therefore, if $\PartialSigltsGroupSet_{\ltslayer, \ltsgroup}$ is non-faulty, the resulting partial signature $\partialsig_{\ltslayer-1, \ltsgroup}$ is also non-faulty or valid.
\end{definition}

\begin{definition}[{\bf Parent of Partial Signatures}]
The $i$-th partial signature $\partialsig_{\ltslayer, i}$ at Layer~$\ltslayer$ corresponds to the $\IndexInParentGroup_{\ltslayer, i}$-th element of the $\ParentIndex_{\ltslayer, i}$-th group of partial signatures at Layer~$\ltslayer$, denoted by $\PartialSigltsGroupSet_{\ltslayer, \ParentIndex_{\ltslayer, i}}$, where $\ParentIndex_{\ltslayer, i}$ and  $\IndexInParentGroup_{\ltslayer, i}$ are defined in \eqref{eq:LTSparentMap}, 
for $\ltslayer \in [\ltslayerMax]$ and $i \in [\ltslayerTotalNodes_{\ltslayer}]$. 
The partial signature generated by $\PartialSigltsGroupSet_{\ltslayer, \ParentIndex_{\ltslayer, i}}$, denoted by $\partialsig_{\ltslayer-1, \ParentIndex_{\ltslayer, i}}$, is considered the $\ParentIndex_{\ltslayer, i}$-th partial signature at Layer~$\ltslayer - 1$. 
We refer to $\partialsig_{\ltslayer-1, \ParentIndex_{\ltslayer, i}}$ as the \emph{parent} of the partial signature $\partialsig_{\ltslayer, i}$.
\end{definition}

\begin{definition}[{\bf Partial Signature Layer}]
The $\ltslayer$-th layer of partial signatures, denoted by $\PartialSigLayer_{\ltslayer}$, is defined as the collection of all partial signatures at Layer~$\ltslayer$, i.e.,
\[
\PartialSigLayer_{\ltslayer} = \{\partialsig_{\ltslayer, 1}, \partialsig_{\ltslayer, 2}, \ldots, \partialsig_{\ltslayer, \ltslayerTotalNodes_{\ltslayer}} \}
\]
for $\ltslayer \in [\ltslayerMax]$, where $\ltslayerTotalNodes_{\ltslayer} = \prod_{\ltslayer'=1}^{\ltslayer} n_{\ltslayer'}$.
\end{definition}

\begin{definition}[{\bf Partial Signature Tree}]
The full tree of partial signatures, denoted by $\PartialSigTree$, is defined as the union of all partial signature layers, i.e.,
\[
\PartialSigTree = \bigcup_{\ltslayer=1}^{\ltslayerMax} \PartialSigLayer_{\ltslayer}.
\]
\end{definition}

\begin{definition}[{\bf Partial Signature Collection}]  \label{def:PSC}
A partial signature collection, denoted by $\PartialSigCollection$, is a subset of partial signatures at Layer~$\ltslayerMax$, i.e., $\PartialSigCollection \subseteq \PartialSigLayer_{\ltslayerMax}$.
\end{definition}

\begin{definition}[{\bf Good Cases of Partial Signature Collection}]     \label{def:GoodCasePSC}
A partial signature collection $\PartialSigCollection$ is said to be a \emph{good case} of partial signature collection if there exist partial signature subsets $\PartialSigLayerSubset_{\ltslayer} \subseteq \PartialSigLayer_{\ltslayer}$ for each $\ltslayer \in [\ltslayerMax - 1]$, such that:
\begin{itemize}
    \item $\PartialSigLayerSubset_{\ltslayerMax - 1}$ includes all the parents of $\PartialSigCollection$;
    \item for all $\ltslayer \in [\ltslayerMax - 2]$, the subset $\PartialSigLayerSubset_{\ltslayer}$ includes all the parents of $\PartialSigLayerSubset_{\ltslayer + 1}$;
    \item all partial signatures in $\bigcup_{\ltslayer=1}^{\ltslayerMax - 1} \PartialSigLayerSubset_{\ltslayer}$ are valid;
    \item and $|\PartialSigLayerSubset_{1}| \geq \TSthreshold_{1}$.
\end{itemize}
\end{definition}

Good-Case Robustness of the $\LTS$ scheme ensures that an adaptive adversary  cannot prevent the honest nodes from successfully producing a valid final signature  in the good cases of partial signature collection. 

\begin{definition}[{\bf Good-Case Robustness of $\LTS$ under ACMA}] \label{def:LTSRobustness}
A non-interactive $(n, \TSthreshold, \ltslayerMax, \{n_{\ltslayer}, \TSthreshold_{\ltslayer}, \ltslayerTotalNodes_{\ltslayer}\}_{\ltslayer=1}^{\ltslayerMax})$ $\LTS$ scheme $\PtotocolSigma_{\LTS} = (\LTSSetup,\LTSGen, \VoteLTS, \LTSVerify, \LTSCombine,\LTSVerify)$ is good-case robust under adaptive chosen-message attack ($\LTSROBACMA$) if for any $\PPT$ adversary $\Adversary$ that corrupts at most $t$ nodes, the following properties hold with overwhelming probability:

\begin{enumerate}
    \item \textbf{Correctness of   Partial Signatures:} For any message $\wv$ and any  $i \in [n]$, a partial signature $\partialsig_{\ltslayerMax, i} \leftarrow \VoteLTS(\skl_i, \Hash(\wv))$ must always be valid, i.e., $\LTSVerify(\pkl_i, \partialsig_{\ltslayerMax, i}, \Hash(\wv)) = \true$.  
    
    \item \textbf{Correctness of Final Signatures:} For any message $\wv$, let $\PartialSigCollection  \subseteq \PartialSigLayer_{\ltslayerMax}$ be a   partial signature collection (as defined in Definition \ref{def:PSC}) on a message $\wv$. The scheme is robust if $\PartialSigCollection$ is a \emph{good case} of partial signature collection (as defined in Definition \ref{def:GoodCasePSC}), which implies that a valid final signature $\partialsig_{0, 1}$ can be produced through the successive combination of signatures from the lowest layer up to Layer 1, where it can be verified successfully, i.e., $\LTSVerify(\pkl, \partialsig_{0, 1}, \Hash(\wv)) = \true$.  
\end{enumerate}
\end{definition}

Unforgeability of the  $\LTS$ scheme ensures that an adaptive adversary cannot forge a valid final signature on a new message, even with the ability to corrupt up to $t$  nodes and obtain partial signatures.  

\begin{definition}[{\bf Unforgeability of $\LTS$ under ACMA}] \label{def:LTSUnforgeability}
Let $\PtotocolSigma_{\LTS} = (\LTSSetup, \LTSGen, \VoteLTS$, $\LTSVerify, \LTSCombine,\LTSVerify)$ be a non-interactive $(n, \TSthreshold, \ltslayerMax, \{n_{\ltslayer}, \TSthreshold_{\ltslayer}, \ltslayerTotalNodes_{\ltslayer}\}_{\ltslayer=1}^{\ltslayerMax})$ $\LTS$ scheme. We say that $\PtotocolSigma_{\LTS}$ is unforgeable under adaptive chosen-message attack ($\LTSUNFACMA$) if for any $\PPT$ adversary $\Adversary$, the advantage of $\Adversary$ in the following game is negligible.

\begin{itemize}
    \item \textbf{Setup:} The challenger runs $(\pkl, \pkl_1, \dotsc, \pkl_n, \skl_1, \dotsc, \skl_n) \!\!\leftarrow\!\! \LTSGen(1^\kappa\!, \!n, \!\TSthreshold, \!\ltslayerMax, \!\{n_{\ltslayer}, \TSthreshold_{\ltslayer}, \ltslayerTotalNodes_{\ltslayer}\}_{\ltslayer=1}^{\ltslayerMax})$. It gives the public keys $\pkl, \pkl_1, \dotsc, \pkl_n$ to the adversary $\Adversary$ and keeps the private keys. All nodes, including the adversary, have oracle access to a public random oracle $\RandomOracleH$ that models the hash function $\Hash$.
    \item \textbf{Queries:} The adversary $\Adversary$ is given oracle access to the following queries:
    \begin{itemize}
        \item \textbf{RandomOracleQuery($\wv$):} On input a message $\wv$, the challenger computes $\Hash(\wv)$ by querying the random oracle $\RandomOracleH$ and returns the result to $\Adversary$. The challenger maintains a list $\RandomOracleQuerySet$ of all $\wv$ queried to the random oracle.
        \item \textbf{Corrupt($i$):} On input an index $i \in [n]$, the challenger reveals the private key share $\skl_i$ to the adversary. The set of corrupted nodes is denoted by $\CorruptSet \subseteq [n]$. This query can be made at any time. The adversary $\Adversary$ must not be able to corrupt more than $t$ nodes.
        \item \textbf{PartialSign($i, \wv$):} On input an index $i \in [n]$ and a message $\wv$, the challenger computes $\partialsig_{\ltslayerMax, i} \leftarrow \VoteLTS(\skl_i, \Hash(\wv))$ and returns it to $\Adversary$. The challenger maintains a list $\PartialSignQuerySet$ of all pairs $(i, \wv)$ for which a partial signature was requested.
    \end{itemize}
    \item \textbf{Challenge:} After  querying phase, the adversary $\Adversary$ outputs a message $\wv^{\star}$ and a final signature $\finalsig^{\star}$.
    \item \textbf{Win:} The adversary $\Adversary$ wins if the following conditions are met:
    \begin{enumerate}
        \item $\LTSVerify(\pkl, \finalsig^{\star}, \Hash(\wv^{\star})) = \true$.
        \item The number of nodes from which the adversary has obtained a secret share or a partial signature for $\wv^{\star}$ is strictly less than the threshold, i.e., $|\CorruptSet \cup \{i \mid (i, \wv^{\star}) \in \PartialSignQuerySet\}| < \TSthreshold$.
    \end{enumerate}
\end{itemize}
The advantage of the adversary is defined as: $
\text{Adv}_{\Adversary}^{\LTSUNFACMA}(\kappa) = \Pr[\Adversary~\text{wins}]$.  
$\PtotocolSigma_{\LTS}$ is $\LTSUNFACMA$ secure if for any $\PPT$ adversary $\Adversary$, $\text{Adv}_{\Adversary}^{\LTSUNFACMA}(\kappa)$ is a negligible function of $\kappa$.
\end{definition}

\subsection{$\OciorBLS$: A Composition of a Single $\TS$ Scheme and One or More $\LTS$ Schemes}

The proposed $\OciorBLS$ is a composition of a single $\TS$ scheme and one or more $\LTS$ schemes. 
Fig.~\ref{fig:OciorBLS} provides the  description of the algorithms used in the proposed  $\OciorBLS$.   
Algorithm~\ref{algm:OciorBLS} presents $\OciorBLS$ with a single $\TS$ scheme and a single $\LTS$ scheme. 
As discussed in the next subsection, $\OciorBLS$ can be easily extended to support multiple parallel $\LTS$ schemes. 
By combining a single $\TS$ scheme with one or more $\LTS$ schemes, $\OciorBLS$ enables fast signature aggregation in the good cases of partial signature collection, while ensuring both robustness and unforgeability in worst-case scenarios.

While $\OciorBLS$ guarantees both robustness and unforgeability even under asynchronous network conditions and adaptive adversaries, it achieves fast signature aggregation in good cases of partial signature collection.   
Such good cases are likely in scenarios where, for example, network delays are bounded, the actual number of Byzantine nodes is small, and the Byzantine nodes are evenly distributed across groups in the
$\LTS$ schemes. 
To increase the likelihood of such good cases, we incorporate a mechanism that uniformly shuffles nodes into groups at each epoch. 
This promotes a more even distribution of Byzantine nodes across groups for $\LTS$ schemes, thereby improving the chances of good cases of partial signature collection. 
Even in the presence of an \emph{adaptive} adversary, as long as the total number of dishonest nodes remains bounded by $t$ throughout the protocol and shuffling occurs at each epoch, the system will eventually reach a state where the adversary becomes effectively static after $O(t)$ epochs (i.e., it exhausts its corruption budget). 
From that point onward, the probability that good cases of partial signature collection occur becomes very high.

In $\OciorBLS$, an $\ADKG$ protocol is used to generate the public keys and private key shares for both the $\TS$ and $\LTS$ schemes. 
Algorithm~\ref{algm:OciorDKGtd} presents a key generation protocol with a trusted dealer for both the $\TS$ and $\LTS$ schemes, referred to as    $\OciorDKGtd$.  
Algorithm~\ref{algm:OciorADKG} presents the proposed $\ADKG$ protocol, referred to as    $\OciorADKG$. 
It is required that $\pkl = \pk$ and $\skl =  \sk$, where $\pk$ and $ \sk$ are the public key and secret key for the $\TS$ scheme, and $\pkl$ and $\skl$ are the public key and secret key for the $\LTS$ scheme, respectively. 
As described in Algorithm~\ref{algm:OciorBLS}, $\OciorBLS$  involves the following parallel processes for the $\TS$ and $\LTS$ schemes:
\begin{itemize}
\item \textbf{Parallel Signature Generation:}
When a node signs a message $\wv$, it generates a partial signature via $\Vote(\sk_i, \Hash(\wv))$ for the $\TS$ scheme and, in parallel, another partial signature via $\VoteLTS(\skl_i, \Hash(\wv))$ for the $\LTS$ scheme. 
\item \textbf{Parallel  Signature Verification:}
When a verifier checks partial signatures from a signer, it verifies both the $\TS$ and $\LTS$ partial signatures in parallel. 
\item \textbf{Parallel Signature Aggregation:}
When an aggregator collects partial signatures, it performs aggregation for both the $\TS$ and $\LTS$ schemes in parallel. The $\LTS$ scheme supports \emph{instantaneous aggregation}, whereas the $\TS$ scheme must wait for   $\TSthreshold$ partial signatures, plus a  predefined delay period  $\DelayPara$, which allows additional partial signatures to be included in the $\LTS$ scheme, if possible.  
As soon as aggregation completes for either scheme, the process for the other is terminated,  since both schemes produce identical final signatures.
\end{itemize}
 
For the $\LTSCombine$ algorithm described in Fig.~\ref{fig:OciorBLS}, since $n_{\ltslayer}$ and $\TSthreshold_{\ltslayer}$ are both finite numbers (e.g., $n_{\ltslayer} = 10$ and $\TSthreshold_{\ltslayer} = 8$), the nodes can precompute and store all Lagrange coefficients for every possible $\Tc \in [n_{\ltslayer}]$ with $|\Tc| = \TSthreshold_{\ltslayer}$. 
In contrast to traditional $\TS$ signature aggregation, where the computational bottleneck lies in computing Lagrange coefficients for an unpredictable aggregation set, the proposed 
$\LTS$ enables significantly faster aggregation.
The total computation per message for signature aggregation is $O(n)$ in good cases completed via $\LTS$, and up to $O(n \log^2 n)$ in worst-case scenarios completed via the  $\TS$ scheme.

 \subsection{$\OciorBLS$ with Multiple Parallel $\LTS$ Schemes}

One approach to increase the probability of success for the $\LTS$ signature is to allow each node to sign multiple $\LTS$ partial signatures, each generated using a secret key share derived from a different instantiation of  
$\LTSGen(1^\kappa, n, \TSthreshold, \ltslayerMax, \{n_{\ltslayer}, \TSthreshold_{\ltslayer}, \ltslayerTotalNodes_{\ltslayer}\}_{\ltslayer=1}^{\ltslayerMax}) \to (\pkl, \pkl_1, \dotsc, \pkl_n, \skl_1, \dotsc, \skl_n)$, 
with $\pkl = \pk$. In other words, instead of using a single $\LTS$ scheme, we can use multiple (but finite) parallel $\LTS$ schemes, each with node indices \emph{shuffled} differently and \emph{independently}. As long as any one of these $\LTS$ schemes generates a final signature, the signing process is considered complete. For simplicity, we present $\OciorBLS$ using only a single $\LTS$ scheme.  \\

\begin{algorithm}
\caption{$\OciorDKGtd$ protocol for key generation with a trusted dealer, with an identifier $\IDProtocol$.} \label{algm:OciorDKGtd}   
\begin{algorithmic}[1] 
\vspace{5pt}    
\footnotesize

 \Statex Public Parameters: $\pp = (\Group,   \randomgenerator)$;    $(n, \TSthreshold)$ for $\TS$ scheme; and  $(n, \TSthreshold, \ltslayerMax, \{n_{\ltslayer}, \TSthreshold_{\ltslayer}, \ltslayerTotalNodes_{\ltslayer}\}_{\ltslayer=1}^{\ltslayerMax})$ for $\LTS$ scheme under the constraints:  $n = \prod_{\ltslayer=1}^{\ltslayerMax} n_{\ltslayer}$,  $\prod_{\ltslayer=1}^{\ltslayerMax} \TSthreshold_{\ltslayer} \geq \TSthreshold$, and $\ltslayerTotalNodes_{\ltslayer} := \prod_{\ltslayer'=1}^{\ltslayer} n_{\ltslayer'}$ for each  $\ltslayer \in [\ltslayerMax]$,  with $\ltslayerTotalNodes_{0} := 1$,  for some $\TSthreshold\in [t+1, n-t]$, and  $n\geq 3t+1$.  
 Here we set $\TSthreshold=\lceil \frac{n+t+1}{2} \rceil$.     $\LTSIndexBook$ is a book of  index mapping for $\LTS$ scheme, available at all nodes. 
 
 \Statex
\State  Set $\PolyDegree:= \TSthreshold - 1$ and  $\PolyDegree_{\ltslayer} := \TSthreshold_{\ltslayer} - 1$ for $\ltslayer \in [\ltslayerMax]$.     Set     $\jshuffle:= \LTSIndexBook[\ltuple \IDProtocol, j \rtuple], \forall j \in [n]$

\Statex

\Statex   \emph{//   ***** for the $\TS$ scheme   *****} 

 		\State  sample   $\PolyDegree$-degree random polynomials $\PolynomialFunction(\cdot) \in \FieldZ_{\FiniteFieldSize}[x]$, where $\skinput:=\PolynomialFunction(0)   \in \FieldZ_{\FiniteFieldSize}$ is a randomly generated secret

 \State    set $\pk:= \randomgenerator^{\PolynomialFunction(0)}$,  $\sk_{j} :=  \PolynomialFunction(j)$,   $\pk_{j} := \randomgenerator^{\sk_{j}}$,  $\forall    j\in [n]$

 \Statex
 
 \Statex   \emph{//   ***** for the $\LTS$ scheme   *****}

  \State		sample  $\PolyDegree_{\ltslayer}$-degree  random polynomials   $\LTSPolynomialFunction_{\ltslayer, \ltsgroup} (\cdot)\in \FieldZ_{\FiniteFieldSize}[x]$,   
 for each $ \ltslayer\in[\ltslayerMax]$ and  $\ltsgroup \in [\ltslayerTotalNodes_{\ltslayer-1}]$ under the following constraints:  
 \[\LTSPolynomialFunction_{1,1}(0) =\skinput     \quad \text{and}\quad  \LTSPolynomialFunction_{\ltslayer,\ltsgroup}(0)  =\LTSPolynomialFunction_{\ltslayer -1, \ParentIndex_{\ltslayer-1, \ltsgroup} }(\IndexInParentGroup_{\ltslayer-1, \ltsgroup}) ,  \quad  \forall  \ltslayer\in[2, \ltslayerMax], \ltsgroup \in [\ltslayerTotalNodes_{\ltslayer-1}] \]
  where       $\ParentIndex_{\ltslayer, \ltsgroup}:=\lceil \ltsgroup/ n_{\ltslayer} \rceil $,   $\IndexInParentGroup_{\ltslayer, \ltsgroup}:=\ltsgroup - (\lceil \ltsgroup/ n_{\ltslayer} \rceil -1)  n_{\ltslayer}$      \label{line:AVSSPolyFLTS}

 \State    set $\pkl:= \pk$,  $\skl_{j} :=  \LTSPolynomialFunction_{\ltslayerMax,    \ParentIndex_{\ltslayerMax, j}} (\IndexInParentGroup_{\ltslayerMax, j}) $,   $\pkl_{j} := \randomgenerator^{\skl_{j}}$,  $\forall    j\in [n]$

 \Statex

\State    $\return$   $(\sk_{j}, \skl_{\jshuffle},    \pk, \pkl, \{\pk_{i}\}_{i\in [n]}, \{\pkl_{i}\}_{i\in [n]})$  	to Node~$j$, $\forall    j\in [n]$

\end{algorithmic}
\end{algorithm}

\begin{figure}
\centering
\begin{tcolorbox}[title=$\OciorBLS$: A Composition of a Single $\TS$ Scheme and One or More $\LTS$ Schemes, colframe=blue!55, colback=white, width=1\textwidth]
\small   

 {\bf $\underline{\TSSetup(1^\kappa) \to (\PrimeOrder, \Group, \Group_{T}, \belinearpairing, \randomgenerator, \Hash)}$.}  \\ 
 This algorithm generates the public parameters ($\pp$). 
\begin{enumerate}
    \item \emph{Elliptic Curve and Pairing Selection:}
    Let $\Group$ be a cyclic group of prime order $\PrimeOrder$.  
    Let $\randomgenerator$ be a generator of $\Group$.
    Choose a computable, non-degenerate bilinear pairing $\belinearpairing: \Group \times \Group \to \Group_T$, where $\Group_T$ is a multiplicative cyclic group of order $\PrimeOrder$. The pairing $\belinearpairing$ must satisfy the following properties:
    \begin{itemize}
        \item \emph{Bilinearity:} For all $u, v \in \Group$, and $a, b \in \FieldZ_\PrimeOrder$, $\belinearpairing(u^a, v^b) = \belinearpairing(u, v)^{ab}$.
        \item \emph{Non-degeneracy:} $\belinearpairing(\randomgenerator, \randomgenerator) \ne 1$.
        \item \emph{Computability:} There is an efficient algorithm to compute $\belinearpairing(u, v)$ for any $u, v \in \Group$.
    \end{itemize}

    \item \emph{Hash Function:}
    Choose a cryptographic hash function $\Hash: \{0,1\}^* \to \Group$. This hash function maps arbitrary messages to elements in $\Group$. It is often modeled as a random oracle for security proofs.

    \item \emph{Public Parameters:}
    The public parameters are $\pp = (\PrimeOrder, \Group, \Group_{T}, \belinearpairing, \randomgenerator, \Hash)$.
\end{enumerate}

 \vspace{10pt}
 
 {\bf $\underline{\TSGen(1^\kappa, n, \TSthreshold) \to (\pk, \pk_1, \dotsc, \pk_n, \sk_1, \sk_2, \dotsc, \sk_n)}$.}  \\
 Every honest node $i$ eventually outputs a private key share $\sk_i:=\skshare_i \in \FieldZ_{\FiniteFieldSize}$ and a public key vector
    \[
    \pkvector := (\pk, \pk_1, \dotsc, \pk_n),   \quad \text{where}\quad  \pk:=\randomgenerator^{\skshare} \  \text{and} \ \   \pk_i:=\randomgenerator^{\skshare_i},  \  \text{for} \   i \in [n]  
    \]
    where    $\skshare$ is the final secret jointly generated by the $n$ distributed nodes.  Any subset of $\TSthreshold$ valid private key shares from $\{\skshare_1, \skshare_2, \ldots, \skshare_n\}$ can reconstruct the same unique secret $\skshare$.

  \vspace{10pt}

{\bf $\underline{\Vote(\sk_i, \Hash(\wv)) \to \partialsig_i}$.}   \\  
To sign a message $\wv \in \{0,1\}^*$,  the $i$-th signer, who holds the secret key share  $\sk_i$,    first computes the hash of the message: $\Hash(\wv) \in \Group$, and then computes the 
$i$-th partial signature $\partialsig_i$  as: \[\partialsig_i = \Hash(\wv)^{\sk_i} \in \Group.\]
 
   \vspace{10pt}
   
{\bf $\underline{\TSVerify(\pk_i,  \partialsig_i, \Hash(\wv)) \to \true/\false}$.}   \\  
This algorithm verifies  the $i$-th partial signature $\partialsig_i$ on a message $\wv$ using the corresponding public key $\pk_i$. 
It first computes the hash of the message: $\Hash(\wv) \in \Group$, and then checks the pairing equation: \[\belinearpairing(\partialsig_i, \randomgenerator) \stackrel{?}{=} \belinearpairing(\Hash(\wv), \pk_i).\]
If the equality holds, the partial signature is valid and the algorithm returns $\true$. Otherwise, it is invalid and the algorithm returns $\false$.

  \vspace{10pt}
 
  {\bf $\underline{\TSCombine (n, \TSthreshold,    \TSGroupAgg:=\{(i, \partialsig_i)\}_{i\in \Tc'\subseteq[n], |\Tc'|\geq \TSthreshold}, \Hash(\wv)) \to \finalsig}$.}   \\
  It is required that partial signatures $\partialsig_i$ be successfully verified before being added to the set $\TSGroupAgg$.  \\
    It is also required that   $|\TSGroupAgg| \geq \TSthreshold$ when running this algorithm.  \\
Let $\Tc \subseteq \Tc':= \{i \in [n] \mid (i, \partialsig_i) \in \TSGroupAgg\}$ be a subset of indices included in $\TSGroupAgg$, with $|\Tc|=\TSthreshold$. \\
Let $\LagrangeCoefficientIntAtZero_{\Tc, i} = \prod_{j \in \Tc, j \neq i} \frac{j}{j - i}$ be the $i$-th Lagrange coefficient, for $i\in \Tc$. \\
The final signature $\finalsig$ is computed as follows and then returned:
\[\finalsig =\prod_{i\in \Tc}\partialsig_i^{\LagrangeCoefficientIntAtZero_{\Tc, i} }. \] 

    \vspace{10pt}
  
  {\bf $\underline{\TSVerify(\pk, \finalsig, \Hash(\wv)) \to \true/\false}$.}   \\  
This algorithm verifies  the final signature $\finalsig$ on a message $\wv$ using a corresponding public key $\pk$.  This verification  of the final signature is similar to that of a partial signature, by checking   the pairing equation: $\belinearpairing(\finalsig, \randomgenerator) \stackrel{?}{=} \belinearpairing(\Hash(\wv), \pk)$.
If the equality holds,  the algorithm returns $\true$; otherwise, it   returns $\false$.

\end{tcolorbox}
\end{figure}

\begin{figure}
\centering
\begin{tcolorbox}[title=$\OciorBLS$: A Composition of a Single $\TS$ Scheme and One or More $\LTS$ Schemes   (Continued), colframe=blue!55, colback=white, width=1\textwidth]
\small

 {\bf $\underline{\LTSSetup(1^\kappa) \to (\PrimeOrder, \Group, \Group_{T}, \belinearpairing, \randomgenerator, \Hash)}$.}  \\ 
 The $\LTS$ scheme use the same public parameters $(\PrimeOrder, \Group, \Group_{T}, \belinearpairing, \randomgenerator, \Hash)$ as the $\TS$ scheme.  

 \vspace{10pt}
 
 {\bf $\underline{\LTSGen(1^\kappa, n, \TSthreshold, \ltslayerMax, \{n_{\ltslayer}, \TSthreshold_{\ltslayer}, \ltslayerTotalNodes_{\ltslayer}\}_{\ltslayer=1}^{\ltslayerMax}) \to (\pkl, \pkl_1, \dotsc, \pkl_n, \skl_1, \dotsc, \skl_n)}$.}  \\
 Every honest node $i$ eventually outputs a private key share $\skl_i := \skshareLTS_i \in \FieldZ_{\FiniteFieldSize}$ and a public key vector
    \[
    \pklvector := (\pkl, \pkl_1, \dotsc, \pkl_n),   \quad \text{where}\quad  \pkl=\pk=\randomgenerator^{\skshare} \  \text{and} \ \   \pkl_i:=\randomgenerator^{\skshareLTS_i}, \  \text{for} \   i \in [n]  
    \]
    where    $\skshare$  is the same as that used in the $\TS$ scheme, with $\pkl=\pk=\randomgenerator^{\skshare}$.    A  secret $\skshare$ may   be reconstructed only when at least $\TSthreshold$ valid private key shares  from $\{\skshareLTS_1, \skshareLTS_2, \ldots, \skshareLTS_n\}$ are provided; and  all  reconstructed secrets  are identical.

  \vspace{10pt}

{\bf $\underline{\VoteLTS(\skl_i, \Hash(\wv)) \to \partialsig_{\ltslayerMax, i}}$.}   \\  
To sign a message $\wv \in \{0,1\}^*$,  the $i$-th signer, who holds the secret key share  $\skl_i$,    first computes the hash of the message: $\Hash(\wv) \in \Group$, and then computes the 
$i$-th partial signature $\partialsig_{\ltslayerMax, i}$ at Layer~$\ltslayerMax$, for $i \in [n]$:  \[\partialsig_{\ltslayerMax, i} = \Hash(\wv)^{\skl_i} \in \Group.\]

   \vspace{10pt}

 {\bf $\underline{\LTSVerify(\pkl_{i}, \partialsig_{\ltslayerMax, i}, \Hash(\wv)) \to \true/\false}$.}   \\  
This algorithm verifies  the $i$-th partial signature $\partialsig_{\ltslayerMax, i}$ at Layer~$\ltslayerMax$ on a message $\wv$ using the corresponding public key $\pkl_{i}$. 
It first computes the hash of the message: $\Hash(\wv) \in \Group$, and then checks the pairing equation: \[\belinearpairing(\partialsig_{\ltslayerMax, i}, \randomgenerator) \stackrel{?}{=} \belinearpairing(\Hash(\wv), \pkl_i).\]
If the equality holds,  the algorithm returns $\true$; otherwise, it returns $\false$.

  \vspace{10pt}
 
  {\bf $\underline{\LTSCombine(n_{\ltslayer}, \TSthreshold_{\ltslayer}, \LTSGroupAgg[\ltuple\ltslayer, \ltsgroup\rtuple]:=\{(\IndexInParentGroup_{\ltslayer, i} , \partialsig_{\ltslayer, i})\}_{i\in \Tc' \subseteq \ltsGroupSet_{\ltslayer, \ltsgroup}, |\Tc'| \geq \TSthreshold_{\ltslayer}}, \Hash(\wv)) \to \finalsig_{\ltslayer-1, \ltsgroup}}$}  for $\ltslayer \in [\ltslayerMax]$ and $\ltsgroup \in [\ltslayerTotalNodes_{\ltslayer-1}]$. \\
  When $\ltslayer=\ltslayerMax$, it is required that partial signatures $\partialsig_{\ltslayerMax, i}$ be successfully verified before being added to the set $\LTSGroupAgg$.  \\
    It is also required that   $|\LTSGroupAgg[\ltuple\ltslayer, \ltsgroup\rtuple]| \geq \TSthreshold_{\ltslayer}$ when running this algorithm.  \\
Let $\Tc \subseteq \Tc' := \{i \in \ltsGroupSet_{\ltslayer, \ltsgroup}  \mid (\IndexInParentGroup_{\ltslayer, i} , \partialsig_{\ltslayer, i}) \in \LTSGroupAgg[\ltuple\ltslayer, \ltsgroup\rtuple]\}$ be a subset of indices included in $\LTSGroupAgg[\ltuple\ltslayer, \ltsgroup\rtuple]$, with $|\Tc|=\TSthreshold_{\ltslayer}$.   \\
Here $\ltsGroupSet_{\ltslayer, \ltsgroup} := [(\ltsgroup -1) n_{\ltslayer} + 1, (\ltsgroup -1) n_{\ltslayer} + n_{\ltslayer}]$,     and $\IndexInParentGroup_{\ltslayer, i}:=i - (\lceil i/ n_{\ltslayer} \rceil -1)  n_{\ltslayer}$.  \\  
Let $\LagrangeCoefficientIntAtZero_{\Tc, i} = \prod_{j \in \Tc, j \neq i} \frac{j}{j - i}$ be the $i$-th Lagrange coefficient, for $i\in \Tc$.  \\
Since $n_{\ltslayer}$ and $\TSthreshold_{\ltslayer}$ are both finite numbers (e.g., $n_{\ltslayer} = 10$ and $\TSthreshold_{\ltslayer} = 8$), the nodes can precompute and store all Lagrange coefficients for every possible $\Tc \in [n_{\ltslayer}]$ with $|\Tc| = \TSthreshold_{\ltslayer}$.   \\
The combined signature $\finalsig_{\ltslayer-1, \ltsgroup}$ is  computed as follows and then returned:
\[\finalsig_{\ltslayer-1, \ltsgroup} =\prod_{i\in \Tc}\partialsig_{\ltslayer, i}^{\LagrangeCoefficientIntAtZero_{\Tc, i} } .\] 
 
    \vspace{10pt}
  
  {\bf $\underline{\LTSVerify(\pkl, \partialsig_{0, 1}, \Hash(\wv)) \to \true/\false}$.}   \\  
This algorithm verifies  the final signature $\partialsig_{0, 1}$ on a message $\wv$ using a corresponding public key $\pkl=\pk$.  This verification  of the final signature is similar to that of a partial signature, by checking   the pairing equation: $\belinearpairing(\partialsig_{0, 1}, \randomgenerator) \stackrel{?}{=} \belinearpairing(\Hash(\wv), \pkl)$.
If the equality holds,  the algorithm returns $\true$; otherwise, it   returns $\false$.

\end{tcolorbox}
      \vspace{-8pt}
\caption{The description of the algorithms used in the proposed  $\OciorBLS$.  The   proposed  $\OciorBLS$ is  a composition of a single $\TS$ scheme and one or more $\LTS$ schemes. For the $\LTSCombine$ algorithm, since $n_{\ltslayer}$ and $\TSthreshold_{\ltslayer}$ are both finite numbers (e.g., $n_{\ltslayer} = 10$ and $\TSthreshold_{\ltslayer} = 8$), the nodes can precompute and store all Lagrange coefficients for every possible $\Tc \in [n_{\ltslayer}]$ with $|\Tc| = \TSthreshold_{\ltslayer}$. Compared to traditional $\TS$ signature aggregation, whose computational bottleneck lies in computing Lagrange coefficients for an unpredictable aggregated set, the signature aggregation in the proposed $\LTS$ is significantly faster.}
\label{fig:OciorBLS}
\end{figure}

\begin{algorithm}
\caption{$\OciorBLS$ protocol,   with an epoch identity $\eon$. Code is shown for Node~$\thisnodeindex \in [n]$.} \label{algm:OciorBLS}   
\begin{algorithmic}[1] 
\vspace{5pt}    
\footnotesize

\Statex \emph{// ** For simplicity, we present $\OciorBLS$ with a single $\TS$ scheme and a single $\LTS$ scheme. **}
\Statex \emph{// ** $\OciorBLS$ can be extended easily to support multiple parallel $\LTS$ schemes. **}

\Statex

 \Statex Public Parameters: $\pp = (\PrimeOrder, \Group, \Group_{T}, \belinearpairing, \randomgenerator, \Hash)$;  $(n, \TSthreshold)$ for $\TS$ scheme; and  $(n, \TSthreshold, \ltslayerMax, \{n_{\ltslayer}, \TSthreshold_{\ltslayer}, \ltslayerTotalNodes_{\ltslayer}\}_{\ltslayer=1}^{\ltslayerMax})$ for $\LTS$ scheme under the constraints:  $n = \prod_{\ltslayer=1}^{\ltslayerMax} n_{\ltslayer}$,  $\prod_{\ltslayer=1}^{\ltslayerMax} \TSthreshold_{\ltslayer} \geq \TSthreshold$, and $\ltslayerTotalNodes_{\ltslayer} := \prod_{\ltslayer'=1}^{\ltslayer} n_{\ltslayer'}$ for each  $\ltslayer \in [\ltslayerMax]$,  with $\ltslayerTotalNodes_{0} := 1$,  for some $\TSthreshold\in [t+1, n-t]$.    $\DelayPara$ is a preset delay parameter.

\Statex

\Statex Run  a  key generation protocol  to generate keys for this epoch:   
\Statex $\TSGen(\eon, 1^\kappa, n, \TSthreshold) \to (\pk_{\eon}, \pk_{\eon,1}, \dotsc, \pk_{\eon,n}, \sk_{\eon,1}, \sk_{\eon,2}, \dotsc, \sk_{\eon,n})$ 
\Statex $\LTSGen(\eon, 1^\kappa, n, \TSthreshold, \ltslayerMax, \{n_{\ltslayer}, \TSthreshold_{\ltslayer}, \ltslayerTotalNodes_{\ltslayer}\}_{\ltslayer=1}^{\ltslayerMax}) \to (\pkl_{\eon}, \pkl_{\eon,1}, \dotsc, \pkl_{\eon,n}, \skl_{\eon,1}, \dotsc, \skl_{\eon,n})$
\Statex  It is required that  $\pk_{\eon}= \pkl_{\eon}$. 
\Statex  Every honest node $i$ eventually outputs secret key share $\sk_i$ for $\TS$ scheme, and secret key  share $\skl_i$  for for $\LTS$ scheme. 
 \Statex  Every honest node $i$ also eventually outputs   public key vectors $(\pk_{\eon}, \pk_{\eon,1}, \dotsc, \pk_{\eon,n})$ and $(\pkl_{\eon}, \pkl_{\eon,1}, \dotsc, \pkl_{\eon,n})$. 
\Statex Every honest node updates a dictionary $\LTSIndexBook$ that maps node indices to new indices each epoch, based on index shuffling.

\Statex

\Statex  Initialize  $\Global~$   dictionaries    $\TSGroupAgg\gets \{\};  \LTSGroupAgg\gets \{\}$

\Statex $\Global~$    $\thisnodeindexshuffle \gets\LTSIndexBook[\ltuple \eon,  \thisnodeindex\rtuple]$     \quad   \emph{// $\thisnodeindexshuffle$ is a new index of this node for $\LTS$ scheme based on  index shuffling,       changed every epoch} 
 
 \Statex

\Statex   \emph{//   ***** As a Signer (or Voter)     *****} 

 	\State {\bf upon}  the condition that signing a message $\wv$ with identity $\IDProtocol$ is satisfied       {\bf do}:          		
	\Indent
		\State  $\send$ $\ltuple \VOTE, \IDProtocol,  \Vote(\sk_{\eon, \thisnodeindex}, \Hash(\wv)), \VoteLTS(\skl_{\eon, \thisnodeindexshuffle}, \Hash(\wv))   \rtuple$  to the aggregator  
	\EndIndent
 
\Statex

\Statex   \emph{//   ***** As an Aggregator, who has the  message $\wv$ for  identity $\IDProtocol$   *****} 

\State {\bf upon} receiving    $\ltuple \VOTE, \IDProtocol,  \vote, \votelts  \rtuple$ from Node~$j\in [n]$ for the first time  {\bf do}:  	 \label{line:OciorBLStsAggRx}     \quad       \emph{//   parallel processing between Line~\ref{line:OciorBLStsAggRx} and Line~\ref{line:OciorBLStsTSAgg}}    
\Indent  
	 \If {$\sig$ has not yet been delivered for identity $\IDProtocol$}    		 
 		\State    $\contenthash  \gets \Hash(\wv)$
		\State $\jshuffle\gets \LTSIndexBook[\ltuple \eon,  j\rtuple]$
		\If {$\TSVerify(\pk_{\eon, j}, \vote, \contenthash) = \true$ and $\LTSVerify(\pkl_{\eon,  \jshuffle}, \votelts, \contenthash) = \true$}    
			\State   $[\indicator,  \sig]  \gets  \SigAggregationLTS(\IDProtocol, \eon,  \jshuffle, \votelts, \contenthash)$
	 			\If {$\indicator=\true$ and  $\sig$ has not yet been delivered for identity $\IDProtocol$}    		
	 				\State   $\deliver$ $\sig$ for identity $\IDProtocol$ 
			\EndIf

			\IfThenElse {$\IDProtocol \notin \TSGroupAgg$}  {$\TSGroupAgg[\IDProtocol]\gets  \{j: \vote\}$} {$\TSGroupAgg[\IDProtocol][j]\gets   \vote$}    
		\EndIf	
	\EndIf
\EndIndent
 
\Statex

	\State {\bf upon}      $|\TSGroupAgg[\IDProtocol]|  =  n-t $  and   $\sig$ has not yet been delivered for identity $\IDProtocol$ {\bf do}:  	 \label{line:OciorBLStsTSAgg}      \quad    \emph{//   parallel processing between Line~\ref{line:OciorBLStsAggRx} and Line~\ref{line:OciorBLStsTSAgg}}    
	
	\Indent
		\State  $\wait$ for $\DelayPara$ time     \quad   \emph{//  to include  more partial signatures and complete $\LTS$ scheme, if possible, within the limited delay time}
	 	\If {$\sig$ has not yet been delivered for identity $\IDProtocol$}    		 
		\State  $\sig \gets \TSCombine(n, \TSthreshold, \TSGroupAgg[\IDProtocol], \contenthash)$    
	 		\If {$\sig$ has not yet been delivered for identity $\IDProtocol$}    		
	 				\State   $\deliver$ $\sig$ for identity $\IDProtocol$ 
			\EndIf	

		\EndIf
	\EndIndent
	
\Statex

\Procedure{$\SigAggregationLTS$}{$\IDProtocol, \eon,  \jshuffle, \votelts, \contenthash$}      \label{line:OciorBLStsSigAggregationLTS}

			\State  $\ltsgroup\gets \lceil \jshuffle/ n_{\ltslayerMax} \rceil$ \quad     \emph{//  map Node~$j$ to the $\ltsgroup$th group   and its  index $\ltsnodeindex$ within  this group at Layer~$\ltslayerMax$ for this epoch}       
			\State     $\ltsnodeindex\gets  \jshuffle - (\ltsgroup -1)  n_{\ltslayerMax} $         
			\IfThenElse {$\ltuple \IDProtocol, \ltslayerMax, \ltsgroup\rtuple      \notin \LTSGroupAgg$}  {$\LTSGroupAgg[\ltuple \IDProtocol, \ltslayerMax, \ltsgroup\rtuple]\gets  \{\ltsnodeindex: \votelts\}$ } {$\LTSGroupAgg[\ltuple \IDProtocol, \ltslayerMax, \ltsgroup\rtuple] [\ltsnodeindex]\gets \votelts $      }

		\State $\ltslayer \gets  \ltslayerMax$ 	
	\While {$\ltslayer \geq 1$} 
		\If {$|\LTSGroupAgg[\ltuple \IDProtocol, \ltslayer, \ltsgroup\rtuple]| = \TSthreshold_{\ltslayer}$}  
			\State  $\sig \gets \LTSCombine(n_{\ltslayer}, \TSthreshold_{\ltslayer}, \LTSGroupAgg[\ltuple\IDProtocol, \ltslayer, \ltsgroup\rtuple], \contenthash)$     
 			\If {$\ltslayer =1$}   
				\State $\Return$ $[\true, \sig]$ 	
			\EndIf	
		\Else
			\State $\Break$
		\EndIf
		\State $\ltsgroupprime\gets \lceil \ltsgroup/ n_{\ltslayer -1} \rceil$ 
		\State $\ltsnodeindex\gets  \ltsgroup - (\ltsgroupprime -1)  n_{\ltslayer -1} $   
		\State $\ltslayer \gets  \ltslayer-1$ 	
		\State $\ltsgroup\gets  \ltsgroupprime$     
		\IfThenElse {$\ltuple \IDProtocol, \ltslayer, \ltsgroup\rtuple      \notin \LTSGroupAgg$}  {$\LTSGroupAgg[\ltuple \IDProtocol, \ltslayer, \ltsgroup\rtuple]\gets  \{\ltsnodeindex: \sig\}$ } {$\LTSGroupAgg[\ltuple \IDProtocol, \ltslayer, \ltsgroup\rtuple] [\ltsnodeindex]\gets \sig $      }   
	
	\EndWhile

\State $\Return$ $[\false, \defaultvalue]$

\EndProcedure

\end{algorithmic}
\end{algorithm}

 \section{$\Ocior$}   \label{sec:Ocior}

The proposed $\Ocior$ protocol is an \emph{asynchronous} $\BFT$ consensus protocol.  
It is described in Algorithm~\ref{algm:Ocior} and supported by Algorithms~\ref{algm:OciorProcedures} and~\ref{algm:OciorNoFPKLSeed}.   
Before delving into the full protocol, we first introduce some key features and concepts of $\Ocior$.

 \subsection{Key Features and Concepts of $\Ocior$}   \label{sec:OciorKeyFeatures}
   
{\bf \emph{Parallel Chains:}} Unlike traditional blockchain consensus protocols, which typically operate on a single chain, $\Ocior$ is a chained-type protocol consisting of $n$ parallel chains. The $i$-th chain is proposed independently by Node~$i$, in parallel with others, for all $i \in [n]$.  
Fig.~\ref{fig:OciorChains} illustrates these parallel chains in $\Ocior$, where each Chain~$i$, proposed by Node~$i$, grows over heights $\height = 0, 1, 2, \dots$.  
Each Chain~$i$ links a sequence of threshold signatures $\sig_{i,0} \to \sig_{i,1} \to \sig_{i,2} \to \cdots$, where $\sig_{i,\height}$ is a threshold signature generated at  height $\height$ on a transaction $\tx$ proposed by Node~$i$.    
It is possible for multiple signatures, e.g., $\sig_{i,\height}$ and $\sig_{i,\height}'$, to be generated at the same height of the same chain by a dishonest node, each corresponding to a different transaction, $\tx$ and $\tx'$, respectively.

\begin{figure} [h!]
\centering
\includegraphics[width=13cm]{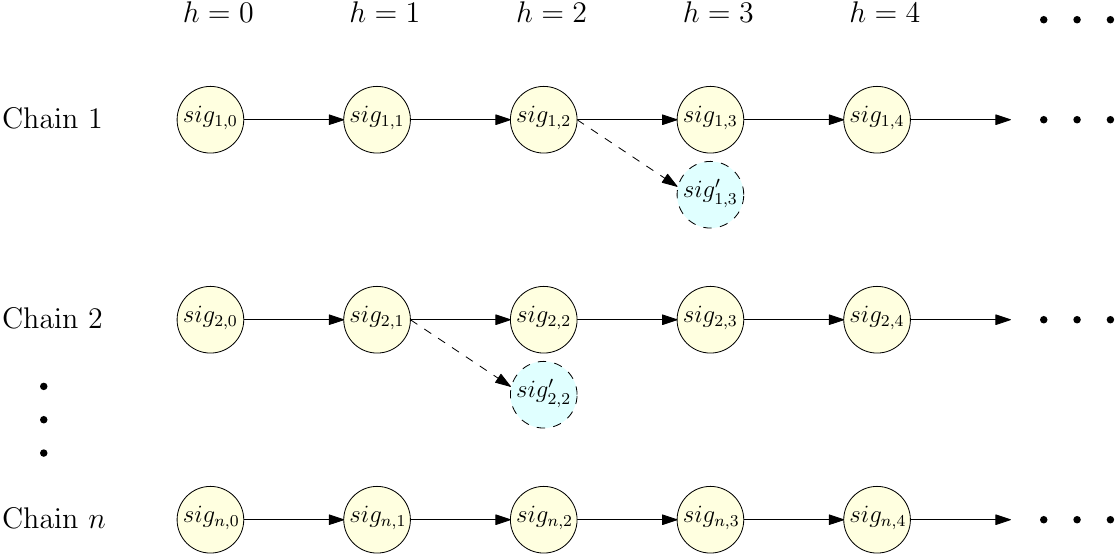}
\caption{The description of the parallel chains in Ocior.   
$\sig_{i,\height}$ and $\sig_{i,\height}'$ are distinct  threshold signatures generated at the same height $\height$ of the chain proposed by a dishonest Node~$i$, with each signature corresponding to a different transaction, for $i \in [n]$ and $\height \geq 0$. 
}
\label{fig:OciorChains}
\end{figure}

{\bf  \emph{Epoch:}} The $\Ocior$ protocol proceeds in epochs. In each epoch, each consensus node is allowed to propose up to $\MaxNumTx$ transactions, where $\MaxNumTx$ is a preset parameter. 
 In $\Ocior$, the consensus nodes receive incoming transactions from connected $\RPC$ nodes, where clients can also serve as $\RPC$ nodes.
The selection of new transactions follows rules that will be described later.  Fig.~\ref{fig:OciorEpoch} illustrates the epoch structure of $\Ocior$, where different nodes may enter a new epoch at different times due to the asynchronous nature of the network.

{\bf \emph{Asynchronous Distributed Key Generation ($\ADKG$):}}  
Before the beginning of a new epoch, an $\ADKG$ scheme is executed to generate new threshold signature keys for the $\TS$ and $\LTS$ schemes for the upcoming epoch, as described in Fig.~\ref{fig:OciorEpoch}.  
Specifically, during the system initialization phase, an $\ADKG[1]$ scheme is executed to generate keys for Epoch~$1$.  
During Epoch~$\eon$, for $\eon \geq 1$, the $\ADKG[\eon+1]$ scheme is executed to generate keys for Epoch~$\eon+1$, with $\ADKG[\eon+1]$ running in parallel to the transaction-chain growth of Epoch~$\eon$. 
When $\MaxNumTx$ is sufficiently large, the cost of $\ADKG$ schemes becomes negligible. 
Algorithm~\ref{algm:OciorADKG} presents the proposed $\ADKG$ protocol, called $\OciorADKG$, which generates the public keys and private key shares for both the $\TS$ and $\LTS$ schemes.  
When $\eon \geq 1$,  $\ADKG[\eon+1]$  can be implemented using an efficient key-refreshing scheme based on the existing $\TS$ keys of Epoch~$\eon$.

\begin{figure}[h!] 
\centering
\includegraphics[width=13cm]{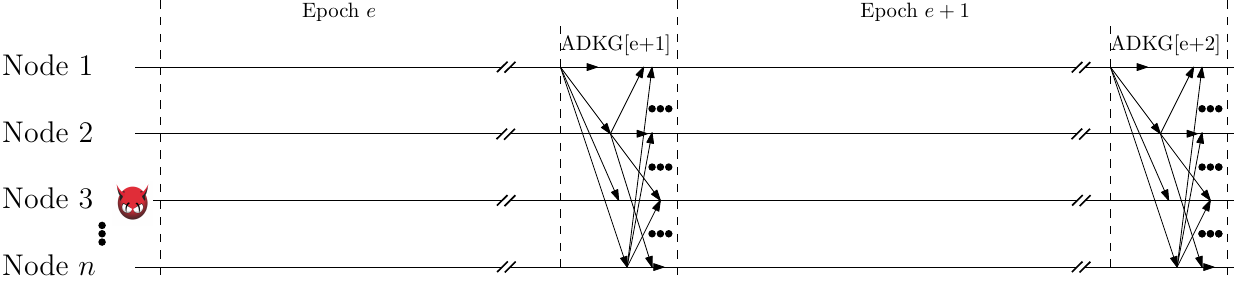}
\caption{The epoch structure of $\Ocior$, where different nodes may enter a new epoch at different times due to the asynchronous nature of the network.   
During Epoch~$\eon$, for $\eon \geq 1$, the $\ADKG[\eon+1]$ scheme is executed to generate keys for Epoch~$\eon+1$, with $\ADKG[\eon+1]$ running in parallel to the transaction-chain growth of Epoch~$\eon$. 
When each Node~$i$ enters a new Epoch~$\eon$, it erases all old private key shares  from previous epochs. 
When $\eon \geq 1$,  $\ADKG[\eon+1]$  can be implemented using an efficient key-refreshing scheme based on the existing $\TS$ keys of Epoch~$\eon$.
 }
\label{fig:OciorEpoch}
\end{figure}

{\bf \emph{Erasing Old Private Key Shares:}}  
When each Node~$i$ enters a new Epoch~$\eon$, it erases all old private key shares $\{\sk_{\eonprime, \thisnodeindex}, \skl_{\eonprime, \cdot}\}_{\eonprime=1}^{\eon-1}$, along with any temporary data related to those secrets from previous epochs. Erasing the old private key shares enhances the \emph{long-term security} of the chains.

{\bf \emph{Signature Content:}} A threshold signature $\sig_{i,\height}$ generated at  height $\height$ of Chain~$i$ during Epoch~$\eon$ is signed on a message of the form \[\ltuple   i, \eon, \NumProposed,  \height,  \tx, \sigvp, \sigoptuple \rtuple\]
for the $\NumProposed$-th transaction $\tx$  proposed by Node~$i$ in  Epoch~$\eon$, where $\sigvp$ is a virtual parent signature and $\sigoptuple$ is a tuple of official parent  signatures, as defined below. 
By our definition, $\sig_{i,\height}$ links to the virtual parent signature $\sigvp$ and to the tuple of official parent signatures $\sigoptuple$.
If there is only one official parent signature, then $\sigoptuple$ reduces to a single signature.

{\bf \emph{Official Parent ($\OP$) and Virtual Parent ($\VP$):}}  
As defined in Definition~\ref{def:parenttx}, if Client~$B$ subsequently initiates a new transaction $\Tx_{B,C}$ to transfer the assets received in $\Tx_{A,B}$ to Client~$C$, then $\Tx_{A,B}$ must be cited as a \emph{parent} transaction.  
The signature on transaction $\Tx_{A,B}$ becomes the official parent signature of the signature on transaction $\Tx_{B,C}$.  
As illustrated in Fig.~\ref{fig:OciorChainsOPVP}, if $\sig_{1,3}'$ is the signature on transaction $\Tx_{A,B}$ and $\sig_{i,5}$ is the signature on transaction $\Tx_{B,C}$, then $\sig_{1,3}'$ is an official parent signature of $\sig_{i,5}$.  
At the same time, as shown in Fig.~\ref{fig:OciorChainsOPVP}, the signature $\sig_{i,5}$, generated at height~$5$, also links to the signature $\sig_{i,4}$ as its virtual parent signature.  
It is possible for a single transaction $\tx$ to have multiple threshold signatures generated on different chains.

\begin{figure} [h!]
\centering
\includegraphics[width=13cm]{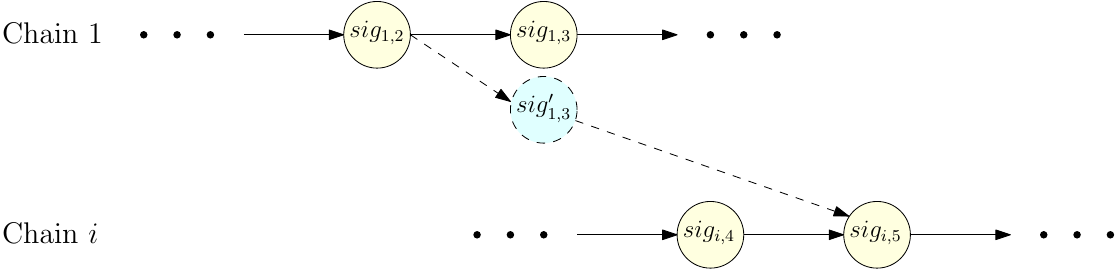}
\caption{Description of the official parent signature and virtual parent signature in $\Ocior$. 
In this example, $\sig_{1,3}'$ is the official parent signature of $\sig_{i,5}$. 
At the same time, the signature $\sig_{i,5}$, generated at height~$5$, links to $\sig_{i,4}$ as its virtual parent signature. }
\label{fig:OciorChainsOPVP}
\end{figure}

{\bf \emph{Acceptance Weight ($\AW$):}}  
Each Chain~$i$, proposed by Node~$i$, grows over heights $\height = 0, 1, 2, \dots$.  
For Chain~$i$, we have $\sig_{i,0} \to \sig_{i,1} \to \sig_{i,2} \to \cdots \to \sig_{i,\heightstar}$, where $\heightstar$ is the top height.  
For any $\sig_{i,\heightprime}$ included in this chain, the \emph{acceptance weight} of $\sig_{i,\heightprime}$ is defined as  
\[
\AW(\sig_{i,\heightprime}) = \heightstar - \heightprime + 1.
\]  
If $\sig_{i,\heightprime}$ is signed on a transaction $\tx$, then the acceptance weight of $\tx$ is \emph{at least} $\heightstar - \heightprime + 1$, i.e.,  
\[
\AW(\tx) \geq \heightstar - \heightprime + 1.
\]

 {\bf \emph{Transaction Identity ($\ID$) and Transaction Clusters:}}  
The identity of a transaction $\tx$, denoted by $\idtx$, is defined as the hash output of $\tx$, i.e., $\idtx = \HashZ(\tx)$, where $\HashZ: \{0,1\}^* \to \FieldZ_{\PrimeOrder}$ is a hash function that maps arbitrary-length bit strings to elements of $\FieldZ_{\PrimeOrder}$.  
A transaction $\tx$ is said to belong to the $\thisnodeindex$-th transaction cluster if $\HashZ(\tx) \mod n = \thisnodeindex -1$, where $\thisnodeindex \in [n]$.

{\bf \emph{Randomized Sets:}}  
A \emph{randomized set} is implemented internally using a list and a dictionary. The average time complexity for adding, removing an element, or randomly selecting a value from a randomized set $R$ is $O(1)$. 
 
{\bf \emph{Rules for New Transaction Selection:}}  
The consensus nodes select new transactions to propose by following a set of designed rules aimed at maximizing system throughput.  
These rules are implemented in $\NewTXProcess$ (Lines~\ref{line:GetNewTxBegin}-\ref{line:GetNewTxEnd} of Algorithm~\ref{algm:OciorProcedures}).  
Specifically, each Node~$\thisnodeindex$ first selects a transaction identity $\idtx$ from $\TnewIDSetOtherProposedLastEpoch$, and then $\TAWOneIDSetOtherProposedLastEpoch$, if these sets  are not empty. Here, $\TAWOneIDSetOtherProposedLastEpoch$ denotes a randomized set of identities of \emph{accepted} transactions (with acceptance weight less than $3$ but greater than $0$) proposed by other nodes in the previous epoch, while $\TnewIDSetOtherProposedLastEpoch$ denotes a randomized set of identities of \emph{pending} transactions proposed by other nodes in the previous epoch. 
If both $\TnewIDSetOtherProposedLastEpoch$  and $\TAWOneIDSetOtherProposedLastEpoch$ are empty, Node~$\thisnodeindex$ selects $\idtx$ from $\TnewselfQueue$ with probability $1/\NumProposedIntevalRandomTxSelf$, where $\NumProposedIntevalRandomTxSelf$ is a preset parameter and $\TnewselfQueue$ is a FIFO queue containing identities of pending transactions within Cluster~$\thisnodeindex$, maintained by node~$\thisnodeindex$.  
Otherwise, Node~$\thisnodeindex$ selects $\idtx$ from $\TnewIDSetOtherProposed$ with probability \[(1 - \frac{1}{\NumProposedIntevalRandomTxSelf}) \cdot \frac{1}{\NumProposedIntevalRandomTxFromOtherProp}\]
where $\NumProposedIntevalRandomTxFromOtherProp$ is a preset parameter and $\TnewIDSetOtherProposed$ is a randomized set of identities of pending transactions proposed by other nodes.   
Finally, Node~$\thisnodeindex$ selects $\idtx$ from $\TnewIDSet$ with the remaining probability, where $\TnewIDSet$ is a randomized set containing identities of pending transactions.

{\bf \emph{Type~I $\APS$:}}  
If a threshold signature $\sig$ is generated at height~$\height$ of Chain~$i$ on a signature content of the form  
\[
\content = \ltuple i, *, *, \height, \tx, *, * \rtuple
\]  
for a transaction $\tx$, then $\ltuple \sig, \content \rtuple$ is called a Type~I $\APS$ for the transaction $\tx$.  

{\bf \emph{Type~II $\APS$:}}  
If two threshold signatures $\sig$ and $\sigprime$ are generated at heights~$\height$ and $\height+1$ of Chain~$i$ on signature contents of the forms  
\[
\content = \ltuple i, *, *, \height, \tx, *, * \rtuple
\]  
and  
\[
\contentprime = \ltuple i, *, *, \height+1, *, \sig, * \rtuple,
\]  
respectively, for a transaction $\tx$, then $\ltuple \sig, \content, \sigprime, \contentprime \rtuple$ is called a Type~II $\APS$ for the transaction $\tx$.

{\bf \emph{Proposal Completion:}}  
A proposal made by Node~$j$ for a transaction $\txdiamond$ is considered completed if a threshold signature on $\txdiamond$ from this proposal has been generated, or if $\txdiamond$ conflicts with another transaction.   
A consensus Node~$i$ must verify that all of Node~$j$'s preceding proposals in the same Epoch~$\eonstar$, as well as all proposals voted on by Node~$\thisnodeindex$ for Chain~$j$, are completed before voting on  Node~$j$'s current proposal.

 \subsection{Basic $\Ocior$}

 Before presenting the chained $\Ocior$ protocol, we first describe the basic $\Ocior$ protocol to clarify the ideas underlying the chained version.  
Figure~\ref{fig:Ocior2Rounds} illustrates the two-round consensus process of the basic $\Ocior$ for two-party transactions. In this example, a transaction $\Tx_{A,B}$ is proposed by an honest node, Node~$1$, and consensus is finalized in two rounds: \emph{Propose} and \emph{Vote}.  

The transaction $\Tx_{A,B}$ represents a transfer from Client~$A$ to Client~$B$. 
Client~$A$ submits $\Tx_{A,B}$ together with $(\sigop, \contentop)$, where $\sigop$ denotes the threshold signature and $\contentop$ denotes the corresponding content of the official parent transaction linked to $\Tx_{A,B}$.
For simplicity, in this example we focus on the transaction where there is only one official parent transaction.  
Client~$A$ can act as an $\RPC$ node to submit $\ltuple \TX,  \Tx_{A,B},  \sigop, \contentop   \rtuple$  directly to connected consensus nodes, such as Node~$1$ in this example.

{\bf \emph{Propose:}} Upon receiving $\ltuple \TX,  \Tx_{A,B},  \sigop, \contentop   \rtuple$, if $\Tx_{A,B}$ is legitimate and its acceptance weight is less than $2$, Node~$1$ proposes $\Tx_{A,B}$ by sending the following proposal message to all consensus nodes:  
\[
\ltuple \PROPOSE, 1, \eon, \NumProposed, \height, \Tx_{A,B}, \sigvp, \sigop, \contentop, *, *, * ,* \rtuple. 
\]
Here, $\eon$ denotes the current epoch number, $\NumProposed$ indicates that $\Tx_{A,B}$ is the $\NumProposed$-th transaction proposed by this node during the current epoch, $\height$ denotes the height of the most recent proposed transaction, and $\sigvp$ denotes the virtual parent signature linked to $\Tx_{A,B}$. The remaining elements, denoted by $*$, will be defined later in the chained protocol.

{\bf \emph{Vote:}} Upon receiving  $\ltuple \PROPOSE, 1, \eonstar, \numproposedstar, \heightstar, \Tx_{A,B}, \sigvp, \sigop, \contentop, *, *, *, * \rtuple$   
from Node~$1$, and after verifying that $\Tx_{A,B}$ is legitimate and that $\eon = \eonstar$, each honest Node~$\thisnodeindex$ sends, in the second round, its $\TS$ and $\LTS$ partial signatures (also called votes) back to Node~$1$. The message sent takes the form:  
\[
\ltuple \VOTE, 1, \eonstar, \numproposedstar, \Vote(\sk_{\eonstar, \thisnodeindex}, \Hash(\content)), \VoteLTS(\skl_{\eonstar, \thisnodeindexshuffle}, \Hash(\content)) \rtuple
\]  
where the partial signatures are computed over:  
\[
\content = \ltuple 1, \eonstar, \numproposedstar, \heightstar, \Tx_{A,B}, \sigvp, \sigop \rtuple
\]  
Here, $\sk_{\eonstar, \thisnodeindex}$ and $\skl_{\eonstar, \thisnodeindexshuffle}$ are Node~$\thisnodeindex$'s secret key shares for the $\TS$ and $\LTS$ schemes, respectively, for Epoch~$\eonstar$.

After receiving  $\TSthreshold = \lceil \frac{n+t+1}{2} \rceil$ valid $\TS$ partial signatures from distinct nodes, the proposer (Node~$1$) generates a threshold signature $\sig$ over the $\content$ of $\Tx_{A,B}$. In good cases of partial signature collection, the threshold signature $\sig$ can be generated from valid $\LTS$ partial signatures.  
The pair $(\sig, \content)$ forms the \emph{$\APS$} for $\Tx_{A,B}$, which can be verified by any node. Node~$1$ then sends $\ltuple \TXDONE, *, *,  \sig, \content  \rtuple$  to  $\RPC$ nodes, which propagate this $\APS$ to clients.

As will be shown later, once an $\APS$ for $\Tx_{A,B}$ is created, no other $\APS$ for a conflicting transaction can be formed. Moreover, if $\Tx_{A,B}$ is legitimate, both Client~$A$ and Client~$B$ will eventually receive its valid $\APS$. Upon receiving a valid $\APS$ for $\Tx_{A,B}$, Client~$B$ may initiate a new transaction $\Tx_{B,*}$ to transfer the assets obtained in $\Tx_{A,B}$, citing $(\sig, \content)$ as the $\APS$ of the \emph{parent} transaction $\Tx_{A,B}$.

 \begin{figure} 
\centering
\includegraphics[width=11cm]{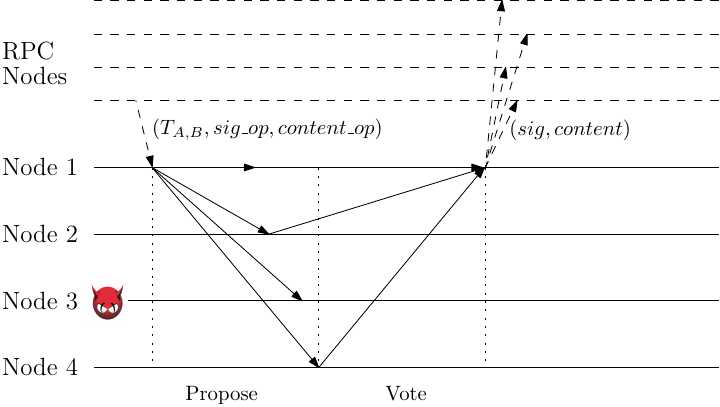}
\caption{The two-round consensus of basic Ocior for two-party transactions. In this description, the transaction $\Tx_{A,B}$ is proposed by an honest node, Node~$1$, and the consensus is finalized in two rounds: \emph{Propose} and \emph{Vote}.}
\label{fig:Ocior2Rounds}
\end{figure}

 \subsection{Chained $\Ocior$}

 The chained $\Ocior$ is described in Fig.~\ref{fig:ChainedOcior} for Chain~$i$ with $i \in [n]$. The other chains follow a similar Propose-Vote structure. When a transaction $\tx$ is proposed by an honest Node~$i$ at height $\height$, this proposal links exactly one fixed threshold signature $\sig_{i, \height-1}$, generated at height $\height-1$, as its virtual parent signature. At the end of the \emph{Vote} round, a threshold signature $\sig_{i, \height}$ for $\tx$ may be generated at height $\height$. This signature $\sig_{i, \height}$ can then be linked as the virtual parent signature to the next proposal at height $\height+1$.   Here, $\sig_{i, 0}$ denotes an initialized threshold signature. 
Chained $\Ocior$ is described in Algorithm~\ref{algm:Ocior} and supported by Algorithms~\ref{algm:OciorProcedures} and~\ref{algm:OciorNoFPKLSeed}.  
Tables~\ref{tb:OciorNotation} and~\ref{tb:OciorNotationContinued} provide some notations for the proposed $\Ocior$ protocol.   
We now present an overview of the chained $\Ocior$ protocol from the perspective of an honest Node~$\thisnodeindex$, with $\thisnodeindex \in [n]$.

\begin{figure}  
\centering
\includegraphics[width=13cm]{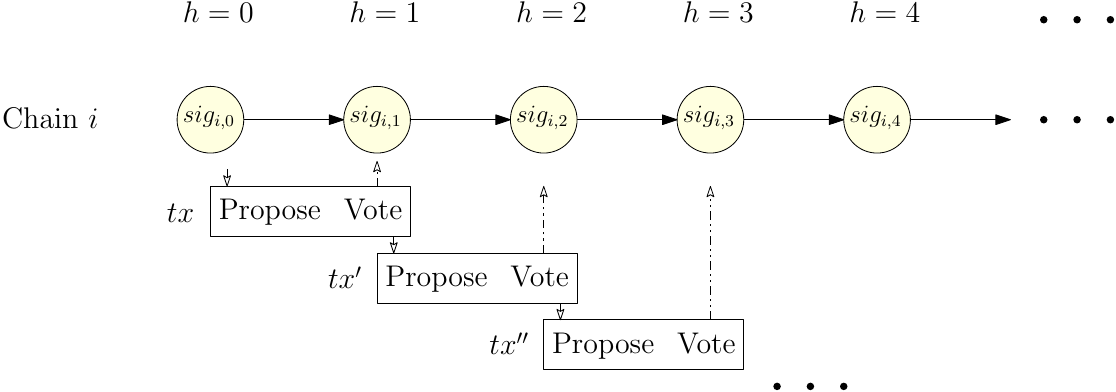}
\caption{The structure of chained $\Ocior$, shown here for Chain~$i$ with $i \in [n]$. The other chains follow a similar Propose-Vote structure.  
}
\label{fig:ChainedOcior}
\end{figure}

{\bf \emph{New Transaction Selection:}}                 
Node~$\thisnodeindex$, for $\thisnodeindex \in [n]$, maintains a set of pending transactions that are legitimate and have an acceptance weight less than $\AWthrehold$, by updating $\TnewIDSet$, $\TnewIDSetOtherProposed, \TAWOneIDSetOtherProposed$,   $\TnewIDSetOtherProposedLastEpoch$, and $\TAWOneIDSetOtherProposedLastEpoch$. 
The parameter $\AWthrehold$ is a  threshold on the acceptance weight, with $\AWthrehold = 2$ for Type~I transactions (see Theorems~\ref{thm:APSIIthreshold2} and \ref{thm:APSIthreshold2} in Section~\ref{sec:OciorAnalysis}) and $\AWthrehold = 3$ for Type~II transactions (see Theorems~\ref{thm:APSII} and \ref{thm:APSI} in Section~\ref{sec:OciorAnalysis}).      
     For simplicity and consistency, we set $\AWthrehold = 3$ for all transactions in the protocol description. 
     Here, $\TnewIDSet$ denotes a randomized set containing the IDs of \emph{pending} transactions. 
The sets $\TnewIDSetOtherProposed$ and $\TnewIDSetOtherProposedLastEpoch$ are randomized sets containing the IDs of \emph{pending} transactions proposed by other nodes in the current epoch and in the previous epoch, respectively. 
Similarly, $\TAWOneIDSetOtherProposed$ and $\TAWOneIDSetOtherProposedLastEpoch$ are randomized sets containing the IDs of \emph{accepted} transactions (with acceptance weight less than $3$ but greater than $0$) proposed by other nodes in the current epoch and in the previous epoch, respectively. 
Node~$\thisnodeindex$ selects new transactions to propose by following a set of rules designed to maximize system throughput. These rules are implemented in $\NewTXProcess$ (Lines~\ref{line:GetNewTxBegin}-\ref{line:GetNewTxEnd} of Algorithm~\ref{algm:OciorProcedures}) and are described in Section~\ref{sec:OciorKeyFeatures}.

{\bf \emph{Propose:}} Node~$\thisnodeindex$ executes the steps of the \emph{Propose} round as specified in Lines~\ref{line:OciorFixstartpro}-\ref{line:OciorProposalEnd} of Algorithm~\ref{algm:Ocior}. 
Specifically,  at the $\NumProposed$-th transaction most recently proposed in the current epoch~$\eon$, Node~$\thisnodeindex$ selects a new transaction $\tx$ that is legitimate and has an acceptance weight less than $\AWthrehold$, following the new transaction selection rules. Node~$\thisnodeindex$ then proposes $\tx$ by sending the following proposal message to all consensus nodes:
\begin{align}
\ltuple \PROPOSE, \thisnodeindex, \eon, \NumProposed, \height, \tx, \sigvp, \sigoptuple, \contentoptuple, \eondiamond, \numproposeddiamond, \myproofofPreviousProposedTx, \myproofProposedTx \rtuple       \label{eq:OciorProposal}
\end{align}  
Here:
 \begin{itemize}
    \item $\height$ is the height of the most recent proposed transaction.
\item $\sigvp$ is the threshold signature generated at height $\height-1$, serving as the virtual parent signature linked to $\tx$.  
  The virtual parent transaction was proposed by Node~$\thisnodeindex$ as the $\numproposeddiamond$-th transaction during Epoch~$\eondiamond$.  
  Consensus nodes must verify a valid and fixed $\sigvp$ generated at height $\height-1$ of Chain~$\thisnodeindex$ before voting for the transaction at height $\height$.

\item  $\sigoptuple$ and $\contentoptuple$ are, respectively, a tuple of official parent signatures and the corresponding contents linked to $\tx$.  
  Consensus nodes must verify valid official parent signatures $\sigoptuple$ before voting for the proposed transaction.

\item  $\myproofofPreviousProposedTx$ is a proof that the previous proposal made by Node~$\thisnodeindex$ for the $(\NumProposed-1)$-th transaction (or a transaction proposed in the previous epoch), denoted $\txdiamond$, has been completed.  
A proposal for a transaction $\txdiamond$ is considered complete if a threshold signature on $\txdiamond$ from this proposal has been generated,   or if $\txdiamond$ conflicts with another transaction.  
Consensus nodes must verify that $\myproofofPreviousProposedTx$ indeed proves the proposal completion  for   $\txdiamond$ before voting on the $\NumProposed$-th transaction, for any $\NumProposed \geq 1$.
The value of $\myproofofPreviousProposedTx$ is set as follows:  
 \begin{itemize}
    \item  If a threshold signature has been generated by this node for $\txdiamond$, then $\myproofofPreviousProposedTx$ is set as follows (Lines~\ref{line:OciorMyProofNormalDoneLTS} and \ref{line:OciorMyProofNormalDoneTS} of Algorithm~\ref{algm:Ocior}):  
\begin{align}
     \myproofofPreviousProposedTx = \defaultvalue .      \label{eq:OciorMyProofDefault}
\end{align}  
  \item If $\txdiamond$ conflicts with another transaction $\txconflict$, then $\myproofofPreviousProposedTx$ is set as follows (Line~\ref{line:OciorMyProofC}):  
\begin{align}
 \myproofofPreviousProposedTx = \ltuple \ProofTypeConflict, \eonccirc, \numproposedcirc,  \txconflict \rtuple       \label{eq:OciorMyProofC}
\end{align} 
where $\ltuple  \eonccirc, \numproposedcirc \rtuple$ denotes a pair of indices of the proposal for $\txdiamond$.  
\end{itemize}
\item  $\myproofProposedTx$ is a Type~I $\APS$ proof for the proposed transaction $\tx$ if it has been proposed by another node and accepted at Node~$\thisnodeindex$, but the acceptance weight is less than $\AWthrehold$.  
If $\tx$ conflicts with another transaction, the consensus nodes must verify that $\myproofProposedTx$ is a valid $\APS$ proof for $\tx$ before voting for it.
The value of $\myproofProposedTx$ is set as follows (see Lines~\ref{line:GetNewTxNoKLStepBegin}-\ref{line:GetNewTxNoKLStepEnd} of Algorithm~\ref{algm:OciorProcedures}):  
\begin{itemize}
    \item If $\tx$ has been accepted at Node~$\thisnodeindex$, but the acceptance weight is less than $\AWthrehold$ (except for a specific case, see Line~\ref{line:GetNewTxNoKLCondition} of  Algorithm~\ref{algm:OciorProcedures}), then $\myproofProposedTx$ is set as follows (Line~\ref{line:proofProposedTxvalue} of Algorithm~\ref{algm:OciorProcedures}):  
    \begin{align}
        \myproofProposedTx = \ltuple \sig, \ltuple j, \eonstar, \numproposedstar, \heightstar, \defaultvalue, \sigvp, \defaultvalue \rtuple \rtuple \label{eq:proofProposedTxvalue}  				
    \end{align}
    where $\ltuple \sig, \ltuple j, \eonstar, \numproposedstar, \heightstar, \tx, \sigvp, \sigoptuple \rtuple \rtuple$ is a Type~I $\APS$ for $\tx$, and the missing elements $\tx$ and $\sigoptuple$ can be obtained from the proposal in \eqref{eq:OciorProposal}. 
    \item Otherwise, $\myproofProposedTx$ is set as follows (Line~\ref{line:GetNewTxNoKLStepBegin} of Algorithm~\ref{algm:OciorProcedures}):  
    \begin{align}
        \myproofProposedTx = \defaultvalue . \label{eq:OciorMyProofProposedTxDefault}
    \end{align}   
\end{itemize}
\end{itemize}
If $\eon>1$ and $\NumProposed =1$, Node~$\thisnodeindex$ proposes the last proposal proposed in the previous epoch (see Lines~\ref{line:OciorProposeLastEA} and \ref{line:OciorProposeLastEB} of Algorithm~\ref{algm:Ocior}).

{\bf \emph{Vote:}} Upon receiving  $\ltuple \PROPOSE, j, \eonstar, \numproposedstar,    \heightstar,  \tx, \sigvp, \sigoptuple, \contentoptuple, \eondiamond, \numproposeddiamond, \proofofPreviousProposedTx$, $\proofProposedTx \rtuple$ from Node~$j$, for $j \in [n]$,  Node~$\thisnodeindex$ executes the steps of the \emph{Vote} round as specified in Lines~\ref{line:OciorVoteBegin}-\ref{line:OciorVoteEnd} of Algorithm~\ref{algm:Ocior}. The main steps are outlined below.  
 \begin{itemize}
    \item  If $\eonstar> \eon$, Node~$\thisnodeindex$	  waits until       $\eon \geq \eonstar$. 
    \item Node~$\thisnodeindex$ waits until the proposal of the $\VP$ transaction has been received (Lines~\ref{line:OciorVoteVPCondition} and \ref{line:OciorVoteVPConditionNeed}).
    \item Node~$\thisnodeindex$ must verify that $\sigvp$ is a valid threshold signature for the $\VP$ transaction before voting on this proposal (Line~\ref{line:OciorVoteVPCheck}). Node~$\thisnodeindex$ must also verify that $\sigvp$ is the \emph{only one} threshold signature accepted at height $\heightstar-1$ of Chain~$j$ (Lines~\ref{line:OciorVoteVPCheckCond}-\ref{line:OciorVoteVPCheck}).  
    \item Node~$\thisnodeindex$ waits until the previous proposals (indexed by $\numproposedstar-1, \numproposedstar-2, \dotsc, 1$) proposed by Node~$j$ during the same Epoch~$\eonstar$ have been received.  
\item Node~$\thisnodeindex$ must verify that all of Node~$j$'s preceding proposals in the same Epoch~$\eonstar$, as well as all proposals voted on by Node~$\thisnodeindex$ for Chain~$j$, are completed before voting on this proposal (Lines~\ref{line:OciorVotePPCheckCond} and~\ref{line:OciorVotePPCheck}).
    \item Node~$\thisnodeindex$ must verify that the transaction  $\tx$  is legitimate  before voting on this proposal (Lines~\ref{line:OciorVoteOPTXCheckBegin}-\ref{line:OciorVoteOPTXCheckEnd}).    
    \item If $\tx$ conflicts with another transaction, Node~$\thisnodeindex$ must verify that $\proofProposedTx$ is a valid $\APS$ proof for $\tx$ before voting on this proposal (Lines~\ref{line:OciorAPSproofTxCheck}-\ref{line:OciorVoteOPTXCheckEnd}). If $\myproofProposedTx$ is a valid $\APS$ for $\tx$, Node~$\thisnodeindex$ can vote for it, even if $\tx$ conflicts with another transaction.
   \end{itemize}
 After that, Node~$\thisnodeindex$ may send a message back to Node~$j$, depending on which conditions are satisfied: 
\begin{itemize}
\item If all of the above conditions are satisfied, then Node~$\thisnodeindex$ sends the following vote message back to Node~$j$:
\begin{align}
\ltuple \VOTE, j, \eonstar, \numproposedstar,  \Vote(\sk_{\eonstar, \thisnodeindex},  \Hash(\content) ), \VoteLTS(\skl_{\eonstar, \thisnodeindexshuffle},  \Hash(\content) )   \rtuple     \label{eq:OciorVoteOK}
\end{align} 
(Line~\ref{line:OciorVoteOK}), where $\content = \ltuple j, \eonstar, \numproposedstar, \heightstar, \tx, \sigvp, \sigoptuple \rtuple$. 	
\item If all of the above conditions are satisfied except that $\tx$ conflicts with another transaction $\txconflict$ recorded by Node~$\thisnodeindex$ (and $\proofProposedTx$ is not a valid  $\APS$ proof for $\tx$), then Node~$\thisnodeindex$ sends the following  message back to Node~$j$ (Line~\ref{line:OciorVoteEnd}): 		
\begin{align}
\ltuple \CONFLICT, j, \eonstar, \numproposedstar, \txconflict \rtuple.    \label{eq:OciorVoteEnd}
\end{align}				
\end{itemize}

{\bf \emph{Proposal Completion:}}  After receiving at least $\TSthreshold = \lceil \frac{n+t+1}{2} \rceil$ messages of valid partial signatures as in~\eqref{eq:OciorVoteOK} from distinct nodes, the proposer (focusing on Node~$i$) generates a threshold signature $\sig$ over the $\content$ of $\tx$ (see Lines~\ref{line:OciorAggRx}-\ref{line:OciorMyProofNormalDoneTS} of Algorithm~\ref{algm:Ocior}). In good cases of partial signature collection, the threshold signature $\sig$ can be generated from valid $\LTS$ partial signatures (Lines~\ref{line:OciorAggRx}-\ref{line:OciorMyProofNormalDoneLTS}).  
Once the threshold signature $\sig$ is generated for the proposed  $\tx$, the proposer sends an $\APS$ message to  $\RPC$ nodes:  
\[
\ltuple \TXDONE, \sigvp, \contentvp, \sig, \content \rtuple
\]  
(Lines~\ref{line:OciorSendAPSLTS} and~\ref{line:OciorSendAPSTS}), where $\contentvp$ is the content of the threshold signature $\sigvp$ for the virtual parent linked to $\tx$. Here, $\ltuple \sig, \content \rtuple$ is a Type~I $\APS$ for $\tx$, and $\ltuple \sigvp, \contentvp, \sig, \content \rtuple$ is a Type~II $\APS$ for the virtual parent linked to $\tx$. 
In this case, the proposal indexed by $(j, \eon, \NumProposed)$ is completed.  
If the proposer receives any message as in~\eqref{eq:OciorVoteEnd} that includes a transaction $\txconflict$ conflicting with the proposed transaction $\tx$, and given $\myproofProposedTx=\defaultvalue$ (see \eqref{eq:OciorMyProofProposedTxDefault} and Line~\ref{line:OciorConfMyproofProposedTXDefault}), then the proposal is also considered completed.   
It is worth noting that if the proposer is honest, and if $\myproofProposedTx \neq \defaultvalue$, then $\myproofProposedTx$ should be a valid $\APS$ for $\tx$, and all honest nodes should vote for it, even if $\tx$ conflicts with another transaction. 
When a proposal is completed, the proposer sets the value of $\myproofofPreviousProposedTx$ accordingly, as in~\eqref{eq:OciorMyProofDefault} and \eqref{eq:OciorMyProofC}, and then proposes the next transaction.

 {\bf \emph{Real-Time and Complete $\APS$ Dissemination:}}  
In $\Ocior$, there are three scenarios that guarantee real-time and complete $\APS$ dissemination to consensus nodes and $\RPC$ nodes: 
\begin{itemize}
    \item \textbf{Dissemination Scenario 1:} If the proposer is honest and generates a valid $\APS$ for a transaction $\tx$, then this $\APS$ is sent to all $\RPC$ nodes \emph{immediately} (Lines~\ref{line:OciorSendAPSLTS} and~\ref{line:OciorSendAPSTS} of Algorithm~\ref{algm:Ocior}).   

    \item \textbf{Dissemination Scenario 2:} If the proposer is dishonest but its chain continues to grow, then the valid $\APS$ generated by this proposer is guaranteed to be propagated in a timely manner through the following mechanism.  
    Specifically, during the \emph{Vote} round on a proposal indexed by $\ltuple j, \eonstar, \numproposedstar, \heightstar\rtuple$ at height $\heightstar$, Node~$\thisnodeindex$ sends an $\APS$ message
    \[
    \ltuple \TXDONE, \sigprime, \contentprime, \sigdiamond, \contentdiamond \rtuple
    \]
    to \emph{one randomly selected} $\RPC$ node, provided this $\APS$ message has not already been sent (Line~\ref{line:OciorAPSIsent} of Algorithm~\ref{algm:Ocior}).  
    Here, $\sigprime$ and $\sigdiamond$ are accepted at heights $\heightstar-2$ and $\heightstar-1$, respectively, of Chain~$j$. Once an $\RPC$ node receives a valid $\APS$, it forwards the message to its connected $\RPC$ nodes via network gossiping.  
    This mechanism ensures that the $\APS$ is propagated \emph{immediately}, even if the proposer fails to disseminate it directly to the $\RPC$ nodes.   

\item \textbf{Dissemination Scenario 3:} If a chain has grown by $\NumHeightMulticast$ new locked heights, the network invokes the $\HMDM$ algorithm to multicast these $\NumHeightMulticast$ signatures and their corresponding contents in the locked chain to all consensus nodes and all $\RPC$ nodes (see Lines~\ref{line:OciorHMDMBegin}-\ref{line:OciorHMDMEndEnd} of Algorithm~\ref{algm:Ocior}).  
If a chain has not grown for $\EpochOutageThreshold$ epochs (for a preset parameter $\EpochOutageThreshold$), the network broadcasts the remaining signatures locked in this chain, together with one additional signature accepted at a height immediately following that of  the top locked signature (see Lines~\ref{line:OciorBroadcastBegin}-\ref{line:OciorHMDMEnd} of Algorithm~\ref{algm:Ocior}).  
\end{itemize}

{\bf \emph{Update of $\TnewIDSetOtherProposed$, $\TnewIDSetOtherProposedLastEpoch$, $\TAWOneIDSetOtherProposed$, and $\TAWOneIDSetOtherProposedLastEpoch$:}}  
In $\Ocior$, at any point in time, each of the sets $\TnewIDSetOtherProposed$ and $\TnewIDSetOtherProposedLastEpoch$ contains at most \emph{three} $\ID$s of transactions proposed at Chain~$j$, while each of the sets $\TAWOneIDSetOtherProposed$ and $\TAWOneIDSetOtherProposedLastEpoch$ contains at most \emph{two} $\ID$s of transactions proposed at Chain~$j$, for every $j \in [n]$. 
These four sets are updated as follows:   
\begin{itemize}
 \item When a node votes for a transaction $\tx$ at height $\height$ of Chain~$j$, it adds the $\ID$ of $\tx$ to $\TnewIDSetOtherProposed$ and adds to $\TAWOneIDSetOtherProposed$ the $\ID$ of a transaction accepted at height $\height-1$ of Chain~$j$. At the same time, it removes from $\TnewIDSetOtherProposed$, $\TnewIDSetOtherProposedLastEpoch$, $\TAWOneIDSetOtherProposed$, and $\TAWOneIDSetOtherProposedLastEpoch$ the $\ID$ of a transaction accepted at height $\height-3$ of Chain~$j$ (see Lines~\ref{line:OciorUpdateTnewIDSetPO}-\ref{line:OciorUpdateTnewIDSetPOLE} of Algorithm~\ref{algm:Ocior}).  

 \item When a proposed transaction $\tx$ conflicts with another transaction, the $\ID$ of $\tx$ is removed from $\TnewIDSetOtherProposed$ and $\TnewIDSetOtherProposedLastEpoch$ (Line~\ref{line:RemoveUpdateTnewIDSetOtherProposed} of Algorithm~\ref{algm:OciorProcedures}).  

 \item When a node proposes a transaction $\tx$, the $\ID$ of $\tx$ is removed from $\TnewIDSetOtherProposed$, $\TnewIDSetOtherProposedLastEpoch$, $\TAWOneIDSetOtherProposed$, and $\TAWOneIDSetOtherProposedLastEpoch$ accordingly (Lines~\ref{line:GetNewTxBegin}-\ref{line:GetNewTxEnd} of Algorithm~\ref{algm:OciorProcedures}).  

 \item At the end of an epoch, the set $\TnewIDSetOtherProposed$ is merged into $\TnewIDSetOtherProposedLastEpoch$, and the set $\TAWOneIDSetOtherProposed$ is merged into $\TAWOneIDSetOtherProposedLastEpoch$ (Line~\ref{line:OciorUpdateEpochTnewIDSetOtherProposedLastEpoch} of Algorithm~\ref{algm:Ocior}).  
\item  Based on the rules for new transaction selection, at the beginning of a new epoch, each node first proposes transactions from $\TnewIDSetOtherProposedLastEpoch$ and $\TAWOneIDSetOtherProposedLastEpoch$ until both sets are empty. 
\end{itemize}  

The $\APS$ dissemination, together with the updates of $\TnewIDSetOtherProposed$, $\TnewIDSetOtherProposedLastEpoch$, $\TAWOneIDSetOtherProposed$, and $\TAWOneIDSetOtherProposedLastEpoch$, guarantees the following two properties:  

\begin{itemize}
    \item In $\Ocior$, if a valid Type~II $\APS$ is generated for a transaction $\tx$, then eventually all honest consensus nodes and $\RPC$ nodes will receive a valid Type~II $\APS$ and accept $\tx$, even if another transaction conflicts with $\tx$ (see Theorem~\ref{thm:APSII}).    
    \item In $\Ocior$, if a valid Type~I $\APS$ is generated for a \emph{legitimate} transaction $\tx$, and $\tx$ remains legitimate, then eventually all honest consensus nodes and $\RPC$ nodes will receive a valid Type~II $\APS$ and accept $\tx$ (see Theorem~\ref{thm:APSI}).  Note that any node that receives a transaction $\tx$ together with a valid Type~I or Type~II $\APS$ will accept $\tx$, even if another transaction conflicts with it (see Theorem~\ref{thm:Safety}). 
\end{itemize}

{\bf \emph{Key Regeneration and Epoch Transition:}} 
In $\Ocior$, before entering a new epoch~$(\eon+1)$, an $\ADKG[\eon+1]$ scheme is executed to generate new threshold signature keys for the $\TS$ and $\LTS$ schemes for the upcoming epoch, as described in Lines~\ref{line:OciorStartSeedGen}-\ref{line:OciorADKGoutput} of Algorithm~\ref{algm:Ocior}.  Specifically, when $\NumProposed=\NumProposedSeed$, Node~$\thisnodeindex$  activates $\SeedGeneration$ protocol with other nodes to generate a random seed (see Lines~\ref{line:OciorStartSeedGen}-\ref{line:OciorStartSeedInput} of Algorithm~\ref{algm:Ocior}, and Algorithm~\ref{algm:OciorNoFPKLSeed}).     The random seed is used to   shuffle the node indices for the $\LTS$ scheme for the next epoch. 
When $\NumProposed = \MaxNumTx$ and $\ADKG[\eon+1]$ outputs the keys for the $\TS$ and $\LTS$ schemes, Node~$\thisnodeindex$ votes for the next epoch by sending $(\EONVOTE, \eon)$ to all nodes (Line~\ref{line:OciorSendEonVote}). After one more round of message exchange (Line~\ref{line:OciorSendEONCONFIRM}), the nodes eventually reach consensus to proceed to the next epoch (Lines~\ref{line:OciorReadyNewEpoch}-\ref{line:OciorVoteNewEpochEnd}).

 \section{Analysis of  $\Ocior$: Safety, Liveness, and Complexities}      \label{sec:OciorAnalysis}
  
Here we provide an analysis of $\Ocior$ with respect to safety, liveness, and its communication, computation, and round complexities.

\begin{theorem}[Safety] \label{thm:Safety}
In $\Ocior$, any two transactions accepted by honest nodes do not conflict. Furthermore, if two valid $\APS$s are generated for two different transactions, then those transactions must also be non-conflicting. 
Finally, any node that receives a transaction together with its valid $\APS$ will accept the transaction. 
\end{theorem}

\begin{proof}
We prove this result by contradiction. Suppose there exist two conflicting transactions, $\tx$ and $\txprime$, that are accepted by honest nodes. We first consider the case where $\tx$ and $\txprime$ conflict with each other and share the same parent. 
Since $\tx$ and $\txprime$ are accepted by honest nodes, there must exist two valid $\APS$s for $\tx$ and $\txprime$, respectively. Let $\sig$ and $\sigprime$ be the threshold signatures included in the $\APS$s for $\tx$ and $\txprime$, respectively.  

To generate the threshold signature $\sig$, at least $\TSthreshold = \lceil \frac{n+t+1}{2} \rceil$ partial signatures from distinct nodes voting (signing) for $\tx$ are required, which implies that at least $\lceil \frac{n+t+1}{2} \rceil - t$ honest nodes have voted for $\tx$. In $\Ocior$, when an honest node votes for $\tx$, it will not vote for any transaction that conflicts with $\tx$ and shares the same parent.

Therefore, if $\sig$ is generated for $\tx$, then the number of partial signatures voting for a conflicting $\txprime$ is at most
\[
n - \Bigl( \Bigl\lceil \frac{n+t+1}{2} \Bigr\rceil - t \Bigr) < \TSthreshold
\]
given the identity $n+t < 2 \bigl(    \frac{n+t+1}{2}    \bigr)  \leq 2 \bigl( \bigl\lceil \frac{n+t+1}{2} \bigr\rceil \bigr) $, for $\TSthreshold = \lceil \frac{n+t+1}{2} \rceil$. 
This implies that the threshold signature $\sigprime$ for $\txprime$ could not have been generated. 
This contradicts our assumption, and thus we conclude that transactions accepted by   honest nodes and sharing the same parent do not conflict with each other.  

Let us now consider the case where $\tx$ and $\txprime$ conflict with each other but do not share a direct parent, i.e., at least one of them is a descendant of a conflicting transaction. When a node votes for $\tx$ (or $\txprime$), it must verify that the proposal includes a valid $\APS$ for each parent of $\tx$ (or $\txprime$). From the above result, at most one valid $\APS$ will be generated for conflicting transactions that share the same parent. This implies that one of the conflicting transactions $\tx$ or $\txprime$ that do not share a direct parent will not receive any votes from honest nodes due to the lack of a valid $\APS$ for its parent.  

Thus, any two transactions accepted by   honest nodes do not conflict. 
Furthermore, the above result implies that if two valid $\APS$s are generated for two different transactions, then those transactions must also be non-conflicting. 
Finally, in $\Ocior$, any node that receives a transaction $\tx$  together with its valid $\APS$ accepts the transaction $\tx$ (see $\Accept$ function in  Lines~\ref{line:AcceptBegin}-\ref{line:AcceptEnd} of Algorithm~\ref{algm:OciorProcedures}).      
\end{proof}

\begin{theorem} [Liveness]  \label{thm:liveness}
In $\Ocior$, if a legitimate transaction $\tx$  is received and proposed by at least one consensus node that remains uncorrupted throughout the protocol, and  $\tx$  remains legitimate, then a valid $\APS$ for $\tx$ is eventually generated, delivered to, and accepted by all honest consensus nodes and all active $\RPC$ nodes.  
\end{theorem} 
\begin{proof}
From Lemma~\ref{lm:livenessEventualProOut},  if a legitimate transaction $\tx$ is proposed by a consensus node that remains uncorrupted throughout the protocol, and $\tx$ remains legitimate, then eventually it will be proposed by at least one honest node with $\AW(\tx) \geq \AWthrehold$, where  $\AWthrehold=3$. 
 From Lemma~\ref{lm:APSIIALL}, if $\AW(\tx) \geq 3$, then eventually all honest consensus nodes and $\RPC$ nodes will receive a valid Type~II $\APS$ and accept $\tx$, even if another transaction conflicts with $\tx$.    
 This completes the proof. 
\end{proof}

\begin{lemma}    \label{lm:livenessEventualPro}
In $\Ocior$, if a transaction $\tx$ has been voted for by an honest node, and this node is guaranteed to propose $\MaxNumTx$ complete transactions in a new epoch, then unless $\tx$ becomes conflicting or $\AW(\tx) \geq \AWthrehold$, this node will eventually propose $\tx$ in its chain with $\AW(\tx)\geq \AWthrehold$, where the parameter $\MaxNumTx$ is typically set as $\MaxNumTx \geq O(n^2)$ and   $\AWthrehold=3$.
\end{lemma}

\begin{proof}
When an honest node votes for a transaction $\tx$ at height $\height$ of Chain~$j$, it adds the $\ID$ of $\tx$ to $\TnewIDSetOtherProposed$. At the same time, it removes from $\TnewIDSetOtherProposed$ and $\TnewIDSetOtherProposedLastEpoch$ the $\ID$ of the transaction accepted at height $\height-3$ of Chain~$j$ (see Lines~\ref{line:OciorUpdateTnewIDSetPO}--\ref{line:OciorUpdateTnewIDSetPOLE} of Algorithm~\ref{algm:Ocior}).  

If a proposed transaction $\tx$ conflicts with another transaction, then its $\ID$ is removed from both $\TnewIDSetOtherProposed$ and $\TnewIDSetOtherProposedLastEpoch$ (Line~\ref{line:RemoveUpdateTnewIDSetOtherProposed} of Algorithm~\ref{algm:OciorProcedures}).  

At any point in time, each of the sets $\TnewIDSetOtherProposed$ and $\TnewIDSetOtherProposedLastEpoch$ contains at most \emph{three} $\ID$s of transactions proposed at Chain~$j$, while each of the sets $\TAWOneIDSetOtherProposed$ and $\TAWOneIDSetOtherProposedLastEpoch$ contains at most \emph{two} $\ID$s of transactions proposed at Chain~$j$, for every $j \in [n]$.

At the end of an epoch, $\TnewIDSetOtherProposed$ is merged into $\TnewIDSetOtherProposedLastEpoch$, and $\TAWOneIDSetOtherProposed$ is merged into $\TAWOneIDSetOtherProposedLastEpoch$ (Line~\ref{line:OciorUpdateEpochTnewIDSetOtherProposedLastEpoch} of Algorithm~\ref{algm:Ocior}).  
According to the rules for new transaction selection, each node first selects transactions from $\TnewIDSetOtherProposedLastEpoch$ and $\TAWOneIDSetOtherProposedLastEpoch$ until both sets are empty.  

Therefore, if a transaction $\tx$ has been voted for by a node, and this node is guaranteed to propose $\MaxNumTx$ complete transactions in a new epoch, then unless $\tx$ becomes conflicting or $\AW(\tx)\geq \AWthrehold$, the node will eventually propose $\tx$ in its chain with $\AW(\tx)\geq \AWthrehold$ at the beginning of a new epoch, and subsequently remove the $\ID$ of $\tx$ from $\TnewIDSetOtherProposed$ and $\TnewIDSetOtherProposedLastEpoch$ (Lines~\ref{line:GetNewTxBegin}--\ref{line:GetNewTxEnd} of Algorithm~\ref{algm:OciorProcedures}).  
\end{proof}

 \begin{lemma}    \label{lm:livenessEventualProtandone}
In $\Ocior$, if a valid Type~I $\APS$ is generated for a \emph{legitimate} transaction $\tx$, and $\tx$ remains legitimate, then eventually it will be proposed by at least one honest node with $\AW(\tx) \geq \AWthrehold$.
\end{lemma}

\begin{proof} 
When a valid Type~I $\APS$ is generated for a \emph{legitimate} transaction $\tx$, at least 
\[
\TSthreshold - |\Fc| \geq \TSthreshold - t = \big\lceil \tfrac{n+t+1}{2} \big\rceil - t
\]
honest nodes must have voted for $\tx$, where $\Fc$ denotes the set of dishonest nodes with $|\Fc|\leq t$, and where the parameter $\TSthreshold=\lceil (n+t+1)/2 \rceil$ is the threshold for the threshold signature.  

Furthermore, all honest nodes, except for at most $t$ of them, are guaranteed to propose $\MaxNumTx$ complete transactions in each epoch, where typically $\MaxNumTx \geq O(n^2)$.  
Note that an honest consensus node enters a new epoch only if it has received confirmation messages from at least $n-t$ distinct nodes, each confirming that they have proposed $\MaxNumTx$ complete transactions.  

Therefore, when a valid Type~I $\APS$ is generated for a \emph{legitimate} transaction $\tx$, at least
\[
\TSthreshold - |\Fc| - t \geq \big\lceil \tfrac{n+t+1}{2} \big\rceil - 2t \geq \tfrac{3t+1+t+1}{2} - 2t > 1
\]
honest node must have voted for $\tx$ and is also guaranteed to propose $\MaxNumTx$ complete transactions in a new epoch.  

By Lemma~\ref{lm:livenessEventualPro}, if a transaction $\tx$ has been voted for by an honest node, and this node is guaranteed to propose $\MaxNumTx$ complete transactions in a new epoch, then unless $\tx$ becomes conflicting or $\AW(\tx)\geq \AWthrehold$, this node will eventually propose $\tx$ in its chain $\AW(\tx) \geq \AWthrehold$.  

Thus, if a valid Type~I $\APS$ is generated for a \emph{legitimate} transaction $\tx$, and $\tx$ remains legitimate, then eventually it will be proposed by at least one honest node with $\AW(\tx)\geq \AWthrehold$. 
\end{proof}

\begin{lemma}    \label{lm:livenessEventualProOut}
In $\Ocior$, if a legitimate transaction $\tx$ is proposed by a consensus node that remains uncorrupted throughout the protocol, and $\tx$ remains legitimate, then eventually it will be proposed by at least one honest node with $\AW(\tx) \geq \AWthrehold$.
\end{lemma}

\begin{proof} 
If a legitimate transaction $\tx$ is proposed by a consensus node that remains uncorrupted throughout the protocol, and $\tx$ continues to be legitimate, then there are two possible cases:  
1) the proposal of $\tx$ is completed with a valid Type~I $\APS$; or  
2) the proposal is the last proposal made by this consensus node.  

For the first case, by Lemma~\ref{lm:livenessEventualProtandone}, if a valid Type~I $\APS$ is generated for a legitimate transaction $\tx$, and $\tx$ remains legitimate, then eventually it will be proposed by at least one honest node with $\AW(\tx) \geq \AWthrehold$.  

We now turn to the second case.  
If a legitimate transaction $\tx$ is proposed by a consensus node that remains uncorrupted throughout the protocol, $\tx$ remains legitimate, and this proposal is the last proposal made by that consensus node, then every honest node eventually adds the $\ID$ of $\tx$ to the set $\TnewIDSetOtherProposed$ (see Line~\ref{line:OciorUpdateTnewIDSetPO} of Algorithm~\ref{algm:OciorProcedures}).

As noted in the proof of Lemma~\ref{lm:livenessEventualPro}, at any point in time, each of the sets $\TnewIDSetOtherProposed$ and $\TnewIDSetOtherProposedLastEpoch$ contains at most \emph{three} $\ID$s of transactions proposed at Chain~$j$, while each of the sets $\TAWOneIDSetOtherProposed$ and $\TAWOneIDSetOtherProposedLastEpoch$ contains at most \emph{two} $\ID$s of transactions proposed at Chain~$j$, for every $j \in [n]$.  

At the end of an epoch, $\TnewIDSetOtherProposed$ is merged into $\TnewIDSetOtherProposedLastEpoch$, and $\TAWOneIDSetOtherProposed$ is merged into $\TAWOneIDSetOtherProposedLastEpoch$ (Line~\ref{line:OciorUpdateEpochTnewIDSetOtherProposedLastEpoch} of Algorithm~\ref{algm:Ocior}).  
According to the rules for new transaction selection, each node first selects transactions from $\TAWOneIDSetOtherProposedLastEpoch$ and $\TnewIDSetOtherProposedLastEpoch$ until both sets are empty.  

As noted in the proof of Lemma~\ref{lm:livenessEventualProtandone}, all honest nodes, except for at most $t$ of them, are guaranteed to propose $\MaxNumTx$ complete transactions in each epoch.  
If the $\ID$ of $\tx$ has been added to $\TnewIDSetOtherProposedLastEpoch$ by an honest node that is guaranteed to propose $\MaxNumTx$ complete transactions in a new epoch, then unless $\tx$ becomes conflicting or $\AW(\tx)\geq \AWthrehold$, this node will eventually propose $\tx$ in its chain with $\AW(\tx) \geq \AWthrehold$.  
\end{proof}

\begin{theorem}  [Type~II $\APS$]  \label{thm:APSII}
In $\Ocior$, if a valid Type~II $\APS$ is generated for a transaction $\tx$, then eventually all honest consensus nodes and active $\RPC$ nodes will receive a valid Type~II $\APS$ and accept $\tx$, even if another transaction conflicts with $\tx$. 
\end{theorem}
\begin{proof}
When a valid Type~II $\APS$ has been generated for a transaction $\tx$, at least $t+1$ honest nodes must have accepted $\tx$ with a Type~I $\APS$. 
When an honest Node~$i$ accepts $\tx$ at height $\height-1$ of Chain~$j$ (during a \emph{vote} round for a proposal at height $\height$ proposed by Node~$j$), Node~$i$ adds the $\ID$ of $\tx$ into $\TAWOneIDSetOtherProposed$, and removes from $\TAWOneIDSetOtherProposed$ and $\TAWOneIDSetOtherProposedLastEpoch$ the $\ID$ of the transaction accepted at height $\height-3$ of Chain~$j$ (see Line~\ref{line:OciorUpdateTAWOneIDSetPO} of Algorithm~\ref{algm:Ocior}). 
The number of transaction identities included in each of $\TAWOneIDSetOtherProposed$ and $\TAWOneIDSetOtherProposedLastEpoch$ for Chain~$j$ is always bounded by $2$.  

At the end of Epoch~$\eon$, the set $\TAWOneIDSetOtherProposed$ is merged into $\TAWOneIDSetOtherProposedLastEpoch$ (Line~\ref{line:OciorUpdateEpochTnewIDSetOtherProposedLastEpoch} of Algorithm~\ref{algm:Ocior}). 
According to the rules for new transaction selection, at the beginning of Epoch~$\eon+1$, Node~$i$ first proposes transactions from $\TnewIDSetOtherProposedLastEpoch$ and $\TAWOneIDSetOtherProposedLastEpoch$, until both sets becomes empty (Lines~\ref{line:GetNewTxBegin}-\ref{line:GetNewTxEnd} of Algorithm~\ref{algm:OciorProcedures}).    

When honest Node~$i$ reproposes an accepted transaction $\tx$, it attaches a valid $\APS$ for $\tx$ (Line~\ref{line:proofProposedTxvalue} of Algorithm~\ref{algm:OciorProcedures}). 
Consequently, all honest nodes will eventually accept $\tx$ and vote for it again, even if another transaction conflicts with $\tx$ (see Lines~\ref{line:OciorAPSproofTxCheck}-\ref{line:OciorVoteOPTXCheckEnd} of Algorithm~\ref{algm:Ocior}).  

Moreover, among the $t+1$ honest nodes that have accepted $\tx$ with a Type~I $\APS$, at least $t+1-t=1$ honest node must both have accepted $\tx$ and be guaranteed to propose $\MaxNumTx$ complete transactions in a new epoch (as in the proof of Lemma~\ref{lm:livenessEventualProOut}).  

Therefore, from these $t+1$ honest nodes, unless $\AW(\tx) \geq 3$, at least one honest node that is guaranteed to propose $\MaxNumTx$ complete transactions in a new epoch will eventually repropose $\tx$ with $\AW(\tx) \geq 3$.  
From Lemma~\ref{lm:APSIIALL}, if $\AW(\tx) \geq 3$, then eventually all honest consensus nodes and $\RPC$ nodes will receive a valid Type~II $\APS$ and accept $\tx$, even if another transaction conflicts with $\tx$.  

Thus, if a valid Type~II $\APS$ is generated for a transaction $\tx$, then eventually all honest consensus nodes and $\RPC$ nodes will receive a valid Type~II $\APS$ and accept $\tx$, even if another transaction conflicts with $\tx$. 
\end{proof}

\begin{theorem}[Type~I $\APS$]    \label{thm:APSI}
In $\Ocior$, if a valid Type~I $\APS$ is generated for a \emph{legitimate} transaction $\tx$, and $\tx$ remains legitimate, then eventually all honest consensus nodes and $\RPC$ nodes will receive a valid Type~II $\APS$ and accept $\tx$. 
\end{theorem}
\begin{proof}
From Lemma~\ref{lm:livenessEventualProtandone},    if a valid Type~I $\APS$ is generated for a \emph{legitimate} transaction $\tx$, and $\tx$ remains legitimate, then eventually it will be proposed by at least one honest node with $\AW(\tx) \geq \AWthrehold$,  which implies that a valid Type~II $\APS$ is eventually generated for $\tx$. 
Then, from Theorem~\ref{thm:APSII}, if a valid Type~II $\APS$ is generated for a transaction $\tx$, then eventually all honest consensus nodes and $\RPC$ nodes will receive a valid Type~II $\APS$ and accept $\tx$. 
\end{proof}

\begin{lemma}    \label{lm:APSIIALL}
If $\AW(\tx) \geq 3$, then eventually all honest consensus nodes and $\RPC$ nodes will receive a valid Type~II $\APS$ for $\tx$ and accept $\tx$, even if another transaction conflicts with $\tx$. 
\end{lemma}
\begin{proof}
If $\AW(\tx) \geq 3$, then at least $\TSthreshold - |\Fc| \geq \big\lceil \tfrac{n+t+1}{2} \big\rceil - t \geq t+1$  
honest consensus nodes have accepted $\tx$ with $\AW(\tx) \geq 2$ in a locked chain. 
We say that Chain~$j$ is \emph{locked} at height $\heightstar$ at Node~$i$ if all signatures at heights $\heightprime \leq \heightstar+1$ have been accepted in Chain~$j$.  
From Lemma~\ref{lm:lockchain}, if Chain~$j$ is locked at height $\heightstar$ at both Node~$i$ and Node~$i'$, then Nodes~$i$ and $i'$ hold the same copy of the locked Chain~$j$.  

If a chain grows by $\NumHeightMulticast$ new locked heights, the $\HMDM$ algorithm is invoked to multicast these $\NumHeightMulticast$ signatures and their corresponding contents in the locked chain to all consensus nodes and all $\RPC$ nodes (see Lines~\ref{line:OciorHMDMBegin}--\ref{line:OciorHMDMEndEnd} of Algorithm~\ref{algm:Ocior}).  
If a chain has not grown for $\EpochOutageThreshold$ epochs (for a preset parameter $\EpochOutageThreshold$), the network broadcasts the remaining signatures locked in the chain, along with one additional signature accepted at a height immediately following that of the top locked signature (see Lines~\ref{line:OciorBroadcastBegin}-\ref{line:OciorHMDMEnd} of Algorithm~\ref{algm:Ocior}).  
The broadcast of the last signature ensures that every node receives the signatures needed to construct a valid Type~II $\APS$ for $\tx$ whenever $\AW(\tx) \geq 3$.

Therefore, if $\AW(\tx) \geq 3$, then eventually all honest consensus nodes and $\RPC$ nodes will receive a valid Type~II $\APS$ for $\tx$ and accept $\tx$, even if another transaction conflicts with $\tx$.  
\end{proof}

\begin{lemma}    \label{lm:lockchain}
If Chain~$j$ is locked at height $\heightstar$ at both honest Node~$i$ and honest Node~$i'$, then Nodes~$i$ and $i'$ hold the same copy of the locked Chain~$j$.  
\end{lemma}

\begin{proof}
When a node votes for a proposal at height $\heightprime$ of Chain~$j$, it must verify that the proposal links to a fixed $\VP$ accepted at height $\heightprime - 1$, and that all signatures at heights $\heightprime - 1, \heightprime - 2, \dots, 1$ have been accepted.  
Due to the network quorum, if a signature at height $\heightprime$ is generated, it can only link to a single $\VP$ accepted at height $\heightprime - 1$.  

If Chain~$j$ is \emph{locked} at height $\heightstar$ at Node~$i$, then all signatures at heights $\heightprime \leq \heightstar+1$ have been accepted in Chain~$j$.  
Therefore, if Chain~$j$ is locked at height $\heightstar$ at both Node~$i$ and Node~$i'$, then Nodes~$i$ and $i'$ must hold the same copy of the locked Chain~$j$.  
\end{proof}

\begin{lemma}    \label{lm:complete}
If a node is guaranteed to propose $\MaxNumTx$ transactions in an epoch $\eon$, then each of its proposals is eventually completed.   
\end{lemma}

\begin{proof}
In $\Ocior$, at least $n-t-|\Fc| \geq t+1$ honest nodes are guaranteed to propose $\MaxNumTx$ transactions in an epoch $\eon$. 
For each such node, and for each proposal it makes in epoch $\eon$, the proposal is eventually either accompanied by a threshold signature for the proposed transaction (see Lines~\ref{line:OciorAggRx}-\ref{line:OciorMyProofNormalDoneTS} of Algorithm~\ref{algm:Ocior}) or by a proof showing that the proposed transaction conflicts with another transaction (see Lines~\ref{line:OciorMyProofCBegin}-\ref{line:OciorMyProofC} of Algorithm~\ref{algm:Ocior}).  
In either case, the proposal is considered completed, allowing the node to attach a completeness proof and subsequently propose a new transaction in epoch $\eon$.  

Note that for any node that is not guaranteed to propose $\MaxNumTx$ transactions in epoch $\eon$, each of its proposals is also eventually completed, except possibly the last proposal if the node transitions to a new epoch $\eon+1$.  
In this case, the last proposal from epoch $\eon$ will be re-proposed at the beginning of epoch $\eon+1$.  
\end{proof}

\begin{lemma}    \label{lm:livenessReceiver}
For a two-party transaction $\Tx_{A,B}$, if the recipient $B$ receives a valid $\APS$ for $\Tx_{A,B}$, then $B$ can initiate a new legitimate transaction $\Tx_{B,*}$ to successfully transfer the assets received in $\Tx_{A,B}$. 
\end{lemma}

\begin{proof}
From Theorem~\ref{thm:Safety},  if two valid $\APS$s are generated for two different transactions, then those transactions must be non-conflicting. Furthermore, any node that receives a transaction together with its valid $\APS$ will accept the transaction.  

Therefore, for a two-party transaction $\Tx_{A,B}$, if the recipient $B$ receives a valid $\APS$ for $\Tx_{A,B}$, then $B$ can initiate a new legitimate transaction $\Tx_{B,*}$ to transfer the assets obtained in $\Tx_{A,B}$ by attaching the $\APS$ of $\Tx_{A,B}$ to $\Tx_{B,*}$.  
Any node that receives the $\APS$ for $\Tx_{A,B}$ will accept $\Tx_{A,B}$ and subsequently vote for $\Tx_{B,*}$, even if another transaction conflicts with $\Tx_{A,B}$.
\end{proof}

 \begin{theorem}  [Round complexity]   \label{thm:goodcaseround}
 $\Ocior$ achieves a  good-case latency of \emph{two} asynchronous rounds for Type~I transactions. 
\end{theorem}
\begin{proof} 
For a  legitimate \emph{two-party} (Type~I) transaction, a valid Type~I $\APS$ can be generated by an honest node after the \emph{propose} and \emph{vote} rounds. Thus, it can be finalized with a \emph{good-case latency} of \emph{two} asynchronous rounds,   for any $\ResilianceCondition$. The \emph{good case} in terms of latency  refers to the scenario where the transaction is proposed by any (not necessarily designated) honest node. 
\end{proof}

\begin{theorem}[Communication complexity] \label{thm:cc}
In $\Ocior$, the total expected message complexity per transaction is $O(n)$, and the total expected communication in bits per transaction is $O(n \kappa)$, where $\kappa$ denotes the size of a threshold signature.
\end{theorem}

\begin{proof}
For each proposal made by an honest node, the total message complexity is $O(n)$, while the total communication in bits is $O(n \kappa)$, where $\kappa$ denotes the size of a threshold signature.

It is guaranteed that, when proposing a new transaction, each honest node selects a transaction different from those proposed by other honest nodes with constant probability, provided that the pool of pending transactions is sufficiently large.  
More precisely, according to the transaction selection rules, after choosing transactions from  $\TAWOneIDSetOtherProposedLastEpoch$ and $\TnewIDSetOtherProposedLastEpoch$ (each of size at most $O(n)$), each node  selects a new transaction $\tx$ with $\AW(\tx)<3$ from $\TnewselfQueue$ with probability \[1/\NumProposedIntevalRandomTxSelf\] where $\NumProposedIntevalRandomTxSelf>1$ is a preset parameter (e.g.,   $\NumProposedIntevalRandomTxSelf = 2$). Here, $\TnewselfQueue$ is a FIFO queue maintained by node~$\thisnodeindex$, containing the identities of pending transactions within Cluster~$\thisnodeindex$.  
In addition, each honest node  selects a new transaction $\tx$ with $\AW(\tx)<3$ from $\TnewIDSet$ with probability  
\[
\Biggl(1 - \frac{1}{\NumProposedIntevalRandomTxSelf}\Biggr) \cdot \Biggl(1 - \frac{1}{\NumProposedIntevalRandomTxFromOtherProp}\Biggr)
\]
where $\TnewIDSet$ is a randomized set containing the identities of pending transactions, and the parameter    $\NumProposedIntevalRandomTxFromOtherProp > 1$  can be set as   $\NumProposedIntevalRandomTxFromOtherProp = \lceil n/10 \rceil$.

When a Chain~$j$ has grown $\NumHeightMulticast$ new locked heights, the network activates the $\OciorHMDMHash$ protocol to multicast these $\NumHeightMulticast$ signatures and contents locked in the chain, where $\NumHeightMulticast \geq \lceil n \log n \rceil$. The total message complexity of $\OciorHMDMHash$ is $O(n^2)$, while the total communication in bits is $O(n \NumHeightMulticast \kappa + \kappa n^2 \log n)$ for multicasting $\NumHeightMulticast$ signatures and contents. This implies that, on average, for multicasting one signature and its content, the total message complexity is at most $O(n/\log n)$, and the total communication in bits is  $O(n \kappa)$.  

If a chain has not grown for $\EpochOutageThreshold$ epochs, with a preset parameter $\EpochOutageThreshold = n$, the network broadcasts the remaining signatures locked in this chain, together with one additional signature accepted at a height immediately following that of the top locked signature (see Lines~\ref{line:OciorBroadcastBegin}-\ref{line:OciorHMDMEnd} of Algorithm~\ref{algm:Ocior}). By setting a sufficiently large parameter $\EpochOutageThreshold = n$, and given that the maximum number of transactions that can be proposed by a node in an epoch is $\MaxNumTx > n^2$, this cost is negligible in the overall communication complexity.  

There is also a communication cost incurred from the $\ADKG$ scheme in each epoch. However, by setting the parameter $\MaxNumTx > n^2$, this cost becomes negligible in the overall communication complexity.  

Thus, the total expected message complexity per transaction is $O(n)$, and the total expected communication in bits per transaction is $O(n \kappa)$.
\end{proof}

\begin{theorem}[Computation complexity] \label{thm:cp}
In $\Ocior$, the total computation per transaction is $O(n)$ in the best case, and $O(n \log^2 n)$ in the worst case. 
\end{theorem}

\begin{proof}
In $\Ocior$, for each proposal made by an honest node, the total computation is $O(n)$ in the best case, and $O(n \log^2 n)$ in the worst case, dominated by signature aggregation.  
The computation cost is measured in units of cryptographic operations, including signing, signature verification, hashing, and basic arithmetic operations (addition, subtraction, multiplication, and division) on values of signature size.
The best case is achieved with the $\LTS$ scheme, based on the proposed $\OciorBLS$.  

As in the case of communication complexity, the computation costs incurred by $\OciorHMDMHash$, signature broadcasting, and the $\ADKG$ scheme in each epoch are negligible in the overall computation complexity, provided sufficiently large parameters are set, i.e., $\NumHeightMulticast \geq \lceil n \log n \rceil$, $\EpochOutageThreshold = n$, and $\MaxNumTx > n^2$.  
\end{proof}

\subsection{The case of $\AWthrehold=2$}  

In the above analysis, we focused on the default case of $\AWthrehold=3$. The following two theorems are derived under the specific case of $\AWthrehold=2$. 

\begin{theorem}[Type~II $\APS$, $\AWthrehold=2$]  \label{thm:APSIIthreshold2}
Given $\AWthrehold=2$, if a valid Type~II $\APS$ is generated for a transaction $\tx$, then eventually all honest consensus nodes and active $\RPC$ nodes will receive a valid Type~I $\APS$ and accept $\tx$, even if another transaction conflicts with $\tx$. 
\end{theorem}
\begin{proof}
Suppose a valid Type~II $\APS$ has been generated for a transaction $\tx$ on Chain~$j$ at height $\heightstar$, with $\AW(\tx)\geq 2$. Then, at least $t+1$ honest nodes must have accepted $\tx$ with a Type~I $\APS$ and have linked the threshold signature $\sig_{j,\heightstar}$ for $\tx$ as a $\VP$ at height $\heightstar$. Note that an honest node links only one fixed signature as a $\VP$ at a given height of a chain. Thus, due to quorum, the signature $\sig_{j,\heightstar}$ for $\tx$ is locked at height $\heightstar$ of Chain~$j$. 

From Lemma~\ref{lm:lockchain}, if Chain~$j$ is locked at height $\heightstar$ at both honest Node~$i$ and honest Node~$i'$, then Nodes~$i$ and $i'$ hold the same copy of the locked Chain~$j$.  

Similar to the proof of Lemma~\ref{lm:APSIIALL}, if a chain grows by $\NumHeightMulticast$ new locked heights, the $\HMDM$ algorithm is invoked to multicast these $\NumHeightMulticast$ signatures and their corresponding contents in the locked chain to all consensus nodes and all $\RPC$ nodes. 
If a chain does not grow for $\EpochOutageThreshold$ epochs (for a preset parameter $\EpochOutageThreshold$), the network broadcasts the remaining signatures locked in the chain, along with one additional signature accepted at the height immediately following that of the top locked signature. 
The broadcast of this last signature ensures that every node receives the signatures for $\tx$ whenever $\AW(\tx) \geq 2$.

Thus, given $\AWthrehold=2$, if a valid Type~II $\APS$ is generated for a transaction $\tx$, then eventually all honest consensus nodes and active $\RPC$ nodes will receive a valid Type~I $\APS$ and accept $\tx$, even if another transaction conflicts with $\tx$. 
\end{proof}

\begin{theorem}[Type~I $\APS$, $\AWthrehold=2$] \label{thm:APSIthreshold2}
Given $\AWthrehold=2$, if a valid Type~I $\APS$ is generated for a \emph{legitimate} transaction $\tx$, and $\tx$ remains legitimate, then eventually all honest consensus nodes and $\RPC$ nodes will receive a valid Type~I $\APS$ and accept $\tx$. 
\end{theorem}
\begin{proof}
From Lemma~\ref{lm:livenessEventualProtandone}, if a valid Type~I $\APS$ is generated for a \emph{legitimate} transaction $\tx$, and $\tx$ remains legitimate, then eventually it will be proposed by at least one honest node with $\AW(\tx) \geq \AWthrehold$, which implies that a valid Type~II $\APS$ is eventually generated for $\tx$. 
Then, given $\AWthrehold=2$, by Theorem~\ref{thm:APSIIthreshold2}, if a valid Type~II $\APS$ is generated for a transaction $\tx$, then eventually all honest consensus nodes and $\RPC$ nodes will receive a valid Type~I $\APS$ and accept $\tx$. 
\end{proof}

 {\renewcommand{\arraystretch}{0.6}
\begin{table} [htbp]
\footnotesize
\begin{center}
\caption{{\small Some notations for the proposed $\Ocior$ protocol.}} \label{tb:OciorNotation}
\vspace{-8pt}
\begin{tabu}{c|c}
\toprule
 
Notations & Interpretation \\
\midrule
\midrule
 
	 $n$  & The total number of consensus nodes in the network. \\ 
	 \midrule 
		 $t$  & The maximum number of corrupted nodes that can be tolerated in the network, for $\tResilianceCondition$.  \\ 
	 \midrule 
		 $\TSthreshold$  & The threshold of the threshold signature scheme, defined as 	$\TSthreshold=\lceil \frac{n+t+1}{2} \rceil$.  \\
		 
\midrule		 
 $\MaxNumTx$  & The maximum number of transactions that can be proposed by a node  in an epoch, typically    $\MaxNumTx > n^2$. \\ 
\midrule
 $\NumProposed$ & The number of transactions that have been proposed by  a node as a proposer in an epoch. \\ 
\midrule

 $\height$  & A  height $\height$ (not necessarily accepted yet) of  a   transaction chain.   \\ 
 \midrule
 $\eon$ &  The epoch number.  \\       
 \midrule 
 $\AWthrehold$ & A threshold on the acceptance weight, with $\AWthrehold = 2$ for Type~I transactions and $\AWthrehold = 3$ for Type~II transactions. \\       
              & For simplicity and consistency, we set $\AWthrehold = 3$ for all transactions in the protocol description. \\
\midrule
 $\NumHeightMulticast$ & An interval on a chain for multicasting   signatures and contents locked in the chain,    for $\NumHeightMulticast  \geq     \lceil n \log n \rceil$.  \\       
\midrule
 $\EpochOutageThreshold$ &   If a chain is not growing for  $\EpochOutageThreshold$ epochs,   broadcast   remaining  signatures  locked in this chain, for  $\EpochOutageThreshold\!=\!n$.  \\   

\midrule
   $\tx$ & A transaction. \\ 
\midrule
  $\idtx$ & The transaction ID, which is the hash output of the transaction $\tx$, i.e., $\idtx=\HashZ(\tx)$. \\ 
\midrule

Trans. Cluster~$\thisnodeindex$   &    A transaction $\tx$ is said to belong to    Cluster~$\thisnodeindex$  if   $\HashZ(\tx) \mod n = \thisnodeindex -1$.  \\ 
  
  \midrule

 $\Tx_{A,B}$          & A transaction made from Client~$A$ to Client~$B$,  \\ & where $A$ and $B$ denote the transaction addresses of  clients.   \\
\midrule

  $\OP$ &   Each transaction has one or multiple official parent  ($\OP$) transactions, except for the genesis transactions.    \\	

   &   Genesis transactions were accepted by all nodes either initially or at specific events according to policy. 	  \\
\midrule
  $\VP$ &   Each proposal  for a transaction $\tx$ at height $\height$  is linked to a previous  transaction $\txprime$  \\ 
    & accepted at height $\height-1$ on the same chain.    $\txprime$ is called   the virtual parent $(\VP)$ of $\tx$.  \\
\midrule

$\sig$   &  A threshold signature of a transaction.\\ 
 
$\sigop$   &  A threshold signature of an official parent transaction. \\ 
 
$\sigvp$   &  A threshold signature of a virtual parent transaction. \\ 
\midrule
 $\TxHeightAcceptTemp$  &  $\TxHeightAcceptTemp[\ltuple j, \heightdiamond\rtuple]$   records one and only one $\sig$ at a given height $\heightdiamond$ of  Chain~$j$.  \\ 
\midrule  	
$\TxHeightAcceptLock$  &    $\TxHeightAcceptLock[\ltuple j, \heightdiamond\rtuple]$ records  a  $\sig$  locked at a   height $\heightdiamond$ (with  $\AW\geq 2$), \\    
         & and all previous heights were already locked. \\ 
\midrule   	
$\heightdict$ &  $\heightdict[j]$ denotes the top accepted  height of Chain~$j$.   \\ 
\midrule			
$\heightdictLock$   &  $\heightdictLock[j]=\heightdiamond$ denotes the top locked  height of Chain~$j$, \\   &  meaning that heights $\heightdiamond, \heightdiamond-1, \dotsc, 1$ of Chain~$j$  are locked.  		 \\  	
\midrule
 $\AW$   &   If  $\tx$ is accepted at height $\heightstar$ of Chain~$j$, and   $\heightdict[j]$ denotes the top accepted  height of Chain~$j$,  \\ & then the acceptance weight ($\AW$) of $\tx$ is at least $\heightdict[j]-\heightstar +1$.  \\ 
\midrule

$\NumTxProposedLock$ &   $\NumTxProposedLock[\eonstar][j]$ denotes the  largest  index of  locked proposals   of Chain~$j$ at Epoch~$\eonstar$.  \\ & A proposal is said to be locked if it has been processed and passed the check, \\ & and all preceding proposals in the same epoch and all voted proposals  of Chain~$j$ have also been locked.   \\ 	
\midrule
$\NumTxProposedVoted $    &   $\NumTxProposedVoted[j] \!=\! [ \eonstar, \numproposedstar,    \tx]$ records the information of recent voted proposal of  Chain~$j$.    \\

\midrule
$\SigAccept $    &  A dictionary containing a set of entries of the form  $\{\sig:  \content \}$   for   accepted transactions.   \\ 

\midrule

$\TxAccept $    &  A dictionary containing a set of entries of the form  $\{\idtx:   [\tx, \sig, \sig', \dotsc ] \}$   for   accepted transactions.   \\ 
 
\midrule
 
$\WeightTxTwo $    &  A dictionary containing a set of entries of the form  $\{\idtx:   [j, \heightdiamond, \sigdiamond, \sig ] \}$     with    $\AW \geq 2$.    \\ 

\midrule
$\WeightTxThree $    &  A dictionary containing a set of entries of the form  $\{\idtx:   [j, \heightprime, \sigprime, \sigdiamond, \sig ] \}$    with  $\AW \geq 3$.    \\ 

\midrule

$\Tproposal $    &  A dictionary containing a set of entries of the form  $\{\ltuple j,\eonstar, \numproposedstar\rtuple:   [\contenthash, \content] \}$,  
\\    & where each entry represents a proposal from Node~$j$ at Epoch~$\eonstar$ with index $\numproposedstar$.     \\ 

\midrule

$\PProof $    &  A dictionary containing a set of entries of the form  $\{\ltuple j,\eonstar, \numproposedstar\rtuple :   \ltuple j, \eonstar, \numproposedstar, \eondiamond, \numproposeddiamond,\sigvp,  \proofofPreviousProposedTx  \rtuple\}$,  
\\    & where each entry records a proof $\proofofPreviousProposedTx$ included in the proposal from Node~$j$ at Epoch~$\eonstar$ with index $\numproposedstar$.  \\

 \midrule
 $\Txochildren$ & A dictionary containing a set of entries of the form $\{\ltuple \idop, \opindex\rtuple: [\idtx, \tx]\}$, \\
& where each entry represents a transaction ID $\idop$ with recipient index $\opindex$ that has been validly  \\
& linked  as an official parent by a transaction $\tx$ with ID $\idtx$. \\

 \midrule

\end{tabu}
\end{center}
\end{table}
}

 {\renewcommand{\arraystretch}{0.6}
\begin{table} [htbp]
\footnotesize
\begin{center}
\caption{{\small Some notations for the proposed $\Ocior$ protocol (Continued). }} \label{tb:OciorNotationContinued}
\vspace{-8pt}
\begin{tabu}{c|c}
\toprule
 
Notations & Interpretation \\
\midrule
\midrule

  $\TnewTxDictionary $    &  A dictionary containing    pending transactions.   \\ 

\midrule

$\TnewselfQueue $    &  A  FIFO queue containing IDs of pending transactions  within Cluster~$\thisnodeindex$, maintained by    this node~$\thisnodeindex$.   \\ 

\midrule

$\TxConflictingDic$ & A dictionary containing transactions that conflict with other transactions. \\

\midrule
 	
 $\TproposedIDSet $    &  A  set containing IDs of  transactions that have been proposed by this    node~$\thisnodeindex$.   \\ 
  $\TAWOnePOProposedIDSet $    &  A  set containing IDs of  accepted transactions that have been proposed by this    node~$\thisnodeindex$.   \\ 
 
\midrule
$\TnewIDSet $    &  A  randomized set containing IDs of  pending transactions.   \\ 
 
$\TnewIDSetOtherProposed $    &  A randomized set of IDs of   pending transactions proposed by other nodes.   \\ 
$\TnewIDSetOtherProposedLastEpoch $    &  A randomized set of IDs of   pending transactions proposed by other nodes in the last epoch.   \\ 

$\TAWOneIDSetOtherProposed $    &  A randomized set of IDs of   accepted transactions ($1\!\leq \!\AW\!<\!3$) proposed by other nodes.   \\ 
$\TAWOneIDSetOtherProposedLastEpoch $    &  A randomized set of IDs of   accepted transactions ($1\!\leq \!\AW\!<\!3$) proposed by other nodes in   last epoch.   \\

\midrule

  $\txinitial$   & An  initial transaction.     \\
   $\siginitial$  & An  initial  signature for $\txinitial$.  	\\
  $\contentinitial$ &  $ \contentinitial = \ltuple  0, 0, 0,  0,  \txinitial, \defaultvalue, \defaultvalue\rtuple$ represents an  initial  content of  signature  $\siginitial$.  \\
 $\idtxinitial$   &  $\idtxinitial= \HashZ(\txinitial)$. \\ 
 $\contenthashinitia$&  $\contenthashinitia = \Hash(\contentinitial)$.	  	 \\	
\midrule

$\DoneDict$ & $\DoneDict[\ltuple \eon, \NumProposed\rtuple] = 1$ indicates   the $\NumProposed$-th proposal is complete; otherwise, it is incomplete,   $\NumProposed \! \in [0, \MaxNumTx]$. \\

\midrule		   	  

  	   $\ltslayerMax$  & The  number of layers for $\LTS$ scheme. \\        
    $n_{\ltslayer}$  & The size of each group at Layer~$\ltslayer$,  for $\ltslayer\in[\ltslayerMax]$.       \\   
      	  $\TSthreshold_{\ltslayer}$  &  A threshold on the number partial signatures    within a group at Layer~$\ltslayer$,  for $\ltslayer\in[\ltslayerMax]$. \\  
  	   $t_{\ltslayer}$  &  $ t_{\ltslayer} =  n_{\ltslayer} - \TSthreshold_{\ltslayer} $,  for $\ltslayer\in[\ltslayerMax]$ . \\

\midrule

   $\LTS$ Requirements &   $n=\prod_{\ltslayer=1}^{\ltslayerMax} n_{\ltslayer}$    and   $\prod_{\ltslayer=1}^{\ltslayerMax} \TSthreshold_{\ltslayer}   \geq \TSthreshold$  	  \\     	
 
   	\midrule

$\LTSIndexBook$  &   $\LTSIndexBook[\ltuple \eonstar,  j\rtuple]\to   \jshuffle$   maps Node~$j$  to a new index~$\jshuffle$ for  Epoch~$\eonstar$, \\  & based on the index shuffling of    Epoch~$\eonstar$. \\   	   	
 \midrule
   $\thisnodeindexshuffle$  & $\thisnodeindexshuffle$  is the new index of this node $\thisnodeindex$ for $\LTS$ scheme, \\      &   based on  index shuffling, which will be changed every epoch.    \\
   
   \midrule

    $D\pop(key)$  &    $D\pop(key)$ returns the value associated with the key $key$ from a dictionary $D$, \\ & and then removes this key-value pair from $D$.     \\ 
    \midrule
     $D\pop(key, \None)$  &   $D\pop(key, \None)$  returns the value  $\None$ if the key $key$ is not in a dictionary $D$; \\ & otherwise, it behaves the same as $D.\text{pop}(key)$.    \\ 
\midrule
    Time        &       $D\pop(key)$:  $O(1)$ on average.  \\ 
    Complexities      &      $D\pop(key, \None)$:   $O(1)$ on average.  \\ 
\midrule 

Randomized Set    &   A randomized set is implemented internally using a list and a dictionary.  \\      
  \midrule
$R\Getrandom()$  &    The operation $R\Getrandom()$  returns a randomly selected value from a randomized set $R$.   \\   
  \midrule
Time    &     $R\Add()$: $O(1)$ on average,  in a randomized set $R$.  \\   
Complexities &    $R\Remove()$:   $O(1)$ on average,   in a randomized set $R$.   \\ 
      &   $R\Getrandom()$:   $O(1)$, in a randomized set $R$.    \\
    
    \midrule 
     
$Q\append()$ &      $Q\append()$ adds an item to the  back of a FIFO queue $Q$.  \\ 

 $Q\Popleft()$  &   $Q\Popleft()$    returns and removes an item from the front of  a FIFO queue  $Q$.     \\ 
 \midrule
 
 Time    &      $Q\append()$:  $O(1)$ amortized,  in a FIFO queue.  \\   
Complexities &    $R\Popleft()$:   $O(1)$,   in a FIFO queue.   \\ 
  
  \midrule
  
 Parameter $\NumProposedSeed $    &  When $\NumProposed=\NumProposedSeed$, this node   activates $\SeedGeneration$ protocol with other nodes to generate a random seed.   \\ 
  & The random seed is used to   shuffle the node indices for the $\LTS$ scheme for the next epoch. \\ 
\midrule

  Parameter  $\NumProposedIntevalRandomTxSelf $    &  A   parameter (positive integer) that is set as, e.g.,  $\NumProposedIntevalRandomTxSelf=2$.  \\ 
  \midrule
   Parameter    $\NumProposedIntevalRandomTxFromOtherProp $    &  A   parameter (positive integer) that is set as, e.g.,  $\NumProposedIntevalRandomTxFromOtherProp=\lceil n/10 \rceil$.  \\ 
    \midrule 
    New & First, select $\idtx$ from  $\TnewIDSetOtherProposedLastEpoch$ or $\TAWOneIDSetOtherProposedLastEpoch$, if they are not empty. \\  
Transaction & If they are both empty, select $\idtx$ from $\TnewselfQueue$ with probability $1/\NumProposedIntevalRandomTxSelf$. \\
Selection & Otherwise, select $\idtx$ from $\TnewIDSetOtherProposed$ with probability $(1 - \frac{1}{\NumProposedIntevalRandomTxSelf}) \cdot \frac{1}{\NumProposedIntevalRandomTxFromOtherProp}$. \\
Rule & Finally, select $\idtx$ from $\TnewIDSet$ with the remaining probability. \\
\midrule
        $\Hash $    &   A  hash function $\Hash: \{0,1\}^* \!\to\! \Group$  maps   messages to elements in $\Group$ and is modeled as a random oracle.  \\ 
\midrule
        $\HashZ $    &   A  hash function $\HashZ: \{0,1\}^*\! \to \!\FieldZ_\PrimeOrder$   maps   messages to elements in $\FieldZ_\PrimeOrder$ and is modeled as a random oracle.    \\

\midrule

         $\DelayPara $    &  A preset delay parameter.     \\ 
      
      \midrule

\end{tabu}
\end{center}
\end{table}
}

\begin{algorithm}
\caption{$\Ocior$ protocol. Code is shown for Node~$\thisnodeindex$. }    \label{algm:Ocior} 
\begin{algorithmic}[1]
\vspace{5pt}    
\footnotesize

 \Statex {\bf Initialization}:
\Indent
    	   	 	 	 
  			\State $\Global~$ Sets or Dictionaries $\SigAccept \gets \{\}; \TxAccept\gets \{\};   \WeightTxTwo\gets \{\}; \WeightTxThree\gets \{\};$  
  	$\heightdict\gets \{\};   \heightdictLock\gets \{\}; \TxHeightAcceptLock \gets \{\}$; $\TxHeightAcceptTemp\gets \{\}; \Txochildren\gets \{\};   \Tproposal\gets \{\}; \TMyproposal\gets\{\};  \TproposalPending \gets \{\}; 
  	\TSGroupAgg\gets \{\};  \LTSGroupAgg\gets \{\}; \LTSIndexBook\gets \{\};   \ADKGCertificate\gets \{\};   \SeedRecord\gets\{\}; \SeedVoteRecord\gets \{\};   \TnewTxDictionary\gets \{\};    \TproposedIDSet\gets \{\};   \TAWOnePOProposedIDSet\gets \{\}; \TxConflictingDic \gets \{\};\DoneDict\gets \{\};  \NumTxProposedLock\gets \{\}; \NumTxProposedVoted \gets \{\}; \NNeedProposalVP\gets \{\}; \NNeedProposalHeightLock\gets \{\}; \NNeedProposalHeightLockMLock\gets \{\}; \PProof\gets \{\}; \APSISent\gets \{\}$

  \State $\Global~$  Randomized Sets $\TnewIDSet\gets \{\}; \TnewIDSetOtherProposed\gets \{\}; \TnewIDSetOtherProposedLastEpoch\gets \{\};  \TAWOneIDSetOtherProposed\gets \{\};  \TAWOneIDSetOtherProposedLastEpoch\gets \{\}$   
 
  \State $\Global~$  FIFO Queue  $\TnewselfQueue\gets []$

 	 \State  $\Global~\eon \gets 1$;   $\height\gets 0; \NumProposed\gets 0$;  \  $\DelayPara \gets$ a preset delay parameter   
 	 \State $\Global~n \gets$ the total number of consensus nodes;  $t \gets $ the maximum number of corrupted nodes tolerated, for $\tResilianceCondition$; $\TSthreshold=\lceil \frac{n+t+1}{2} \rceil$	
 	 \State $\Global~ \MaxNumTx \gets  $  the  maximum number of transactions that can be proposed by a node  in an epoch, typically set as $\MaxNumTx > n^2$ 
 	 \State $\Global~  \NumProposedSeed \gets  \lceil \MaxNumTx/2 \rceil; \   \NumProposedIntevalRandomTxSelf\gets 2; \ \NumProposedIntevalRandomTxFromOtherProp\gets  \lceil n/10 \rceil$;  $\NumHeightMulticast \gets    \lceil n \log n \rceil$;   $\EpochOutageThreshold \gets  n$   
  	\State $\Global~$   $\ltslayerMax \gets$ number of layers for $\LTS$ scheme        \quad \emph{// $\LTS$ parameter requirements: $n=\prod_{\ltslayer=1}^{\ltslayerMax} n_{\ltslayer}$    and   $\prod_{\ltslayer=1}^{\ltslayerMax} \TSthreshold_{\ltslayer}   \geq \TSthreshold$} 
  	\State $\Global~$   $n_{\ltslayer} \gets$ the size of each group at Layer~$\ltslayer$,  for $\ltslayer\in[\ltslayerMax]$      
  	\State $\Global~$    $\TSthreshold_{\ltslayer} \gets $ the  threshold on the number partial signatures    within a group at Layer~$\ltslayer$,  for $\ltslayer\in[\ltslayerMax]$ 
  	\State $\Global~$    $t_{\ltslayer} \gets  n_{\ltslayer} - \TSthreshold_{\ltslayer}$,  for $\ltslayer\in[\ltslayerMax]$;    $\thisnodeindexshuffle \gets \thisnodeindex $;   $\myproofofPreviousProposedTx \gets  \defaultvalue $; $\myproofProposedTx \gets  \defaultvalue$;     $\myLastProposal \gets  \defaultvalue$; 
   	\State  $\Global~$  $\txinitial \gets$  an  initial transaction; $\siginitial \gets $   an  initial  signature;  $\contentinitial \gets \ltuple  0, 0, 0,  0,  \txinitial, \defaultvalue, \defaultvalue\rtuple$  
   	\State  $\Global~$  $\idtxinitial \gets \HashZ(\txinitial)$;  $\contenthashinitia \gets \Hash(\contentinitial)$; $\NumTxProposedLock[\eon]\gets \{\}$

  	\State    $\LTSIndexBook[\ltuple \eon, j \rtuple]\gets j; \  \heightdict[j] \gets 0;  \ \heightdictLock[j] \gets 0;   \NumTxProposedLock[\eon][j] \gets 0,  \  \forall j \in [n]$ 

	\State    $\SigAccept[\siginitial] \gets  \contentinitial;  \TxAccept[\idtxinitial]\gets   [\txinitial, \siginitial]$     
	\State $\TxHeightAcceptLock[\ltuple j, 0  \rtuple] \gets  [\idtxinitial, \siginitial,  \defaultvalue];  \  \TxHeightAcceptTemp[\ltuple j, 0 \rtuple]\gets  [\idtxinitial, \siginitial,  \defaultvalue],      \  \forall j \in [n]$   						
	\State $\Tproposal[\ltuple j, 0, 0  \rtuple] \gets  [\contenthashinitia, \contentinitial],   \  \forall j \in [n]$

	\State run  $\OciorADKG[\eon]$   protocol with other  nodes to generate threshold signature keys for   Epoch~$\eon$
 
	\State wait until   $\OciorADKG[\eon]$   outputs   $\sk_{\eon, \thisnodeindex}$, $\skl_{\eon, \thisnodeindexshuffle}$, $[\pk_{\eon}, \pk_{\eon, 1}, \dotsc, \pk_{\eon, n}]$, and $[\pkl_{\eon}, \pkl_{\eon, 1}, \dotsc, \pkl_{\eon, n}]$   as global parameters

  	\State  $\ADKGCertificate\gets \ADKGCertificate\cup\{1\}$;    $\DoneDict[\ltuple \eon, \numproposedstar\rtuple   ] \gets 0, \forall \numproposedstar \in [\MaxNumTx]$; $\DoneDict[\ltuple \eon, \NumProposed\rtuple ] \gets 1$

\EndIndent

 \Statex
 
	\Statex \  \emph{//  ************************************************  As a Proposer  ***********************************************}

	\State {\bf upon} $(\DoneDict[\ltuple \eon, \NumProposed\rtuple ] =1)\AND( \NumProposed < \MaxNumTx)$ {\bf do}:   \label{line:OciorFixstartpro}    
	\Indent
       
 		\State $\NumProposed \gets \NumProposed+1$  
		\If { $(\NumProposed=1)\AND(\eon>1)$}    \emph{// in this case, re-propose the last proposal (in the previous epoch)  but replace the indices with $\eon, \NumProposed$} 	\label{line:OciorProposeLastEA}    
   		\State $ \ltuple \PROPOSE, *, *, *,    \height,  \tx, \sigvp, \sigoptuple, \contentoptuple,  \eondiamond, \numproposeddiamond, \myproofofPreviousProposedTx, \myproofProposedTx   \rtuple \gets \myLastProposal$ 	\label{line:OciorProposeLastEB}	
		  
		\Else   \quad  \emph{// in this case,  propose a new   proposal} 
	 		 \State $[*, \sigvp, *]\gets \TxHeightAcceptTemp[\ltuple \thisnodeindex, \heightdict[\thisnodeindex]\rtuple]$           
	 		 \State $\ltuple   *, \eondiamond, \numproposeddiamond,  *,  *, *, *\rtuple     \gets \SigAccept[\sigvp] $             		 
			\State $\height \gets \heightdict[\thisnodeindex] +1$   	\       \emph{//  $\sigvp$ is  a recently accepted sig  of this Chain~$\thisnodeindex$, accepted at   $\heightdict[\thisnodeindex] $, at Epoch~$\eondiamond$ with index $\numproposeddiamond$} 	
			\State $[ \tx, \sigoptuple, \contentoptuple, \myproofProposedTx] \gets \GetNewTxNoKL()$    \quad     \emph{//   return a new legitimate $\tx$ with   $\AW(\tx)< \AWthrehold$}  
			
		\EndIf

		\State  $\content\gets \ltuple  \thisnodeindex, \eon, \NumProposed,  \height,  \tx, \sigvp, \sigoptuple \rtuple$;  $ \contenthash\gets \Hash(\content)$     	  
		\State  $\TMyproposal[\ltuple \thisnodeindex,\eon, \NumProposed\rtuple]\gets  [\contenthash, \content] $

   		\State $\myLastProposal\gets \ltuple \PROPOSE, \thisnodeindex, \eon, \NumProposed,    \height,  \tx, \sigvp, \sigoptuple, \contentoptuple,  \eondiamond, \numproposeddiamond, \myproofofPreviousProposedTx, \myproofProposedTx   \rtuple$ 
 		
 		\State $\send$ $\myLastProposal$ to  Node~$j$,  $\forall j \in [n]$      	 	   \label{line:OciorProposalEnd}    

	\EndIndent

\Statex

	\Statex \  \emph{//  ************************************************  As a Voter  ************************************************ }

	\State {\bf upon} receiving  $\ltuple \PROPOSE, j, \eonstar, \numproposedstar,    \heightstar,  \tx, \sigvp, \sigoptuple, \contentoptuple, \eondiamond, \numproposeddiamond, \proofofPreviousProposedTx, \proofProposedTx \rtuple$ from Node~$j \in [n]$   {\bf do}:      \label{line:OciorVoteBegin}       		
	\Indent

	\If { $\eonstar> \eon$}   
		\State  wait until       $\eon \geq \eonstar$	\quad 	  \emph{// record this pending proposal in  $\TproposalPending[\eonstar]$, and then open it for processing when $\eon \geq \eonstar$}  	  
	\EndIf

	\If {$\ltuple j,\eonstar, \numproposedstar\rtuple \notin \Tproposal$   and $\numproposedstar\in[\MaxNumTx]$ and $\heightstar\geq 1$}

		\State  $\content\gets \ltuple   j, \eonstar, \numproposedstar,  \heightstar,  \tx, \sigvp, \sigoptuple \rtuple$     	
		\State  $ \contenthash\gets \Hash(\content)$     	  
		\State  $\Tproposal[\ltuple j,\eonstar, \numproposedstar\rtuple]\gets  [\contenthash, \content] $  
		\State  $\PProof[\ltuple j,\eonstar, \numproposedstar\rtuple]\gets  \ltuple j, \eonstar, \numproposedstar, \eondiamond, \numproposeddiamond,\sigvp,  \proofofPreviousProposedTx  \rtuple $  
 	  	\State  $\ppinputtuple\gets \ltuple j, \eonstar, \numproposedstar,    \heightstar,  \tx, \sigvp, \sigoptuple, \contentoptuple, \eondiamond, \numproposeddiamond, \contenthash, \proofProposedTx \rtuple$			
		\If { $\ltuple j,\eondiamond, \numproposeddiamond\rtuple \in \Tproposal$}      \label{line:OciorVoteVPCondition}       	
			\State  $\ProposalProcess \ltuple  \ppinputtuple \rtuple$
 
		\ElsIf{$\ltuple j,\eondiamond, \numproposeddiamond\rtuple \in \NNeedProposalVP$}
			\State $\ltuple *, \eonprime, \numproposedprime,    *,  *, *, *, *, *, *, *, * \rtuple\gets   \NNeedProposalVP [\ltuple j,\eondiamond, \numproposeddiamond\rtuple] $	
			\If {$(\eonstar > \eonprime)\OR ((\eonstar =\eonprime)\AND(\numproposedstar>\numproposedprime))$}   
				\State $\NNeedProposalVP [\ltuple j,\eondiamond, \numproposeddiamond\rtuple] \gets \ppinputtuple$			
			\EndIf	
		\Else
			\State $\NNeedProposalVP [\ltuple j,\eondiamond, \numproposeddiamond\rtuple] \gets \ppinputtuple$								
		\EndIf		
		\If { $\ltuple j,\eonstar, \numproposedstar\rtuple \in \NNeedProposalVP$}      \label{line:OciorVoteVPConditionNeed} 
			\State $\inputtuple \gets \NNeedProposalVP\pop(\ltuple j,\eonstar, \numproposedstar\rtuple)$		       \emph{//return value associated with  the key and remove key-value   from dictionary}   
			\State  $\ProposalProcess (\inputtuple)$	
 		
		\EndIf

	\EndIf

	\EndIndent

\algstore{OciorTwoRoundA}

\end{algorithmic}
\end{algorithm}

\begin{algorithm}
\begin{algorithmic}[1]
\algrestore{OciorTwoRoundA}
\vspace{5pt}    
\footnotesize

\Procedure{$\ProposalProcess$}{$ j, \eonstar, \numproposedstar,    \heightstar,  \tx, \sigvp, \sigoptuple, \contentoptuple, \eondiamond, \numproposeddiamond, \contenthash, \proofProposedTx $}   
\Statex	    \emph{// ** Before running this, make sure $\ltuple j,\eondiamond, \numproposeddiamond\rtuple \in \Tproposal$. **}   

\Statex	    \emph{// ** Run this   to make sure $\sigvp$ is accepted and   the only one  $\VP$ accepted at   $\heightstar-1$ of Chain~$j$, before further executions. **}    
\Statex	    \emph{// ** Also make sure $\heightdictLock[j]=  \max\{\heightstar -2, 0\}$, before further executions. **}

	\If {$\ltuple j,\eondiamond, \numproposeddiamond\rtuple \in \Tproposal$}     
		\State $[*, \contentvp]\gets \Tproposal [\ltuple j,\eondiamond, \numproposeddiamond\rtuple]$ 	
	 	\State $\acceptcheckindicator \gets \Accept(\sigvp, \contentvp)$ 
 		\State  $ \ltuple   \jdiamond, *, *,  \heightdiamond,  *, *, * \rtuple       \gets \contentvp$ 	
			\If {$(\heightdiamond= \heightstar-1) \AND (\acceptcheckindicator=\true)\AND (\ltuple j, \heightdiamond\rtuple\in \TxHeightAcceptTemp)\AND( \jdiamond=j)$}     \label{line:OciorVoteVPCheckCond} 
  			\State $[*, \sigdiamond, *]\gets \TxHeightAcceptTemp[\ltuple j, \heightdiamond\rtuple]$
			\If {$\sigdiamond=\sigvp$}      \quad   \emph{// $VP$ check is good at this point}         \label{line:OciorVoteVPCheck}  			
				\State  $\HeightLockUpdate \ltuple j \rtuple$		   \       \emph{// interactively update   $\!\TxHeightAcceptLock$,	  $\!\heightdictLock[j]$,   $\!\WeightTxTwo$,  $\!\WeightTxThree$ for Chain~$j$}

  				\State $\heightprime  \gets \heightdictLock[j];$  $\pphlinputtuple\gets  \ltuple j, \eonstar, \numproposedstar,    \heightstar,  \tx,  \sigoptuple, \contentoptuple, \contenthash, \proofProposedTx \rtuple$ 				
	  			\If {$\heightprime=  \max\{\heightstar -2, 0\}$  }       
		  			\State  $\ProposalProcessHeightLock  \ltuple \pphlinputtuple \rtuple$ 
	  			\ElsIf{$\heightprime < \heightstar -2$}
	  			
					\If{$\ltuple j, \heightstar -2\rtuple \notin \NNeedProposalHeightLock$}  
						\State $\NNeedProposalHeightLock [\ltuple j, \heightstar -2\rtuple] \gets \pphlinputtuple$	
					\Else			
						\State $\ltuple *, \eonprime, \numproposedprime,    *,  *, *, *, *, * \rtuple\gets   \NNeedProposalHeightLock [\ltuple j,\heightstar -2\rtuple] $	
						\If {$(\eonstar > \eonprime)\OR ((\eonstar =\eonprime)\AND(\numproposedstar>\numproposedprime))$}   
							\State $\NNeedProposalHeightLock [\ltuple j, \heightstar -2\rtuple] \gets \pphlinputtuple$	
	  					\EndIf
	  				\EndIf

				\ElsIf{$(\heightprime> \heightstar -2)\AND (\ltuple j, \heightprime \rtuple \in \NNeedProposalHeightLock)$}
					\State $\inputtuple \gets \NNeedProposalHeightLock\pop(\ltuple j, \heightprime \rtuple) $		 	 	
				  	\State  $\ProposalProcessHeightLock (\inputtuple)$   
	  			\EndIf

	  		\EndIf     
 		\EndIf
 	\EndIf	
 		
	\State $\Return$ 
									
\EndProcedure

\Statex

\vspace{-9pt}

\Procedure{$\ProposalProcessHeightLock$}{$ j, \eonstar, \numproposedstar,    \heightstar,  \tx,  \sigoptuple, \contentoptuple, \contenthash, \proofProposedTx$}

\Statex	    \emph{// ** Before running this, make sure $\heightdictLock[j]=  \max\{\heightstar -2, 0\}$ and $\sigvp$ has been accepted at   $\heightstar-1$. **}   

\Statex	    \emph{// ** Run this   to make sure $\NumTxProposedLock[\eonstar][j] = \numproposedstar-1$ before further executions. **}

 \If {$\heightdictLock[j]=  \max\{\heightstar -2, 0\}$  }       		
 
  	  	\State   $\ProofCheckNumTxProposedLockUpdate (j, \eonstar)$ \   \emph{//  interactively  check  $\proofofPreviousProposedTx$ and update  $\NumTxProposedLock[\eonstar][j]$   }

		\State $ \numproposedprimeprime \gets \NumTxProposedLock[\eonstar][j]$; $\pphlinputtuple\gets  \ltuple j, \eonstar, \numproposedstar,    \heightstar,  \tx,  \sigoptuple, \contentoptuple, \contenthash, \proofProposedTx \rtuple$ 		  	  	   \label{line:OciorVotePPCheckCond} 
  		\If {$\numproposedprimeprime= \numproposedstar-1 $  }        \quad   \emph{// previous proposals proposed by Node~$j$ have been checked at this point}         \label{line:OciorVotePPCheck} 
 
	  		\State  $\ProposalProcessHeightLockMLock  \ltuple \pphlinputtuple \rtuple$

  		\Else
			\State $\NNeedProposalHeightLockMLock [\ltuple j, \eonstar, \numproposedstar-1\rtuple] \gets \pphlinputtuple$				  			
  		\EndIf  
  			
		\If { $(\ltuple j, \eonstar, \numproposedprimeprime \rtuple \in \NNeedProposalHeightLockMLock)\AND (\numproposedprimeprime \geq \numproposedstar)$}   \quad    \emph{//  just vote for the most recent proposal proposed from Node~$j$}  
			\State $\inputtuple \gets \NNeedProposalHeightLockMLock\pop(\ltuple j,\eonstar, \numproposedprimeprime \rtuple)$		 
			\State  $\ProposalProcessHeightLockMLock (\inputtuple)$	
 		
		\EndIf

\EndIf

\State $\Return$

\EndProcedure

\Statex

\vspace{-9pt}

\Procedure{$\ProposalProcessHeightLockMLock$}{$ j, \eonstar, \numproposedstar,    \heightstar,  \tx,  \sigoptuple, \contentoptuple, \contenthash, \proofProposedTx$}  
\Statex	    \emph{// ** To vote,  make sure $\sigvp$ is accepted and is the right $\VP$ (and the only one $\VP$) accepted at   $\heightstar-1$ of Chain~$j$. **}    
\Statex	    \emph{// **  To vote, also make sure   $\heightdictLock[j]=  \max\{\heightstar -2, 0\}$ and $\NumTxProposedLock[\eonstar][j] = \numproposedstar-1$. **}

 \If {$\heightdictLock[j]\!= \! \max\{\heightstar \!-\!2, 0\}$ and $\NumTxProposedLock[\eonstar][j] \!=\!\numproposedstar\!-\!1$ and $\size(\sigoptuple)\!=\!\size(\contentoptuple)$}         \label{line:OciorVoteOPTXCheckBegin}       
 
 			\State $\numsigop \gets \size(\sigoptuple)$;     $\allacceptcheckindicator\gets \true$;    $\idtx \gets \HashZ(\tx)$     
			\For {$\indexinset \inset \range(\numsigop)$}                                                       
				\State $\acceptcheckindicator\gets \Accept(\sigoptuple[\indexinset], \contentoptuple[\indexinset])$ 
 				\State $\allacceptcheckindicator\gets \allacceptcheckindicator \AND \acceptcheckindicator$   
			\EndFor

 			\If {$\allacceptcheckindicator $}    
 				\State $[\checkindicator, \txconflict ]\gets \CheckTx(\tx)$        \    \emph{//return true if accepted already, otherwise, make sure $\tx$ is legitimate }
 				\State $\checkprooftxindicator \!\gets\!\CheckProofTx(\checkindicator, \txconflict, \tx, \sigoptuple, \proofProposedTx)$            \emph{//vote if  $\tx$ has     $\APS$}    \label{line:OciorAPSproofTxCheck}       
 				
	 			\If {$(\checkindicator =\true)\OR(\checkprooftxindicator=\true)$}                     \label{line:OciorVoteOPTXCheckEnd}       
	 		 		\State $[\idtxprimeprime, *, *]\gets \TxHeightAcceptLock[\ltuple j, \max\{\heightstar -3, 0\} \rtuple]$;  	    $[*, \sigprime, *]\gets \TxHeightAcceptLock[\ltuple j, \max\{\heightstar -2, 0\} \rtuple]$ 	     
  			 		\State $[\idtxdiamond, \sigdiamond, \sigvpdiamond]\gets \TxHeightAcceptTemp[\ltuple j, \heightstar -1\rtuple]$ 						
			 		\If {$(\ltuple j, \heightstar -1 \rtuple \notin  \APSISent)\AND(\heightstar -2 \geq 0)\AND(\sigvpdiamond=\sigprime)$}     	
						\State   $\contentdiamond \gets \SigAccept[\sigdiamond]$;       $\contentprime \gets \SigAccept[\sigprime]$;   $\APSISent\gets \APSISent\cup \{\ltuple j,\heightstar -1 \rtuple \}$ 		 			
						\State $\send$ $\ltuple \TXDONE,   \sigprime, \contentprime, \sigdiamond, \contentdiamond  \rtuple$  to one  randomly selected  $\RPC$ node    \label{line:OciorAPSIsent}  
					\EndIf	 
					\State  $\TnewIDSetOtherProposed\Remove(\idtxprimeprime)$; $\TnewTxDictionary\pop(\idtxprimeprime, \None)$;  	$\TnewIDSetOtherProposed\Add(\idtx)$; $\TnewTxDictionary[\idtx] \gets \tx$     \label{line:OciorUpdateTnewIDSetPO}  
					\State  $\TAWOneIDSetOtherProposed\Remove(\idtxprimeprime)$;  	 $\TAWOneIDSetOtherProposedLastEpoch\Remove(\idtxprimeprime)$;		$\TAWOneIDSetOtherProposed\Add(\idtxdiamond)$;		    \label{line:OciorUpdateTAWOneIDSetPO}  
					\State   $\TnewIDSetOtherProposedLastEpoch\Remove(\idtxprimeprime)$  \label{line:OciorUpdateTnewIDSetPOLE}

			 		\If {$\eonstar=\eon$ and $\NumTxProposedLock[\eonstar][j] \!=\!\numproposedstar\!-\!1$}    \quad  \quad  \emph{// vote for the current epoch only}  
						\State  $\send$ $\ltuple \VOTE, j, \eonstar, \numproposedstar,  \Vote(\sk_{\eonstar, \thisnodeindex}, \contenthash), \VoteLTS(\skl_{\eonstar, \thisnodeindexshuffle}, \contenthash)   \rtuple$  to Node~$j$    \label{line:OciorVoteOK}       
						\State  $\NumTxProposedVoted[j] \gets [ \eonstar, \numproposedstar,    \tx]$         
					\EndIf
				\ElsIf{$(\txconflict \neq\defaultvalue)\AND(\eonstar=\eon)$}
					\State  $\send$ $\ltuple \CONFLICT, j, \eonstar, \numproposedstar,      \txconflict  \rtuple$    to Node~$j$			       \label{line:OciorVoteEnd}       			
				\EndIf
			\EndIf

\EndIf  
\State $\Return$ 	
									
\EndProcedure

\algstore{OciorNoFPNOKLAA}

\end{algorithmic}
\end{algorithm}

\begin{algorithm}
\begin{algorithmic}[1]
\algrestore{OciorNoFPNOKLAA}
\vspace{5pt}    
\footnotesize

	\Statex \  \emph{//  ************************************************  As a Proposer  *********************************************** }

	\State {\bf upon} receiving  $\ltuple \CONFLICT, \thisnodeindex, \eonstar, \numproposedstar,      \txconflict  \rtuple$       from a node,  and  $\ltuple \thisnodeindex,\eonstar, \numproposedstar\rtuple \!\in\! \TMyproposal$,    $\DoneDict[\ltuple \eonstar, \numproposedstar\rtuple ]= 0$,   $\eonstar\!=\!\eon$, $\numproposedstar\!=\!\NumProposed$     {\bf do}:     	 \label{line:OciorMyProofCBegin} 
\Indent
 	\If {$\myproofProposedTx=\defaultvalue$}       	 \emph{//   if  $\myproofProposedTx \neq \defaultvalue$, then $\myproofProposedTx$ should be a valid $\APS$ for $\tx$;   honest nodes should vote for it}	    \label{line:OciorConfMyproofProposedTXDefault} 
		\State  $[*, \content]\gets \TMyproposal[\ltuple \thisnodeindex,\eonstar, \numproposedstar\rtuple]$;   
		   $ \ltuple   *, *, *,  *,  \tx, *, * \rtuple \gets \content$     	
  		\State $\conflictcheckindicator \gets \ConflictTxCheck(\tx, \txconflict)$        \quad    \emph{// $\conflictcheckindicator=\true$  means:  $\tx$ conflicts with  $\txconflict$}		
 		\If {$\conflictcheckindicator = \true  $}       
 			\State $\myproofofPreviousProposedTx\gets \ltuple \ProofTypeConflict,    \eonstar, \numproposedstar,    \txconflict \rtuple $; 
			  $\DoneDict[\ltuple \eonstar, \numproposedstar\rtuple ] \gets 1$       \label{line:OciorMyProofC} 
		\EndIf
	\EndIf
\EndIndent

	\State {\bf upon} receiving    $\ltuple \VOTE, \thisnodeindex, \eonstar, \numproposedstar,  \vote, \votelts  \rtuple$ from Node~$j$,     and $\ltuple \thisnodeindex,\eonstar, \numproposedstar\rtuple \in \TMyproposal$, $\DoneDict[\ltuple \eonstar, \numproposedstar\rtuple ]  = 0$,   $\eonstar=\eon$, $\numproposedstar=\NumProposed\!$  {\bf do}:      \label{line:OciorAggRx}

	\Indent
		\State    $[\contenthash, \content] \gets \TMyproposal [\ltuple \thisnodeindex,\eonstar, \numproposedstar\rtuple]$; 
		  $\jshuffle\gets \LTSIndexBook[\ltuple \eonstar,  j\rtuple]$

		\If {$\TSVerify(\pk_{\eonstar, j}, \vote, \contenthash) = \true$ and $\LTSVerify(\pkl_{\eonstar,  \jshuffle}, \votelts, \contenthash) = \true$}    
			\State   $[\indicator,  \sig]  \gets  \SigAggregationLTS(\ltuple \eonstar, \numproposedstar\rtuple, \eonstar,  \jshuffle, \votelts, \contenthash)$  \quad     \emph{//  see Line~\ref{line:OciorBLStsSigAggregationLTS}  of Algorithm~\ref{algm:OciorBLS}}        
	 		\If {$\indicator=\true$ and  $\DoneDict[\ltuple \eonstar, \numproposedstar\rtuple ]  = 0$}    		
	 			\State $\Accept(\sig, \content)$;  	
	 			   $ \ltuple   *, *, *,  *,  *, \sigvp, * \rtuple      \gets \content$;  	
				   $\contentvp\gets \SigAccept[\sigvp]$       			
				\State $\send$ $\ltuple \TXDONE,   \sigvp, \contentvp, \sig, \content  \rtuple$  to $\RPC$ nodes   \label{line:OciorSendAPSLTS} 
 
 				\State $\myproofofPreviousProposedTx\gets  \defaultvalue $; 
				  $\DoneDict[\ltuple \eonstar, \numproposedstar\rtuple ] \gets 1$       \label{line:OciorMyProofNormalDoneLTS} 
			\EndIf

    		 	\IfThenElse {$\ltuple \eonstar, \numproposedstar\rtuple\notin \TSGroupAgg$}  {$\TSGroupAgg[\ltuple \eonstar, \numproposedstar\rtuple]\gets  \{j: \vote\}$} {$\TSGroupAgg[\ltuple \eonstar, \numproposedstar\rtuple][j]\gets   \vote$}

		\EndIf		
			 			
	\EndIndent

	\State {\bf upon}      $|\TSGroupAgg[\ltuple \eonstar, \numproposedstar\rtuple ]|  =  n-t $,  and $\ltuple \thisnodeindex,\eonstar, \numproposedstar\rtuple \in \TMyproposal$,  and   $\DoneDict[\ltuple \eonstar, \numproposedstar\rtuple ]  = 0$, and  $\eonstar=\eon$, $\numproposedstar=\NumProposed$     {\bf do}:  	 \label{line:OciorTSAgg}

	\Indent
		\State  $\wait$ for $\DelayPara$ time     \quad   \emph{//  to include  more partial signatures and complete $\LTS$ scheme, if possible, within the limited delay time}
	 	\If {$\DoneDict[\ltuple \eonstar, \numproposedstar\rtuple ]  = 0$}    	
		\State    $[\contenthash, \content] \gets \TMyproposal [\ltuple \thisnodeindex,\eonstar, \numproposedstar\rtuple]$	 
		\State  $\sig \gets \TSCombine(n, \TSthreshold, \TSGroupAgg[\IDProtocol], \contenthash)$    
	 		\If {$\DoneDict[\ltuple \eonstar, \numproposedstar\rtuple ]  = 0$}    		
	 			\State $\Accept(\sig, \content)$;  	
	 			   $ \ltuple   *, *, *,  *,  *, \sigvp, * \rtuple      \gets \content$;  	
				   $\contentvp\gets \SigAccept[\sigvp]$       			
				\State $\send$ $\ltuple \TXDONE,   \sigvp, \contentvp, \sig, \content  \rtuple$  to $\RPC$ nodes     \label{line:OciorSendAPSTS} 
 
 				\State $\myproofofPreviousProposedTx\gets  \defaultvalue $; 
				  $\DoneDict[\ltuple \eonstar, \numproposedstar\rtuple ] \gets 1$      \label{line:OciorMyProofNormalDoneTS} 
			\EndIf	

		\EndIf
	\EndIndent

 	\Statex \  \emph{//  ***********************************  As a Node ************* (Process New Transactions and $\HMDM$) *********** }

\State {\bf upon} receiving   $\ltuple \TX,  \tx,  \sigoptuple, \contentoptuple   \rtuple$    message  the first time      {\bf do}:            \quad   \emph{//  process new transactions}
	\Indent
		\State $\NewTXProcess(\tx,   \sigoptuple, \contentoptuple)$      
	\EndIndent

\State {\bf upon}     $\heightdictLock[j]=\heightdiamond$  such that $\heightdiamond \mod \NumHeightMulticast =0$ and $\heightdiamond > 0$    for $j\in[n]$    {\bf do}:               \emph{// $\HMDM$ sig $\&$ contents in  locked chains}    \label{line:OciorHMDMBegin} 						
	\Indent
		\State $\ProtocolID\gets \ltuple j , \heightdiamond - \NumHeightMulticast+1,  \heightdiamond \rtuple$; $\HMDMMsg \gets [   ]$   \quad  \emph{// the index of $\HMDMMsg$ begins with $0$}    
		
		\For {$\heightprime\in  [\heightdiamond - \NumHeightMulticast+1,  \heightdiamond]$}    
			\State $[*, \sig, *]\gets \TxHeightAcceptLock[\ltuple j, \heightprime  \rtuple]$; $\content \gets\SigAccept[\sig]$;  	$\HMDMMsg [\heightprime- (\heightdiamond - \NumHeightMulticast+1)] \gets \ltuple    \sig, \content\rtuple$			
		\EndFor   
	 	 \State  $\Pass$    $\HMDMMsg$  into $\OciorHMDMHash[\ProtocolID]$  as an input     \quad  \emph{// see Algorithm~\ref{algm:OciorHMDMHash}}    	     \label{line:OciorHMDMHashinput} 						
	\EndIndent

\State {\bf upon} $\OciorHMDMHash[ \ltuple j , \heightdiamond \!-\! \NumHeightMulticast\!+\!1,  \heightdiamond\! \rtuple]$ outputting    $\HMDMMsg\!:= \!\![ \ltuple    \sig_{j, \heightdiamond \!-\NumHeightMulticast\!+\!1}, \content_{j, \heightdiamond \!- \NumHeightMulticast\!+\!1} \rtuple , \dotsc, \! \ltuple    \sig_{j, \heightdiamond}, \content_{j, \heightdiamond} \rtuple ]$,   $\!j\!\in\! [n]$     {\bf do}:         
	\Indent
		\State accept the signatures and contents from $\HMDMMsg$ with height $> \heightdictLock[j]$ if any are missing in Chain~$j$ at this node     \label{line:OciorHMDMEndEnd}
	\EndIndent

\State {\bf upon} Chain~$j$ not growing for $\EpochOutageThreshold$ epochs, with $\heightdictLock[j]=\heightdiamond$ and $\heightdiamond >0$,  where $j \in [n]$, {\bf do}:     \label{line:OciorBroadcastBegin} 
	\Indent
		\For {$\heightprime\in  [\heightdiamond - \lfloor \heightdiamond/\NumHeightMulticast \rfloor \cdot \NumHeightMulticast ,  \heightdiamond+1]$}       \quad  \emph{// $\heightdictLock[j]=\heightdiamond$ implies $\ltuple j, \heightprime  \rtuple\in \TxHeightAcceptTemp$ for $1\leq \heightprime \leq \heightdiamond+1$}   
			\State $[*, \sig, *]\gets \TxHeightAcceptTemp[\ltuple j, \heightprime  \rtuple]$; $\content \gets\SigAccept[\sig]$    
			\State $\send$ $\ltuple \APSI,   \sig, \content  \rtuple$  to all consensus nodes and all $\RPC$ nodes  
		\EndFor
	\EndIndent

\State {\bf upon} receiving   $\ltuple \APSI,   \sig, \content  \rtuple$     message  the first time      {\bf do}:          
	\Indent
		 	\If {$\sig\notin \SigAccept$}   $\Accept(\sig, \content)$      \label{line:OciorHMDMEnd} 
		 	\EndIf	
 
	\EndIndent

 	\Statex \  \emph{//  ************************************ As a Node  ************  (Key Regeneration and Epoch Transition)   ************ }

    \State {\bf upon} $\NumProposed=\NumProposedSeed$, at the current epoch  $\eon$ {\bf do}:     \label{line:OciorStartSeedGen} 
	\Indent
	 	 \State  $\Pass$  a value $\numproposedstar=\NumProposedSeed$  into $\SeedGeneration[\eon]$  as an input        \label{line:OciorStartSeedInput} 
	\EndIndent

    \State {\bf upon} $\SeedGeneration[\eon]$ outputs  a  random value  $\seed$ at the current epoch  $\eon$ {\bf do}:         
	\Indent

  		\State  update $\LTSIndexBook$  for Epoch~$\eon+1$ based on  $\seed$; and    
	 	   $\Pass$ updated $\LTSIndexBook$ into $\OciorADKG[\eon+1]$    
	 	 
	\EndIndent

\State {\bf upon}   $\OciorADKG[\eon+1]$   outputs    $\sk_{\eon+1, \thisnodeindex}$, $\skl_{\eon+1, \LTSIndexBook[\eon+1, \thisnodeindex] }$, $[\pk_{\eon+1}, \pk_{\eon+1, 1}, \!\dotsc\!, \pk_{\eon+1, n}]$,   $[\pkl_{\eon+1}, \pkl_{\eon+1, 1},\! \dotsc\!, \pkl_{\eon+1, n}]$ {\bf do}:    
	\Indent
 
  		\State  $\ADKGCertificate\gets \ADKGCertificate\cup\{\eon+1\}$       \label{line:OciorADKGoutput} 
	\EndIndent

    \State {\bf upon} $\NumProposed=\MaxNumTx$   and $\eon+1 \in \ADKGCertificate$ and   $|\TAWOneIDSetOtherProposed| \geq 4$  and  $(\EONVOTE, \eon)$  not yet sent {\bf do}:       \label{line:OciorVoteNewEpochBegin} 
	\Indent
		\State $\send$ $(\EONVOTE, \eon)$ to  all nodes       	   \label{line:OciorSendEonVote} 
	\EndIndent

\State {\bf upon} receiving   $n-t$  $(\EONVOTE, \eon)$  messages from distinct nodes and $\eon+1 \in \ADKGCertificate$ and  $(\EONCONFIRM, \eon)$  not yet sent {\bf do}:     
\Indent  
	\State $\send$ $(\EONCONFIRM, \eon)$ to  all nodes       	 \label{line:OciorSendEONCONFIRM} 
 
\EndIndent

\State {\bf upon} receiving   $t+1$  $(\EONCONFIRM, \eonstar)$  messages from distinct nodes and $\eonstar+1 \in \ADKGCertificate$ and  $(\EONCONFIRM, \eonstar)$  not yet sent, for $\eonstar\geq 1$ {\bf do}:  
\Indent  
	\State $\send$ $(\EONCONFIRM, \eonstar)$ to  all nodes       	
 
\EndIndent

\State {\bf upon} receiving   $2t+1$  $(\EONCONFIRM, \eon)$  messages from distinct nodes  and $\eon+1 \in \ADKGCertificate$  {\bf do}:       \label{line:OciorReadyNewEpoch} 
\Indent  
	\If {$(\EONCONFIRM, \eon)$   not yet sent }     
		 $\send$ $(\EONCONFIRM, \eon)$ to  all nodes       	
	\EndIf
	\State   $\eon\gets \eon+1$; $\NumProposed\gets 0$;  update parameter $\DelayPara$   
	\State $\DoneDict[\ltuple \eon, \numproposedstar\rtuple ]  \gets 0, \forall \numproposedstar \in [\MaxNumTx]$;  $\NumTxProposedLock[\eon] \gets \{\}$; $\NumTxProposedLock[\eon][j] \gets -1, \forall j\in [n]$; 
  	     $\thisnodeindexstar \gets \LTSIndexBook[\eon, \thisnodeindex] $
  	\State $\TnewIDSetOtherProposedLastEpoch\gets \TnewIDSetOtherProposedLastEpoch\cup\TnewIDSetOtherProposed$;
  		   $\TAWOneIDSetOtherProposedLastEpoch\gets \TAWOneIDSetOtherProposedLastEpoch\cup\TAWOneIDSetOtherProposed$    \label{line:OciorUpdateEpochTnewIDSetOtherProposedLastEpoch} 
	\State $\erase$ all  old  private key shares  $\{\sk_{\eonprime, \thisnodeindex}, \skl_{\eonprime, *}\}_{\eonprime=1}^{\eon-1}$  and any temporary data related to those secrets of old epochs  
	\State $\DoneDict[\ltuple \eon, \NumProposed\rtuple ]   \gets 1$;  then go to a new epoch                \label{line:OciorVoteNewEpochEnd}

\EndIndent

\end{algorithmic}
\end{algorithm}

\begin{algorithm}
\caption{Algorithms for  $\Ocior$ protocol. Code is shown for Node~$\thisnodeindex$. }    \label{algm:OciorProcedures} 
\begin{algorithmic}[1]
\vspace{5pt}    
\footnotesize

 \Procedure{$\HeightLockUpdate$}{$ j$}  
\Statex	    \emph{// **  Interactively update   $\TxHeightAcceptLock$,	  $\heightdictLock[j]$,   $\WeightTxTwo$, and $\WeightTxThree$ for Chain~$j$. **}    
 
\State $ \heightdiamond  \gets  \heightdictLock[j]$

	\While {$(\ltuple j,  \heightdiamond +1\rtuple\in \TxHeightAcceptTemp)\AND(\ltuple j,  \heightdiamond +2\rtuple\in \TxHeightAcceptTemp)$} 		
	
		 \State $[*, \sig, \sigvp]\gets \TxHeightAcceptTemp[\ltuple j, \heightdiamond+2\rtuple]$   
 		 \State $[\idtxdiamond, \sigdiamond, \sigvpdiamond]\gets \TxHeightAcceptTemp[\ltuple j, \heightdiamond +1\rtuple]$ 	
		\State $[\idtxprime, \sigprime, *]\gets \TxHeightAcceptLock[\ltuple j, \heightdiamond  \rtuple]$   
  		\If {$(\sigvp=\sigdiamond)\AND(\sigvpdiamond=\sigprime)$}     
			\State $\TxHeightAcceptLock[\ltuple j, \heightdiamond+1\rtuple]\gets  [\idtxdiamond, \sigdiamond, \sigvpdiamond]$    
			\State $\heightdictLock[j] \gets \heightdiamond  +1$   
			\State $\WeightTxTwo[\idtxdiamond] \gets [j, \heightdiamond +1, \sigdiamond, \sig]$       \quad \emph{// $\sigdiamond$  (locked) $\to$ $\sig$ (not   locked)  accepted at heights $\heightdiamond+1$ and $\heightdiamond+2$ }
			\State $\WeightTxThree[\idtxprime] \gets [j, \heightdiamond , \sigprime, \sigdiamond, \sig]$ 		   \quad \emph{// $\sigprime \to \sigdiamond \to \sig$   are  accepted in Chain $j$; the first two  are  locked }		 
		\Else
			\State $\Return$    				
		\EndIf  	
		\State  $\heightdiamond \gets \heightdiamond+1  $  		 	 			
  	 	
	\EndWhile
	
\State $\Return$

\EndProcedure

 \Statex

\Procedure{$\ProofCheckNumTxProposedLockUpdate$}{$ j, \eonstar$}

\If {$(\eonstar = \eon)\AND (j\in \NumTxProposedVoted)$}   
 
 \State $ [ \eonvartriangle, \numproposedvartriangle,  \txvartriangle]\gets \NumTxProposedVoted[j]$

\If {$\NumTxProposedLock[\eonstar][j]  = -1$}  \quad      \emph{//  $\NumTxProposedLock[\eonstar][j]$ is initialized to $0$ when $\eonstar = 1$, and to $-1$ when $\eonstar > 1$}   

		\State $\ProofCheckNumTxProposedLockUpdateNormal( j, \eonvartriangle)$

		\If {$\NumTxProposedLock[\eonvartriangle][j]  \geq  \numproposedvartriangle$}   
			\State $\NumTxProposedLock[\eonstar][j] \gets   \max\{ \NumTxProposedLock[\eonstar][j], 0  \} $     
		\Else
 
			\State $\eonprime  \gets  \eonvartriangle+1 $     
			
			\While {$\eonprime \leq \eonstar$} 		
	
				\State $\ProofCheckNumTxProposedLockUpdateSpecicial( j, \eonprime)$
		
				\If {$\NumTxProposedLock[\eonprime][j]  \geq  1$}   
					\State $\NumTxProposedLock[\eonstar][j] \gets   \max\{ \NumTxProposedLock[\eonstar][j], 0  \} $    
					\State $\Break$      
				\EndIf
				\State $\eonprime \gets \eonprime+1$	 	 	
			\EndWhile		
		
		\EndIf

\EndIf

		\If {$\NumTxProposedLock[\eonstar][j]  \geq  0$}   
			\State $\ProofCheckNumTxProposedLockUpdateNormal( j, \eonstar)$
		\EndIf

\EndIf
	\State $\Return$    		 	      									
\EndProcedure

\Statex

\Procedure{$\ProofCheckNumTxProposedLockUpdateNormal$}{$ j, \eonstar$}

\State $ \numproposedstar  \gets \NumTxProposedLock[\eonstar][j]  +2$     

	\While {$(\ltuple j,\eonstar, \numproposedstar \rtuple \in \PProof)\AND(\ltuple j,\eonstar, \numproposedstar-1\rtuple \in \Tproposal)\AND(\numproposedstar\geq 2)\AND(\NumTxProposedLock[\eonstar][j] \geq 0)$} 		
		\State  $\inputtuple\gets \PProof[\ltuple j,\eonstar, \numproposedstar\rtuple] $  
				
	 	\State $\proofcheckindicator \gets \ProofCheck(\inputtuple)$         \emph{// check on  $\proofofPreviousProposedTx$ for  a previously proposed transaction, only for $\numproposedstar \geq 2$ }				
 		\If {$\proofcheckindicator=\true$}   
	 			\State $\NumTxProposedLock[\eonstar][j] \gets   \NumTxProposedLock[\eonstar][j]+1 $    
		\Else
			\State $\Return$    		 	      
		\EndIf 
		\State  $\numproposedstar \gets   \NumTxProposedLock[\eonstar][j]  +2$  		 	 	
	\EndWhile

	\State $\Return$    		 	      									
\EndProcedure

\Statex

\Procedure{$\ProofCheckNumTxProposedLockUpdateSpecicial$}{$ j, \eonstar$}

\If {$(\NumTxProposedLock[\eonstar][j]  \!<\!1)\AND   (\ltuple j,\eonstar, 2 \rtuple \!\in\! \PProof)\AND(\ltuple j,\eonstar, 1\rtuple\! \in \!\Tproposal)\AND(j\!\in\! \NumTxProposedVoted)$}   
 	\State $ [ *, *,  \txvartriangle]\gets \NumTxProposedVoted[j]$; $\inputtupleprimeprime\gets \PProof[\ltuple j,\eonstar, 2\rtuple] $ 	    		
	\State  $[*, \content]\gets \Tproposal[\ltuple j,\eonstar, 1\rtuple]$;  $ \ltuple   *, *, *,  *,  \tx, *, * \rtuple \gets \content$     	
 	\If {$\txvartriangle=\tx$}   		
		\State $\NumTxProposedLock[\eonstar][j] \gets   \max\{ \NumTxProposedLock[\eonstar][j], 0  \} $     	
	\EndIf 	

 	\If {$\NumTxProposedLock[\eonstar][j] = 0$}

	 	\State $\proofcheckindicator \gets \ProofCheck(\inputtupleprimeprime)$    
 		\If {$\proofcheckindicator=\true$}   
	 			\State $\NumTxProposedLock[\eonstar][j] \gets   \NumTxProposedLock[\eonstar][j] +1 $    
		\Else
			\State $\Return$    		 	      
		\EndIf 
	\EndIf
 
\EndIf

	\State $\Return$    		 	      									
\EndProcedure

\algstore{OciorTwoRoundAA}

\end{algorithmic}
\end{algorithm}

\begin{algorithm}
\begin{algorithmic}[1]
\algrestore{OciorTwoRoundAA}
\vspace{5pt}    
\footnotesize

 \Procedure{$\ProofCheck$}{$j, \eonstar, \numproposedstar, \eondiamond, \numproposeddiamond,\sigvp,  \proofofPreviousProposedTx$}       \quad \quad   \emph{// only for $\numproposedstar \geq 2$} 
 \Statex     \emph{//   ** Check on the $\proofofPreviousProposedTx$ for  a  transaction  $\tx$ previously proposed   by Node~$j$  **} 	
 \Statex     \emph{//   ** If the  proposed $\tx$  conflicts with   other transaction,  then  remove it from  $\TnewTxDictionary$,  $\TnewIDSetOtherProposed$ and $\TnewIDSetOtherProposedLastEpoch$. **}

\State  $\proofcheckindicator\gets  \false $

\If {$\numproposedstar \geq 2$}     
	
	\If {$(\eonstar=\eondiamond)\AND(\numproposedstar=\numproposeddiamond+1)\AND ( \ltuple j,\eondiamond, \numproposeddiamond\rtuple \in \Tproposal)$}     
		\State $[\contenthash, \content]\gets \Tproposal [\ltuple j,\eondiamond, \numproposeddiamond\rtuple]$ 
 		\If {$\TSVerify(\pk_{\eonstar}, \sigvp, \contenthash) = \true$}   	   $\proofcheckindicator\gets \true$         		 	    	 	    
		\EndIf	
		\State $\Return$ $\proofcheckindicator$          		 	    
	\EndIf

	\If {$\proofofPreviousProposedTx$ takes the form   $\proofofPreviousProposedTx\!:=\!\ltuple \ProofTypeConflict,   \eonccirc, \numproposedcirc,   \txconflict \rtuple$ and $(\ltuple j,\eonstar, \numproposedstar \!-\!1\rtuple \in \Tproposal)\AND(\eonccirc\!=\!\eonstar)\AND(\numproposedcirc\!=\!\numproposedstar\! -\!1)$}   
	
		\State  $[*, \content]\gets \Tproposal[\ltuple j,\eonstar, \numproposedstar -1\rtuple]$;  $ \ltuple   *, *, *,  *,  \tx, *, * \rtuple \gets \content$     	
		
  		\State $\conflictcheckindicator \gets \ConflictTxCheck(\tx, \txconflict)$         
 	
 		\If {$\conflictcheckindicator = \true$}         \quad  \emph{// $\conflictcheckindicator=\true$  means:  $\tx$ conflicts with  $\txconflict$}	
			\State $\proofcheckindicator\gets \true$;  $\idtx \gets \HashZ(\tx)$  
			\State  $\TnewIDSetOtherProposed\Remove(\idtx)$;  $\TnewTxDictionary\pop(\idtx, \None)$;   $\TnewIDSetOtherProposedLastEpoch\Remove(\idtx)$    \label{line:RemoveUpdateTnewIDSetOtherProposed}
		\EndIf
		\State $\Return$ $\proofcheckindicator$          			 	    	 	    		
	\EndIf

\EndIf
 
\State $\Return$ $\proofcheckindicator$

\EndProcedure

\Statex

\Procedure{$\CheckTx$}{$\tx$}   
	\Statex     \emph{//   ** Return true if $\tx$ has been already accepted.  **} 	
	\Statex     \emph{//   ** Return true if 1) the signature is valid; 2) the address matches;  3)   the amount matches; and 4) no  double spending.  **} 	
  	\State $\idtx\gets \HashZ(\tx)$

	\If {$\idtx \in  \TxAccept$} 
		 $\Return$ $[\true, \defaultvalue]$ 	
	\EndIf

	\State $\idoptuple  \gets $	 the tuple of 	IDs  of official parents  of  $\tx$, obtained from $\tx$ 		     	     
	\State $\opindextuple  \gets $	 the tuple of indices  referring to   the sender address of $\tx$ in  the corresponding $\op$s, obtained from $\tx$ 
	\State $\sender  \gets $	  the sender address of $\tx$, obtained from $\tx$; 
	  $\outamount  \gets $	  the total output amount  of $\tx$, obtained from $\tx$ 
	\State $\feeamount  \gets $	  the   amount of fee  of $\tx$,   obtained from $\tx$; 
 	 $\inamount \gets 0$     			
	\If {$($the signature in $\tx$ is  NOT valid$)\OR(\size(\idoptuple) \neq \size(\opindextuple))$} 
		  $\Return$ $[\false, \defaultvalue]$  			
	\EndIf	  
	\If {$\idtx\in \TxConflictingDic$} 
		\State $[*,  \txconflict] \gets \TxConflictingDic[\idtx]$   
		\State $\Return$ $[\false, \txconflict]$  	
	\EndIf		
	\For {$\indexinset \inset \range(\size(\idoptuple))$} 
		\State $\idop\gets \idoptuple[\indexinset];   \opindex\gets \opindextuple[\indexinset]$   
		\If {$\idop \notin \TxAccept $} 
			 $\Return$ $[\false, \defaultvalue]$  					
		\EndIf
		\State $\txop \gets \TxAccept[\idop][0]$; 		
		  $\opreceiver \gets $ the address of the $\opindex$ th recipient  in   $\txop$  	   
		\State $\opamount \gets $  the output amount of the $\opindex$ th recipient   in   $\txop$;   	 
		 $\inamount \gets \inamount + \opamount$  	 
		\If {$\opreceiver \neq  \sender$} 
			 $\Return$ $[\false, \defaultvalue]$  					
		\EndIf

		\If {$\ltuple \idop, \opindex\rtuple \notin \Txochildren$}   
			\State $\Txochildren[\ltuple \idop, \opindex\rtuple]  \gets  [\idtx, \tx]$        \emph{//record to avoid voting for double spending from $\txop$; record only  one time}
		\Else
			\State $[\idtxconflict, \txconflict] \gets \Txochildren[\ltuple \idop, \opindex\rtuple]$      
			\If {$\idtxconflict \neq \idtx$}    \quad\quad   \emph{//   this means that $\tx$ conflicts with another transaction} 
				\State $\TxConflictingDic[\idtx] \gets [\tx,  \txconflict]$  \quad  \emph{//   record  $\tx$  as a transaction that conflicts with   another transaction } 
				\State $\Return$ $[\false, \txconflict]$  
			\EndIf 
		\EndIf 		
	\EndFor
    	
    	\IfThenElse {$\inamount= \outamount +\feeamount$}  {$\Return$ $[\true, \defaultvalue]$} {$\Return$ $[\false, \defaultvalue]$}

\EndProcedure
 
\Statex
 
\Procedure{$\ConflictTxCheck$}{$\tx, \txconflict$}  \quad   \emph{//     Return $\true$ if $\tx$ conflicts with  $\txconflict$.   } 	
 	\Statex  \emph{//   ** If $\tx$ conflicts with  $\txconflict$, record  $\TxConflictingDic[\idtx] \gets [\tx,  \txconflict]$.	 **}

  	\State $\OPSetTemporary \gets \{\}$	

	\If {$($the signature in $\tx$ is  valid$)\AND($the signature in $\txconflict$ is  valid$)\AND (\tx \neq  \txconflict)$} 
		\State $\sender  \gets $	  the sender address of $\tx$;
		  $\senderprime  \gets $	  the sender address of $\txconflict$ 
		\If {$\sender = \senderprime$} 	
			
		  	\State $\idtx\gets \HashZ(\tx)$; 
			  $\idoptuple  \gets $	 the tuple of 	IDs  of official parents  of  $\tx$, obtained from $\tx$ 		     	     
			\State $\opindextuple  \gets $	 the tuple of indices  referring to   the sender address of $\tx$ in  the corresponding $\op$s, obtained from $\tx$

		 	\State $\idoptupleprime  \gets $	 the tuple of 	IDs  of official parents  of  $\txconflict$, obtained from $\txconflict$ 		     	     
			\State $\opindextupleprime  \gets $	 the tuple of indices  referring to   the sender address of $\txconflict$ in  the corresponding $\op$s 
 
			\If {$(\size(\idoptuple) = \size(\opindextuple)) \AND (\size(\idoptupleprime) = \size(\opindextupleprime))$} 		 
				\For {$\indexinset \inset \range(\size(\idoptuple))$} 
					\State $\idop\gets \idoptuple[\indexinset]$; $\opindex\gets \opindextuple[\indexinset]$;   				
	  				  $\OPSetTemporary \Add (\ltuple  \idop,  \opindex \rtuple)$					
				\EndFor

	 			\For {$\indexinset \inset \range(\size(\idoptupleprime))$} 
					\State $\idop\gets \idoptupleprime[\indexinset]$;   
					 $\opindex\gets \opindextupleprime[\indexinset]$   
					\If {$ \ltuple  \idop,  \opindex \rtuple \in    \OPSetTemporary$} 					
						\State $\TxConflictingDic[\idtx] \gets [\tx,  \txconflict]$
						\State $\Return$ $\true$  		
					\EndIf
				\EndFor 
			\EndIf
		\EndIf		
	\EndIf
	
	\State $\Return$ $\false$

\EndProcedure

\algstore{OciorNoFPNOKLCCCC}

\end{algorithmic}
\end{algorithm}

\begin{algorithm}
\begin{algorithmic}[1]
\algrestore{OciorNoFPNOKLCCCC}
\vspace{5pt}    
\footnotesize

\Procedure{$\Accept$}{$\sig, \content$}        \label{line:AcceptBegin}
\Statex     \emph{//   ** Return true if $\sig$ is already accepted or is successfully accepted here at the end.  **} 	
\Statex	    \emph{// ** $\TxHeightAcceptTemp[\ltuple j, \heightdiamond\rtuple]$   accepts only one $\sig$ at a given height $\heightdiamond$ of  Chain~$j$ **}

 	\State  $\acceptcheckindicator\gets  \false $

	\If {$\content$ takes the form   $\content:=\ltuple j, \eondiamond, \numproposeddiamond,  \heightdiamond,  \tx, \sigvp, \sigoptuple\rtuple$,   and $\eondiamond\leq \eon$}   
 		 
 		\State  $ \contenthash\gets \Hash(\content)$   
 		\If {$\TSVerify(\pk_{\eondiamond}, \sig, \contenthash) = \true$}     
 
	 		\If {$\sig \notin  \SigAccept$}      
	  			\State $\SigAccept[\sig] \gets  \content  $;   \  $\idtx\gets \HashZ(\tx)$             \label{line:AcceptIntoSet}	
	    			\IfThenElse {$\idtx \notin \TxAccept$}  {$\TxAccept[\idtx]\gets  [\tx, \sig]$} {$\TxAccept[\idtx]\append(\sig)$ }     
	 			\If {$\ltuple j, \heightdiamond\rtuple \notin  \TxHeightAcceptTemp$}      
					\State $\TxHeightAcceptTemp[\ltuple j, \heightdiamond\rtuple]\gets   [\idtx, \sig, \sigvp]$;     $\heightdict[\nodeindexj] \gets \max\{\heightdict[\nodeindexj], \heightdiamond\}$   	
				\EndIf
 
				\State  $\Tproposal[\ltuple j,\eondiamond, \numproposeddiamond\rtuple]\gets  [\contenthash, \content] $  
 
				\State $\idoptuple  \gets $	 the tuple of 	IDs  of official parents  of  $\tx$, obtained from $\tx$ 		     	     
				\State $\opindextuple  \gets $	 the tuple of indices  referring to     sender address of $\tx$ in  the corresponding $\op$s, obtained from $\tx$ 	  
				\For {$\indexinset \in  \range(\size(\idoptuple))$}   
						\State $\Txochildren[\ltuple \idoptuple[\indexinset], \opindextuple[\indexinset] \rtuple ] \gets  [\idtx, \tx]$             \emph{// record   to avoid voting for double spending } 
				\EndFor		 					
			\EndIf  
		 	\State  $\acceptcheckindicator\gets  \true $ 		
		\EndIf
	\EndIf 
	\State $\Return$ $\acceptcheckindicator$          		    \label{line:AcceptEnd}	
\EndProcedure

 \vspace{4pt}

 \Procedure{$\NewTXProcess$}{$\tx,   \sigoptuple, \contentoptuple$}

	\State $\numsigop \gets \size(\sigoptuple)$; 
	  $\idtx \gets   \HashZ(\tx)$ 	
	\If {$(\numsigop = \size(\contentoptuple))\AND (\numsigop \geq 1)\AND(\idtx \notin \WeightTxThree) \AND(\idtx\notin\TproposedIDSet)\AND(\idtx\notin\TnewIDSet)\AND(\idtx\notin \TxConflictingDic)$  }

 			\State $\allacceptcheckindicator\gets \true$;  $\checkindicator \gets \false$   
			\For {$\indexinset \inset \range(\numsigop)$}           
				\State $\acceptcheckindicator\gets \Accept(\sigoptuple[\indexinset], \contentoptuple[\indexinset])$ 
 				\State $\allacceptcheckindicator\gets \allacceptcheckindicator \AND \acceptcheckindicator$   
			\EndFor
 			\If {$\allacceptcheckindicator$}       
 				\State $[\checkindicator, \txconflict ]\gets \CheckTx(\tx)$        \quad    \emph{// make sure $\tx$ is legitimate and no double spending from $\txop$}
 			\EndIf			
		
	 		\If {$\checkindicator =\true$}     
				\State  $\TnewTxDictionary[\idtx] \gets  \tx$;      $\TnewIDSet\Add(\idtx)$ 				
	 			\If {$\idtx \mod n = \thisnodeindex -1$ } 
					\State   $\TnewselfQueue\append(\idtx)$
				\EndIf	
			\EndIf			
	\EndIf		
	
\State $\Return$
							
\EndProcedure

 \vspace{4pt}

\Procedure{$\GetNewTxNoKL$}{$\empty$}             \label{line:GetNewTxBegin}

\While {$|\TnewIDSetOtherProposedLastEpoch| >0$} 
	\State $\idtx \gets \TnewIDSetOtherProposedLastEpoch\Getrandom()$; \  $\TnewIDSetOtherProposedLastEpoch\Remove(\idtx)$;\  $\TnewIDSet\Remove(\idtx)$ 
 
	\State $[\BinaryIndicator, \tx, \sigoptuple, \contentoptuple, \proofProposedTx] \gets \GetNewTxNoKLCheck(\idtx)$  
	\If {$\BinaryIndicator$}  
		\State $\TproposedIDSet\Add(\idtx)$         	
		\State $\Return$ $[\tx, \sigoptuple, \contentoptuple,\proofProposedTx]$         
	\EndIf	
\EndWhile

\While {$|\TAWOneIDSetOtherProposedLastEpoch| >0$} 
	\State $\idtx \gets \TAWOneIDSetOtherProposedLastEpoch\Getrandom()$; \  $\TAWOneIDSetOtherProposedLastEpoch\Remove(\idtx)$ \  $\TnewIDSet\Remove(\idtx)$ 
 
	\State $[\BinaryIndicator, \tx, \sigoptuple, \contentoptuple, \proofProposedTx] \gets \GetNewTxNoKLCheck(\idtx)$  
	\If {$\BinaryIndicator$}  
		\State $\TproposedIDSet\Add(\idtx)$;  $\TAWOnePOProposedIDSet\Add(\idtx)$         	      	
		\State $\Return$ $[\tx, \sigoptuple, \contentoptuple,\proofProposedTx]$         
	\EndIf	
\EndWhile

\While {$(|\TnewselfQueue| >0)\AND (\NumProposed \mod  \NumProposedIntevalRandomTxSelf  =  0)$} 
	\State $\idtx \gets \TnewselfQueue\Popleft()$;  \      $\TnewIDSet\Remove(\idtx)$
 
	\State $[\BinaryIndicator, \tx, \sigoptuple, \contentoptuple, \proofProposedTx] \gets \GetNewTxNoKLCheck(\idtx)$
	\If {$\BinaryIndicator$}  
		\State $\TproposedIDSet\Add(\idtx)$         	
		\State $\Return$ $[\tx, \sigoptuple, \contentoptuple, \proofProposedTx]$        
	\EndIf	
\EndWhile

\While {$(|\TnewIDSetOtherProposed| >0)\AND (\NumProposed \mod  \NumProposedIntevalRandomTxFromOtherProp  =  0)$} 
	\State $\idtx \gets \TnewIDSetOtherProposed\Getrandom()$;  \  $\TnewIDSetOtherProposed\Remove(\idtx)$; \  $\TnewIDSet\Remove(\idtx)$     \label{line:GetUpdateTnewIDSetOtherProposed} 
 
	\State $[\BinaryIndicator, \tx, \sigoptuple, \contentoptuple,\proofProposedTx] \gets \GetNewTxNoKLCheck(\idtx)$
	\If {$\BinaryIndicator$}  
		\State $\TproposedIDSet\Add(\idtx)$         	
		\State $\Return$ $[\tx, \sigoptuple, \contentoptuple,\proofProposedTx]$        
	\EndIf	
\EndWhile

\While {$\true$} 
	\State $\idtx \gets \TnewIDSet\Getrandom()$;  \  $\TnewIDSet\Remove(\idtx)$ 
 
	\State $[\BinaryIndicator, \tx, \sigoptuple, \contentoptuple,\proofProposedTx] \gets \GetNewTxNoKLCheck(\idtx)$
	\If {$\BinaryIndicator$}  
		\State $\TproposedIDSet\Add(\idtx)$         	
		\State $\Return$ $[\tx, \sigoptuple, \contentoptuple,\proofProposedTx]$             \label{line:GetNewTxEnd} 
	\EndIf	
\EndWhile

\EndProcedure

\algstore{OciorNoFPNOKLCCDD}

\end{algorithmic}
\end{algorithm}

\begin{algorithm}
\begin{algorithmic}[1]
\algrestore{OciorNoFPNOKLCCDD}
\vspace{5pt}    
\footnotesize   

\Procedure{$\GetNewTxNoKLCheck$}{$\idtx$}          
\Statex  \emph{// ** The acceptance weight of selected  $\idtx$ needs to be less than  $\AWthrehold$, i.e., $\idtx \notin \WeightTxThree$.   **}   
\Statex  \emph{// ** Make  sure   $\idop\in \TxAccept$, where $\idop$ is the ID of official parent of  selected $\tx$. **}  
 
 	\State  $\sigoplist\gets  [ \ ] ;   \contentoplist\gets  [ \ ]; \sigoptuple\gets \ltuple \rtuple;  \contentoptuple\gets \ltuple \rtuple$;   $\BinaryIndicator\gets  \false ; \proofProposedTx\gets  \defaultvalue $      \label{line:GetNewTxNoKLStepBegin} 	
		\If {$((\idtx \notin \WeightTxThree)\AND(\idtx\notin \TAWOnePOProposedIDSet))\OR((|\TnewIDSetOtherProposedLastEpoch| =0)\AND(1\leq |\TAWOneIDSetOtherProposedLastEpoch| \leq 3))$}         \label{line:GetNewTxNoKLCondition} 		 	  
		 	\If {$\idtx \in \TxAccept$}         
				\State  $\BinaryIndicator\gets  \true;  \sig\gets \TxAccept[\idtx][1]$; $\TnewTxDictionary\pop(\idtx, \None)$ 
				\State $\content \gets\SigAccept[\sig]$;  $\ltuple   j, \eonstar, \numproposedstar,  \heightstar,  \tx, \sigvp, \sigoptuple \rtuple \gets  \content$ 				
 				\State $\proofProposedTx \gets \ltuple\sig,   \ltuple   j, \eonstar, \numproposedstar,  \heightstar,  \defaultvalue, \sigvp, \defaultvalue \rtuple \rtuple$ 	     \label{line:proofProposedTxvalue} 				
				\For {$\sigop\in  \sigoptuple$}    	 \quad     \emph{//  interactively get  from the fist element to the last element in  $\sigoptuple$}     
					\State $\contentop \gets\SigAccept[\sigop]; \contentoplist\append(\contentop)$ 				
				\EndFor   			
				\State  $\contentoptuple\gets  \tuple(\contentoplist)$ 	
		 	\ElsIf{$(\idtx \in \TnewTxDictionary)\AND(\idtx\notin \TproposedIDSet)$}     	 	    
			 		\State $\tx\gets \TnewTxDictionary\pop(\idtx)$          
						\State $[\checkindicator, * ]\gets \CheckTx(\tx)$          \    \emph{// make sure $\tx$   legitimate and no double spending}
						 \If {$\checkindicator$}    
							\State $\idoptuple  \gets $	 the tuple of 	IDs  of official parents  of  $\tx$, obtained from $\tx$

							\For {$\idop\in  \idoptuple$} 
								 \If {$\idop\in \TxAccept$}   
							 		\State $\sigop \gets \TxAccept[\idop][1]$          \emph{//  if $\sigop$ is in $\TxAccept[\idop][1]$,  it  should be in $\SigAccept$ (Line~\ref{line:AcceptIntoSet})}     
								 		 \State $\contentop \!\gets\!\!\SigAccept\![\sigop]; \sigoplist\append(\sigop); \contentoplist\append(\contentop)$ 
								\Else
									\State $\Break$
								\EndIf 					
						 	\EndFor
						 	\If {$\size(\sigoplist)= \size(\idoptuple)$}   
									\State  $\BinaryIndicator\gets  \true ; \sigoptuple\gets  \tuple(\sigoplist);    \contentoptuple\gets  \tuple(\contentoplist)$ 						 
			 
							\EndIf
						\EndIf 
									
			\EndIf	
		\EndIf
	\State $\Return$ $[\BinaryIndicator, \tx, \sigoptuple , \contentoptuple ,\proofProposedTx]$         \label{line:GetNewTxNoKLStepEnd} 
 
\EndProcedure

 \Statex

\Procedure{$\CheckProofTx$}{$\checkindicator, \txconflict, \tx, \sigoptuple, \proofProposedTx$}          
 
 \Statex  \emph{// ** Return $\true$ if    $\tx$ has   a valid   $\APS$, or if $\checkindicator=\true$, or if $\tx$ has been accepted already. **}  
 
	\State $\idtx \gets   \HashZ(\tx)$ 	
 	\If {$(\checkindicator =\true)\OR(\idtx\in \TxAccept)$}     	 	     		
		\State  $\Return$ $\true$ 		 
	\EndIf

 	\If {$\txconflict \neq   \defaultvalue$}     
 		
		\If {$\proofProposedTx$ takes the form $\proofProposedTx:= \ltuple \sig,   \content\rtuple$  and $\content:= \ltuple   j, \eonstar, \numproposedstar,  \heightstar,  *, \sigvp, * \rtuple$ }   
	 		\State $\contentprime \gets  \ltuple   j, \eonstar, \numproposedstar,  \heightstar,  \tx, \sigvp, \sigoptuple \rtuple$;

	 	 	\State $\acceptcheckindicator \gets \Accept(\sig, \contentprime)$

	 		\If {$\acceptcheckindicator= \true$}     
		 	
				\State  $\Return$ $\true$ 		 	 	
				
			\EndIf
		\EndIf	
	\EndIf
	\State $\Return$ $\false$          
 
\EndProcedure

\end{algorithmic}
\end{algorithm}

\begin{algorithm}  
\caption{$\SeedGeneration$  protocol, with an identifier $\eonstar$. Code is shown for Node~$\thisnodeindex$.}    \label{algm:OciorNoFPKLSeed} 
\begin{algorithmic}[1]
\vspace{5pt}    
\footnotesize
 
\Statex   \emph{//   **  Generate a random seed to  shuffle the indices of  nodes, recorded  at $\LTSIndexBook$, for the $\LTS$ scheme for the next epoch.  **}

		\State  initially set $\SeedVoteRecord\gets \{\}; \SeedVoteRecord[\eonstar] \gets \{\}; \SeedRecord\gets \{\}$  
\State {\bf upon} receiving an input  value  $\numproposedstar$, for $\numproposedstar=\NumProposedSeed$   {\bf do}:    
	\Indent
		\State  $\contenthash \gets \Hash \ltuple \SEED, \eonstar \rtuple$  
		\State  $\send$ $\ltuple \SEEDVOTE, \thisnodeindex, \eonstar,   \Vote(\sk_{\eonstar, \thisnodeindex}, \contenthash) \rtuple$  to all nodes
	\EndIndent

    \State {\bf upon} receiving $\ltuple \SEEDVOTE, j, \eonstar,   \vote \rtuple$  from Node~$j$, and  $\eonstar\notin \SeedRecord$ {\bf do}:  
	\Indent
		\State  $\contenthash \gets \Hash \ltuple \SEED, \eonstar \rtuple$  
		\If {$\TSVerify(\pk_{\eonstar, j}, \vote, \contenthash) = \true$}     
			\State $\SeedVoteRecord[\eonstar][j]\gets \vote$
			\If {$|\SeedVoteRecord[\eonstar]| =n-t$}     
		 		\State  $\seed \gets \TSCombine(n, \TSthreshold, \SeedVoteRecord[\eonstar], \contenthash)$   
		 		\State  $\SeedRecord[\eonstar] \gets \seed$  
				\State $\send$ $(\SEED, \eonstar, \seed)$  to  all nodes       	
				\State $\Output$  $\seed$   
	 		\EndIf
 		\EndIf
	\EndIndent

    \State {\bf upon} receiving  a $(\SEED, \eonstar, \seed)$  message,  and  $\eonstar\notin \SeedRecord$ {\bf do}:   
	\Indent
		\If {$\TSVerify(\pk_{\eonstar}, \seed, \Hash \ltuple \SEED, \eonstar \rtuple) = \true)$}     
	 		\State  $\SeedRecord[\eonstar] \gets \seed$  
			\State $\send$ $(\SEED, \eonstar, \seed)$  to  all nodes       	
			\State $\Output$  $\seed$   
		\EndIf
	\EndIndent

\end{algorithmic}
\end{algorithm}

 \section{$\OciorADKG$}   \label{sec:OciorADKG}

We propose an $\ADKG$ protocol, called $\OciorADKG$, to generate the keys for both the $(n, \TSthreshold)$ $\TS$ scheme and the 
$(n, \TSthreshold, \ltslayerMax, \{n_{\ltslayer}, \TSthreshold_{\ltslayer}, \ltslayerTotalNodes_{\ltslayer}\}_{\ltslayer=1}^{\ltslayerMax})$ 
$\LTS$ scheme under the following constraints: 
\[
n = \prod_{\ltslayer=1}^{\ltslayerMax} n_{\ltslayer}, 
\quad 
\prod_{\ltslayer=1}^{\ltslayerMax} \TSthreshold_{\ltslayer} \geq \TSthreshold, 
\quad 
\ltslayerTotalNodes_{\ltslayer} := \prod_{\ltslayer'=1}^{\ltslayer} n_{\ltslayer'} 
\ \text{ for each } \ltslayer \in [\ltslayerMax], 
\]
with $\ltslayerTotalNodes_{0} := 1$, for some $\TSthreshold \in [t+1, n-t]$, and $n \geq 3t+1$.

The proposed $\OciorADKG$ protocol is described in Algorithm~\ref{algm:OciorADKG}, and is supported by Algorithm~\ref{algm:OciorASHVSS}. 
We also introduce a simple strictly-hiding polynomial commitment ($\SHPC$) scheme (see Definition~\ref{def:SHPC} and Fig.~\ref{fig:SPC}).  
The proposed $\SHPC$ scheme guarantees the \emph{Strict Secrecy} property: if the prover is honest, then no adversary can obtain any information about the secret $\secret$ or the corresponding public key $\randomgenerator^\secret$ until a specified timing condition is satisfied. 
This lightweight $\SHPC$ scheme is well-suited for designing efficient and simpler $\ADKG$ protocols.

In this work, we focus on describing the proposed $\OciorADKG$ protocol and the introduced primitives, while leaving detailed proofs to the extended version of this paper.

\subsection{Definitions and New Preliminaries}   \label{sec:definitions}

\begin{definition}[{\bf Strictly-Hiding Verifiable Secret Sharing ($\SHVSS$)}]    \label{def:SHVSS}
We introduce a new primitive, $\SHVSS$.  The $(n, t, \TSthreshold)$ $\SHVSS$ protocol  consists of a \emph{sharing} phase and a \emph{reconstruction} phase.  In the sharing phase, a dealer $\Dealer$ distributes a secret $\secret \in \FieldZ_{\FiniteFieldSize}$ into $n$ shares, each sent to a corresponding node, where up to $t$ nodes may be corrupted by an adversary. 
In the reconstruction phase, the protocol guarantees that any subset of at most $\TSthreshold-1$ shares reveals no information about $\secret$, while any $\TSthreshold$ shares are sufficient to reconstruct $\secret$, for $\TSthreshold \in [t+1, n - t]$. 
Unlike traditional verifiable secret sharing ($\VSS$), the $\SHVSS$ protocol guarantees an additional property: Strict Secrecy, defined below.  
Specifically, the $\SHVSS$ protocol guarantees the following properties, with probability $1 - \negligible(\kappa)$ against any probabilistic polynomial-time ($\PPT$) adversary:
\begin{itemize} 
\item \textbf{Global Secrecy:} If the dealer is honest, any coalition of up to $\TSthreshold-1$ nodes learns no information about the secret $\secret$. 
\item \textbf{Private Secrecy:} If the dealer is honest, the share held by any honest node remains hidden from the adversary. 
\item \textbf{Correctness:} If the dealer is honest and has shared a secret $\secret$,   any set of $\TSthreshold$ shares can reconstruct $\secret$. 
\item \textbf{Termination:} If the dealer is honest, then every honest node eventually terminates the sharing phase of $\SHVSS$ protocol. Furthermore, if any honest node terminates the sharing phase, then all honest nodes eventually terminate the sharing phase.   
\item \textbf{Completeness:} If an honest node terminates in the sharing phase, then there exists a  $(\TSthreshold-1)$-degree polynomial $\PolynomialFunction(\cdot) \in \FieldZ_{\FiniteFieldSize}[x]$ such that $\PolynomialFunction(0)=\secret'$ and every node~$i$, for $i\in[n]$, eventually outputs a key share  $\secret_i = \PolynomialFunction(i)$, as well as commitments of $\{\secret_j\}_{j\in[n]}$. Furthermore, if the dealer is honest and has shared a secret $\secret$, then $\secret'=\secret$.   
\item \textbf{Homomorphic Commitment:} If some honest nodes output commitments, then these commitments are additively homomorphic across different $\SHVSS$ instances. 
\item \textbf{Strict Secrecy:} If the dealer is honest, then no adversary can gain any information about the secret $\secret$ or the corresponding public key $\randomgenerator^\secret$ until a specified timing condition is satisfied, where $\randomgenerator \in \Group$ is a randomly chosen generator of a group $\Group$.  
\end{itemize} 
\end{definition}
\noindent In the asynchronous setting, we focus on asynchronous $\SHVSS$  ($\ASHVSS$).     The proposed $\ASHVSS$ protocol, called $\OciorASHVSS$, is described in Algorithm~\ref{algm:OciorASHVSS}.

\begin{definition}[{\bf Strictly-Hiding Polynomial Commitment ($\SHPC$) Scheme}]      \label{def:SHPC}
We propose a simple  strictly-hiding polynomial commitment scheme.   The proposed $\SHPC$ scheme, called $\OciorSHPC$, is presented in Fig.~\ref{fig:SPC}.  
The proposed $\SHPC$  scheme guarantees a Strict Secrecy  property:  if the prover is honest, then no adversary can gain any information about the secret $\secret:=\PolynomialFunction(0)$ or the corresponding public key $\randomgenerator^\secret$ until a specified timing condition is satisfied, where $\randomgenerator \in \Group$ is a randomly chosen generator of a group $\Group$ and  $\PolynomialFunction(\cdot) \in \FieldZ_{\FiniteFieldSize}[x]$ is  the committed polynomial.
The   $\SHPC$ scheme consists of  the following algorithms. 
\begin{itemize} 
\item {\bf $\SPCSetup(1^\kappa) \to \pp$.} This algorithm generates the public parameters ($\pp$) based on the security parameter $\kappa$.    The algorithms below all input the public parameters (and omitted in the presentation for simplicity).  
\item {\bf $\SPCCommit( \PolynomialFunction(\cdot), \CommitPolySecret, \PolyDegree, n) \to  \CommitmentVector$.}   Given the public parameters $\pp$, a polynomial $\PolynomialFunction(\cdot)$ of  degree $\PolyDegree$, the number of evaluation points $n$, and the  random witness $\CommitPolySecret$, this algorithm outputs a commitment vector $\CommitmentVector$ to a polynomial $\PolynomialFunction(\cdot)+ \HashZ(\CommitPolySecret)$.    Here,  $\HashZ(): \Mc \to \FieldZ_{\FiniteFieldSize}$ is a hash function.     

\item {\bf $\SPCWitnessCommit(   \CommitPolySecret, t, n) \to (\HashCommitVector, \WitnessVectorR)$.}   Given    the  random witness $\CommitPolySecret$, this algorithm outputs  a commitment vector $\HashCommitVector$ to the  witness $\CommitPolySecret$, along with a  witness share vector $\WitnessVectorR$.  The witness $\CommitPolySecret$ can be reconstructed by any $t+1$ valid witness shares.

 \item \textbf{$\SPCOpen( \PolynomialFunction(\cdot), \CommitPolySecret, i) \to \skshareAddWitness_{i}$.}   Given   $i\in [n]$, this algorithm outputs the evaluation $\skshareAddWitness_{i}:=\PolynomialFunction(i)+ \HashZ(\CommitPolySecret)$ of the polynomial $\PolynomialFunction(\cdot)+\HashZ(\CommitPolySecret)$. 
 
  \item \textbf{$\SPCWitnessOpen(    \WitnessVectorR, i) \to   \CommitPolySecret_{i}$.}  
  Given  $i\in [n]$, this algorithm outputs a valid  witness share $\CommitPolySecret_{i}$,  for $\WitnessVectorR=[\CommitPolySecret_1,\CommitPolySecret_2, \dotsc, \CommitPolySecret_n]$.

\item \textbf{$\SPCDegCheck( \CommitmentVector, \IndexSet, \PolyDegree) \to \true/\false$.}
Given a set of  evaluation points $\IndexSet$ and a commitment vector $\CommitmentVector$,  this algorithm returns $\true$ if $\CommitmentVector$ is  a commitment to a polynomial of degree at most $\PolyDegree$; otherwise, it returns $\false$. 
  
\item \textbf{$\SPCVerify( \CommitmentSymbol_{i},  \skshareAddWitness_{i}, \HashCommit_{i},   \CommitPolySecret_{i},  i ) \to \true/\false$.}  Given the index $i\in [n]$, the  commitment $\CommitmentSymbol_{i}$ to  the $i$-th  evaluation of a polynomial $\PolynomialFunction(\cdot)+ \HashZ(\CommitPolySecret)$ for some random witness $\CommitPolySecret$, the commitment  $\HashCommit_{i}$ to the $i$-th share of the  witness $\CommitPolySecret$,    this algorithm returns $\true$ if $\skshareAddWitness_{i}= \PolynomialFunction(i)+ \HashZ(\CommitPolySecret)$ and $\HashZ(\CommitPolySecret_{i}) =\HashCommit_{i}$; otherwise, it returns $\false$.  

\item \textbf{$\SPCWitnessReconstruct( \{(j, \HashCommit_{j})\}_{j\in[n]},   \{(j, \CommitPolySecret_{j})\}_{j\in \IndexSet}) \to \CommitPolySecret/\defaultvalue$.}    
This algorithm returns a witness $\CommitPolySecret$  if $\{\CommitPolySecret_{j}\}_{j\in T}$  includes $|T|\geq t+1$ valid witness shares that are matched to the commitments $\{\HashCommit_{j}\}_{j\in[n]}$;   otherwise, it returns a default value $\defaultvalue$.  

\item \textbf{$\SPCPolyReconstruct( \CommitmentVector, \CommitPolySecret, \skshareAddWitness_{i}, i ) \to (\skshare_{i},  \CommitmentVectorOutput)$.}    
Given the index $i\in [n]$, the decoded witness $\CommitPolySecret$, the commitment vector $\CommitmentVector$ and  the  evaluation $\skshareAddWitness_{i}=\PolynomialFunction(i)+ \HashZ(\CommitPolySecret)$ of the polynomial $\PolynomialFunction(\cdot)+\HashZ(\CommitPolySecret)$,  this algorithm outputs the evaluation $\skshare_{i}=\PolynomialFunction(i)$   and the commitment vector $\CommitmentVectorOutput$  of the polynomial $\PolynomialFunction(\cdot)$.

\end{itemize} 
\end{definition}

\begin{figure} 
\centering
\begin{tcolorbox}[title=$\OciorSHPC$ Strictly-Hiding Polynomial Commitment  Scheme, colframe=blue!55, colback=white, width=1\textwidth]
 
\small   

 {\bf $\underline{\SPCSetup(1^\kappa) \to \pp}$.}  \\ This algorithm generates the public parameters $\pp = (\Group, \FieldZ_{\FiniteFieldSize}, \randomgenerator)$ based on the security parameter $\kappa$. Here, $\Group$ is a  group of prime order $\PrimeOrder$,   $\randomgenerator \in \Group$ is a randomly chosen generator of the group, and $\FieldZ_{\FiniteFieldSize}$ is a finite field of order $\FiniteFieldSize$.    
  
 \vspace{10pt}
 
 {\bf $\underline{\SPCCommit( \PolynomialFunction(\cdot), \CommitPolySecret, \PolyDegree, n) \to \CommitmentVector}$.}\\   
 Here $\PolynomialFunction(\cdot) \in \FieldZ_{\FiniteFieldSize}[x]$ is the input polynomial of degree $\PolyDegree$, and  $\CommitPolySecret \in \FieldZ_{\FiniteFieldSize}$ is the    input random witness.  \\
  Let   $\PolynomialFunctionAddWitness(x): =  \PolynomialFunction(x)+\HashZ(\CommitPolySecret)$ be a new polynomial.  \\
 Compute  $\skshare_{i}:=\PolynomialFunction(i)$,   $\skshareAddWitness_{i}:=\PolynomialFunctionAddWitness(i)= \skshare_{i}+ \HashZ(\CommitPolySecret)$, $\forall i\in [0,n]$.   \\
    Compute and output the   commitment vector $\CommitmentVector$ for  the  polynomial   $\PolynomialFunctionAddWitness(\cdot)$: 
     \[\CommitmentVector = \bigl[ \randomgenerator^{\PolynomialFunctionAddWitness(0)}, \randomgenerator^{\PolynomialFunctionAddWitness(1)}, \randomgenerator^{\PolynomialFunctionAddWitness(2)},  \dotsc, \randomgenerator^{\PolynomialFunctionAddWitness(n)}    \bigr] = \bigl[\randomgenerator^{\skshare_{0}+ \HashZ(\CommitPolySecret)},  \randomgenerator^{\skshare_{1}+ \HashZ(\CommitPolySecret)}, \randomgenerator^{\skshare_{2}+ \HashZ(\CommitPolySecret)},  \dotsc, \randomgenerator^{\skshare_{n}+ \HashZ(\CommitPolySecret)}    \bigr].\]

  \vspace{10pt}
 
 {\bf $\underline{\SPCWitnessCommit(   \CommitPolySecret, t, n) \to (\HashCommitVector, \WitnessVectorR)}$.}  \\
 Sample   $t$-degree random polynomial   $\PolynomialFunctionNew(\cdot) \in \FieldZ_{\FiniteFieldSize}[x]$  with    $\PolynomialFunctionNew(0)= \CommitPolySecret$. \\
 Compute  $\CommitPolySecret_{i}:=\PolynomialFunctionNew(i)$, $\forall i\in [n]$.   \\
Compute    the witness share vector:  
     $ \WitnessVectorR=[\CommitPolySecret_{1},   \CommitPolySecret_{2}, \dotsc, \CommitPolySecret_{n}]$.    \\    
  Compute  hash-based  commitment vector  $\HashCommitVector$ for the  shares    of the witness $\CommitPolySecret$:   
    $\HashCommitVector=[\HashZ(\CommitPolySecret_{1}),   \HashZ(\CommitPolySecret_{2}), \dotsc, \HashZ(\CommitPolySecret_{n})]$.\\
The algorithm outputs  $(\HashCommitVector, \WitnessVectorR)$.
   
    \vspace{10pt}

{\bf $\underline{\SPCOpen( \PolynomialFunction(\cdot), \CommitPolySecret, i) \to \skshareAddWitness_{i}}$.}   \\
This algorithm outputs the evaluation $\skshareAddWitness_{i}:=\PolynomialFunction(i)+ \HashZ(\CommitPolySecret)$ of   polynomial $\PolynomialFunction(\cdot)+\HashZ(\CommitPolySecret)$. 
 
     \vspace{10pt}
     
     {\bf $\underline{\SPCWitnessOpen(    \WitnessVectorR=[\CommitPolySecret_1,\CommitPolySecret_2, \dotsc, \CommitPolySecret_n], i) \to   \CommitPolySecret_{i}}$.}   \\
Given $i\in [n]$, this algorithm outputs a valid  witness share $\CommitPolySecret_{i}$, for $\WitnessVectorR=[\CommitPolySecret_1,\CommitPolySecret_2, \dotsc, \CommitPolySecret_n]$. 
 
     \vspace{10pt}

{\bf $\underline{\SPCDegCheck( \CommitmentVector=[\CommitmentSymbol_0,\CommitmentSymbol_1, \dotsc, \CommitmentSymbol_n], \IndexSet=\{0,1,2,\dotsc, n\}, \PolyDegree) \to \true/\false}$.}   \\   
Given a set  of  evaluation points $ \IndexSet=\{0,1,2,\dotsc, n\}$ and a commitment vector $\CommitmentVector=[\CommitmentSymbol_0,\CommitmentSymbol_1, \dotsc, \CommitmentSymbol_n]$,  this algorithm   samples a random polynomial $\RandomPolyFunction(\cdot) \in \FieldZ_{\FiniteFieldSize}[x]$ with $\deg(\RandomPolyFunction) = |\IndexSet| - 2 - \PolyDegree$, and then checks the following condition:
\[
\prod_{i \in \IndexSet} \CommitmentSymbol_i^{\RandomPolyFunction(i) \cdot \LagrangeCoefficient_{i}} \stackrel{?}{=} 1_{\Group},
\]   
where $\LagrangeCoefficient_{i} = \prod_{j \in \IndexSet, j \neq i} \frac{1}{i - j}$; and the correct $\CommitmentSymbol_i$ takes the form of $\CommitmentSymbol_i=\randomgenerator^{\PolynomialFunctionAddWitness(i)}$ for a polynomial $\PolynomialFunctionAddWitness(\cdot)$. 
The algorithm returns $\true$ if the above condition is satisfied; otherwise, it returns $\false$. It is worth noting that for any polynomial $f(x)$ with $\deg(f) \leq |\IndexSet|  - 2$, it holds that:
$\sum_{i \in \IndexSet} f(i) \cdot \LagrangeCoefficient_i = 0$ (see Lemma~\ref{lm:degreecheck} in  Appendix~\ref{sec:degreecheck}).     Furthermore, given $\deg(\RandomPolyFunction) = |\IndexSet| - 2 - \PolyDegree$ and $\deg(\PolynomialFunctionAddWitness) = \PolyDegree$, and let $f(x) : = \PolynomialFunctionAddWitness(x) \cdot \RandomPolyFunction(x) \in \FieldZ_{\FiniteFieldSize}[x]$, then it is true that $\deg(f) \leq\deg(\PolynomialFunctionAddWitness)  + \deg(\RandomPolyFunction)  =|\IndexSet| - 2$.

\vspace{10pt}

{\bf $\underline{\SPCVerify( \CommitmentSymbol_{i}, \skshareAddWitness_{i},  \HashCommit_{i},   \CommitPolySecret_{i},  i ) \to \true/\false}$.}   \\
  {\bf if} $\CommitmentSymbol_{i}=  \randomgenerator^{\skshareAddWitness_{i}}$ and $\HashZ(\CommitPolySecret_{i}) =\HashCommit_{i}$  {\bf then}
  	$\return$ $\true$
 {\bf else}    
   	$\return$ $\false$

      \vspace{10pt}
      
 {\bf $\underline{\SPCWitnessReconstruct( \{(j, \HashCommit_{j})\}_{j\in[n]},    \WitnessShareValidSet:=\{(j, \CommitPolySecret_{j})\}_{j\in \IndexSet}) \to \CommitPolySecret/\defaultvalue}$.}   \\ 
 let $\IndexSet:= \{j \in [n] \mid (j, \CommitPolySecret_j) \in \WitnessShareValidSet\}$ \\
   {\bf if} $|\IndexSet|  <  t+1$  {\bf then}  $\return$ $\defaultvalue$  \\
      {\bf if} $\HashZ(\CommitPolySecret_{j}) =\HashCommit_{j}$, $\forall j \in \IndexSet$  {\bf then}   \\
         \hspace*{2em}  interpolate $\CommitPolySecret=\PolynomialFunctionNew(0)$ and all missing $\CommitPolySecret_{i}=\PolynomialFunctionNew(i)$ from $\{\CommitPolySecret_{j}\}_{j\in \IndexSet}$ using Lagrange interpolation \\
         \hspace*{2em} {\bf if} $\HashZ(\CommitPolySecret_{j}) =\HashCommit_{j}$, $\forall j \in [n]$   {\bf then}   $\return$ $\CommitPolySecret$   {\bf else} $\return$ $\defaultvalue$  \\
        {\bf else}     \\ 
        \hspace*{2em}   $\return$ $\defaultvalue$

      \vspace{10pt}
      
 {\bf $\underline{\SPCPolyReconstruct( \CommitmentVector, \CommitPolySecret, \skshareAddWitness_{i}, i ) \to (\skshare_{i},  \CommitmentVectorOutput)}$.}   \\ 
 Set $\skshare_{i}= \skshareAddWitness_{i} - \HashZ(\CommitPolySecret)$ \\
  Set $\CommitmentVectorOutput= \CommitmentVector \cdot \randomgenerator^{-\HashZ(\CommitPolySecret)}$ \\
  $\return$ $(\skshare_{i},  \CommitmentVectorOutput)$

\end{tcolorbox}
      \vspace{-8pt}
\caption{The description of the proposed strictly-hiding polynomial commitment  scheme $\OciorSHPC$. Here $t$ denotes the maximum number of dishonest nodes controlled by the adversary. $\HashZ(): \Mc \to \FieldZ_{\FiniteFieldSize}$ is a hash function. We use the degree-checking technique from~\cite{CD:17}.}
\label{fig:SPC}
\end{figure}

\begin{definition}[{\bf  $\PKI$ Digital Signatures}]    
Under the Public Key Infrastructure ($\PKI$) setup,   Node~$i$ holds a public-private key pair $(\pkpki_i, \skpki_i)$ for digital  signatures,  a public-private key  pair $(\ek_i, \dk_i)$ for verifiable encryption, and   all public keys  $\{\pkpki_j, \ek_j\}_j$.  The signing and verification algorithms for digital signatures are defined as follows:  
\begin{itemize} 
\item \textbf{$\PKISign(\skpki_i, \Hash(\wv)) \to \digitalsig_i$:}  Given an input message $\wv$, this algorithm produces the signature $\digitalsig_i$ from Node~$i$ using its private key $\skpki_i$, 
for $i \in [n]$. Here, $\Hash(\cdot)$ denotes a hash function.  
\item \textbf{$\PKIVerify(\pkpki_i,  \digitalsig_i, \Hash(\wv)) \to \true/\false$:}  
This algorithm verifies whether $\digitalsig_i$ is a valid signature from Node~$i$ on   $\wv$ using a public key $\pkpki_i$. 
It outputs $\true$ if   verification succeeds, and $\false$ otherwise.   
\end{itemize} 
\end{definition}

\begin{definition}[{\bf Verifiable Encryption ($\VE$) Scheme \cite{Groth:21}}]  
A verifiable encryption scheme for a committed message consists of the following algorithms.  
Each Node~$i$ holds a public-private key pair $(\ek_i, \dk_i)$, with $\ek_i = \randomgenerator^{\dk_i}$,  for $i \in [n]$.  Public keys $\{\ek_j\}_j$ are available to all nodes. 
The scheme employs the Feldman commitment scheme.  

\begin{itemize} 
\item {\bf $\VEbEncProve(I, \{\ek_i\}_{i\in I}, \{\secret_i\}_{i\in I}, \{\CommitmentSymbol_i\}_{i\in I}) \to (\CiphertextVector := \{(i, \CiphertextSymbol_i)\}_{i\in I}, \NIZKProofVector)$.}  
This algorithm takes as input a set of public keys $\{\ek_i\}_{i\in I}$, messages $\{\secret_i\}_{i\in I}$, and commitments $\{\CommitmentSymbol_i\}_{i\in I}$.  
It encrypts each secret $\secret_i$ with ElGamal encryption under the corresponding public key $\ek_i = \randomgenerator^{\dk_i}$, producing ciphertexts $\CiphertextSymbol_i$, where $\dk_i$ is the private decryption key held by Node~$i$. 
It also outputs a non-interactive zero-knowledge ($\NIZK$) proof $\NIZKProofVector$ attesting that, for all $i \in I$, the commitment satisfies $\CommitmentSymbol_i = \randomgenerator^{\secret_i}$ and that $\CiphertextSymbol_i$ is a correct encryption of $\secret_i$.  

\item \textbf{$\VEbVerify(I, \{\ek_i\}_{i\in I}, \{\CommitmentSymbol_i\}_{i\in I}, \CiphertextVector, \NIZKProofVector) \to \true/\false$.}  
This algorithm verifies the $\NIZK$ proof.  
It returns $\true$ if $\NIZKProofVector$ is a valid proof that, for each $i \in I$, there exists a value $\secret_i$ such that $\CommitmentSymbol_i = \randomgenerator^{\secret_i}$ and $\CiphertextSymbol_i$ is a correct encryption of $\secret_i$, where $\CiphertextVector:=\{(i, \CiphertextSymbol_i)\}_{i\in I}$.  
Otherwise, it returns $\false$.  

\item \textbf{$\VEbDec(\dk_i, \CiphertextSymbol_i) \to \secret_i$.}  
This algorithm decrypts $\CiphertextSymbol_i$ using the private key $\dk_i$ to recover the original secret $\secret_i$.  
\end{itemize} 
\end{definition}

\begin{definition}[{\bf Asynchronous Partial Vector Agreement ($\APVA$ \cite{ChenOciorABA:25})}] \label{def:APVA}
The $\APVA$ problem involves $n$ nodes, where each node~$i \in [n]$ starts with an input vector $\Me_i$ of length $n$. Each entry of $\Me_i$ belongs to the set $\Vc \cup \{\missing\}$, where $\Vc$ is a non-empty alphabet, and $\missing$ denotes a missing or unknown value (with $\missing \notin \Vc$). If $\Me_i[j] = \missing$ for some $j \in [n]$, it means node~$i$ lacks a value for the $j$-th position. 
Over time, nodes may learn missing entries and thus reduce the number of $\missing$ symbols in their vectors. The objective of the protocol is for all honest nodes to eventually agree on a common output vector $\Me$, which may still contain missing values. 
The $\APVA$ protocol must satisfy the following properties:

\begin{itemize}
    \item {\bf Consistency:} If any honest node outputs a vector $\wv$, then every honest node eventually outputs the same vector $\wv$.

    \item {\bf Validity:} For any non-missing entry $\Me[j] \neq \missing$ in the output vector of an honest node, there must exist at least one honest node~$i$ such that $\Me_i[j] = \Me[j] \neq \missing$. In addition, the output vector must contain at least $n - t$ non-missing entries.

    \item {\bf Termination:} If all honest nodes have at least $n - t$ common non-missing entries in their input vectors, then all honest nodes eventually produce an output and terminate.
\end{itemize}
\end{definition} 
In this setting, we consider the alphabet to be $\Vc = \{1\}$ and represent missing values with $\missing = 0$.
We use the efficient $\APVA$ protocol proposed in \cite{ChenOciorABA:25}, which achieves $\APVA$ consensus with an expected communication complexity of $O(n^3 \log n)$ bits and an expected round complexity of $O(1)$ rounds.
The $\APVA$ protocol in \cite{ChenOciorABA:25} uses an expected constant number of common coins, which can be generated by efficient coin generation protocols, for example, the protocol in \cite{DDL+:24}, which has an expected communication cost of $O(\kappa n^3)$ and an expected round complexity of $O(1)$.

\subsection{Overview of $\OciorADKG$}   \label{sec:OverviewOciorADKG}

The proposed $\OciorADKG$ protocol is described in Algorithm~\ref{algm:OciorADKG}, while the $\OciorASHVSS$ protocol invoked by $\OciorADKG$ is described in Algorithm~\ref{algm:OciorASHVSS}.  
In addition, $\OciorADKG$ invokes the $\APVA$ protocol proposed in~\cite{ChenOciorABA:25}.  
Fig.~\ref{fig:OciorADKG} presents a block diagram of the proposed $\OciorADKG$ protocol.  
The protocol consists of four phases: (i) $\ASHVSS$ phase, (ii) $\APVA$ selection phase, (iii) witness reconstruction phase, and (iv) key derivation phase.  
The main steps of $\OciorADKG$ are outlined below.

\subsubsection{$\OciorASHVSS$}  
The goal of this phase is to allow each honest Node~$i$ to distribute shares of a random secret $\skinput^{(i)}$ that it generates, for $i \in [n]$.  
The final secret corresponds to the agreed-upon set of secrets generated by the distributed nodes.  
Each honest node executes the following steps:  

\begin{itemize}
\item Each node~$\thisnodeindex$, $\thisnodeindex \in [n]$, randomly samples a secret $\skinput^{(\thisnodeindex)} \in \FieldZ_{\FiniteFieldSize}$ and runs $\OciorASHVSS[\ltuple \IDProtocol, \thisnodeindex\rtuple](\skinput^{(\thisnodeindex)})$.
\item In parallel, Node~$\thisnodeindex$ also participates in $\OciorASHVSS[\ltuple \IDProtocol, j\rtuple]$, $\forall j \in [n] \setminus \{\thisnodeindex\}$, along with the other nodes.
\end{itemize}

In our protocol, we employ the $\SHPC$ scheme.  
Specifically, instead of directly sharing $\skinput$, Node~$i$ shares the hidden secret
\[
\skshareAddWitness = \skinput + \HashZ(\CommitPolySecret),
\]
where $\CommitPolySecret$ is a randomly generated witness.  
Each node receives its share of $\skshareAddWitness$ and the corresponding polynomial commitments, as well as its share of the witness $\CommitPolySecret$ and the associated hash-based polynomial commitments.  

 \begin{figure} 
\centering
\includegraphics[width=17.5cm]{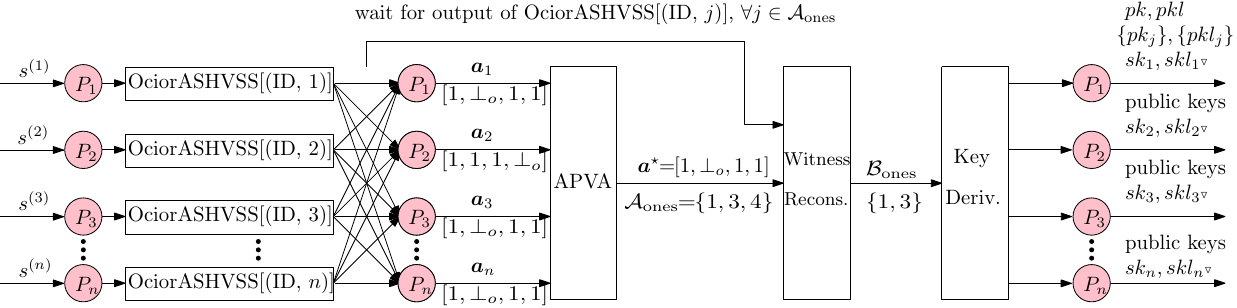}
\caption{A block diagram of the proposed $\OciorADKG$ protocol with an identifier $\IDProtocol$, where $\ABAOneSet:=\{j \in [n] \mid \APBVAoutputNew[j] = 1\}$ and $\BOneSet$ is obtained from $\ABAOneSet$ by removing the indices of nodes that did not  honestly share their secrets in the $\OciorASHVSS$ phase.  In the description, we focus on the example with $n=4$ and $t=1$. 
}
\label{fig:OciorADKG}
\end{figure}

\subsubsection{$\APVA$}
The goal of this phase is to agree on the set of $\OciorASHVSS[\ltuple \IDProtocol, j\rtuple]$ instances that correctly distribute shares of random secrets $\skinput^{(j)}$ generated by Node~$j$, for $j \in [n]$.  
Due to the $\SHPC$ scheme, agreement on the selected set is completely independent of the values of the secrets $\skinput^{(j)}$.  
In particular, the adaptive adversary cannot infer the secrets of honest nodes in order to bias or manipulate the selection set.  
Each honest node executes the following steps:  
\begin{itemize}
\item Each honest node waits for $\APVA[\IDProtocol]$ to output $\APBVAoutputNew$ such that $\sum_{j \in [n]} \APBVAoutputNew[j] \geq n-t$, and then sets
\[
\ABAOneSet := \{ j \in [n] \mid \APBVAoutputNew[j] = 1 \}.
\]
Here, $\APBVAoutputNew[j] = 1$ indicates that Node~$j$ has correctly shared its secret.  

\item Each honest node then waits for $\OciorASHVSS[\ltuple \IDProtocol, \AVSSindex\rtuple]$ to output  
\[(\{\CommitmentSymbol_{j}^{(\AVSSindex)} \}_{j\in [0, n]},   \{ \LTSCommitmentSymbol_{j}^{(\AVSSindex)}\}_{j\in [n]},   \{ \HashCommit_{j}^{(\AVSSindex)}\}_{j\in [n]},    \skshareAddWitness_{\thisnodeindex}^{(\AVSSindex)},   \skshareLTS_{\thisnodeindexshuffle}^{(\AVSSindex)},    \CommitPolySecret_{\thisnodeindex}^{(\AVSSindex)})\]  
for every $\AVSSindex \in \ABAOneSet$.  
Here $\thisnodeindexshuffle$ denotes the index of Node~$\thisnodeindex$ in the $\LTS$ scheme after index shuffling.  

The following notions are defined:  
\begin{itemize}
\item $\CommitmentSymbol_{j}^{(\AVSSindex)}$: the $j$-th evaluation of the polynomial commitment of the hidden secret
\[
\skshareAddWitness^{(\AVSSindex)} = \skinput^{(\AVSSindex)} + \HashZ(\CommitPolySecret^{(\AVSSindex)})
\]
shared by Node~$\AVSSindex$, for the $\TS$ scheme.  

\item $\LTSCommitmentSymbol_{j}^{(\AVSSindex)}$: the $j$-th commitment of $\skshareAddWitness^{(\AVSSindex)}$ for the $\LTS$ scheme.  

\item $\HashCommit_{j}^{(\AVSSindex)}$: the $j$-th commitment of the polynomial for the witness $\CommitPolySecret^{(\AVSSindex)}$ generated by Node~$\AVSSindex$. Specifically, during the $\OciorASHVSS$ phase,  Node~$\AVSSindex$ samples a $t$-degree random polynomial $\PolynomialFunctionNew^{(\AVSSindex)}(x) \in \FieldZ_{\FiniteFieldSize}[x]$ with $\PolynomialFunctionNew^{(\AVSSindex)}(0) = \CommitPolySecret^{(\AVSSindex)}$.  
It then computes the shares
\[
\CommitPolySecret_{j}^{(\AVSSindex)} := \PolynomialFunctionNew^{(\AVSSindex)}(j), \quad \forall j \in [n],
\]
and the corresponding commitments
\[
\HashCommit_{j}^{(\AVSSindex)} := \HashZ(\CommitPolySecret_{j}^{(\AVSSindex)}).
\]

\item $\skshareAddWitness_{\thisnodeindex}^{(\AVSSindex)}$: the $i$-th share of $\skshareAddWitness^{(\AVSSindex)}$ for the $\TS$ scheme.  

\item $\skshareLTS_{\thisnodeindexshuffle}^{(\AVSSindex)}$: the $\thisnodeindexshuffle$-th share of $\skshareAddWitness^{(\AVSSindex)}$ for the $\LTS$ scheme.  

\item $\CommitPolySecret_{\thisnodeindex}^{(\AVSSindex)}$: the $i$-th share of the witness $\CommitPolySecret^{(\AVSSindex)}$.  
\end{itemize}
\end{itemize}

\subsubsection{Witness Reconstruction}
In this phase, each honest Node~$\thisnodeindex$ sends its shares $\{(\AVSSindex, \CommitPolySecret_{\thisnodeindex}^{(\AVSSindex)})\}_{\AVSSindex \in \ABAOneSet}$ to all other nodes.  
This exchange enables reconstruction of the witness $\CommitPolySecret^{(\AVSSindex)}$ for all $\AVSSindex \in \ABAOneSet$.  
Since at least $t+1$ honest nodes provide correct shares matched with   publicly available commitments, each honest node can reconstruct every $\CommitPolySecret^{(\AVSSindex)}$.  
If $\CommitPolySecret^{(\AVSSindex)}$ matches all of its commitments, then all honest nodes consistently accept it and add $\AVSSindex$ to the set $\BOneSet$.  
It is guaranteed that $|\BOneSet| \geq t+1$, ensuring that at least one honest node contributes to the final secret.  

\subsubsection{Key Derivation}
In this phase each honest Node~$\thisnodeindex$ derives the keys as follows: 
 \[\sk_{\thisnodeindex} = \sum_{\AVSSindex\in \BOneSet}( \skshareAddWitness_{\thisnodeindex}^{(\AVSSindex)} -\HashZ(\CommitPolySecret^{(\AVSSindex)})); \quad  \skl_{\thisnodeindexshuffle} = \sum_{\AVSSindex\in \BOneSet}(\skshareLTS_{\thisnodeindexshuffle}^{(\AVSSindex)} -\HashZ(\CommitPolySecret^{(\AVSSindex)})) \]
\[\pk_{j} =  \prod_{\AVSSindex\in \BOneSet} \CommitmentSymbol^{(\AVSSindex)}_{j} \cdot \randomgenerator^{-\HashZ(\CommitPolySecret^{(\AVSSindex)})},  \ \forall j \in [0, n]\]   	 
\[\pkl_{j} =  \prod_{\AVSSindex\in \BOneSet} \LTSCommitmentSymbol^{(\AVSSindex)}_{j}\cdot \randomgenerator^{-\HashZ(\CommitPolySecret^{(\AVSSindex)})},  \ \forall j \in [n]\]  		
\[ \pk= \pk_{0}; \quad   \pkl= \pk_{0}.\]	 				
 Finally Node~$\thisnodeindex$ outputs   
 \[(\sk_{\thisnodeindex}, \skl_{\thisnodeindexshuffle},    \pk, \pkl, \{\pk_{j}\}_{j\in [n]}, \{\pkl_{j}\}_{j\in [n]}).\]	   	
 
The proposed $\OciorADKG$ protocol is adaptively secure under the algebraic group model and the hardness assumption of the one-more discrete logarithm problem.  
In this work, we focus on describing the proposed $\OciorADKG$ protocol and the introduced primitives, while leaving detailed proofs to the extended version of this paper.

\begin{algorithm}
\caption{$\OciorASHVSS$ protocol, with an  identifier $\IDProtocol$. Code is shown for Node~$\thisnodeindex \in [n]$.} \label{algm:OciorASHVSS}   
\begin{algorithmic}[1] 
\vspace{5pt}    
\footnotesize   
 
\Statex   \emph{//   ** This  can serve as  the $\ASHVSS$ protocol for the  $\TS$  scheme only  by simply removing the parts associated with $\LTS$   **} 
 
 \Statex
\Statex   \emph{//   ***** Public and Private Parameters     *****}

 \Statex {\bf Public Parameters:} $\pp = (\Group,  \randomgenerator)$;   $\{ \pkpki_j\}_{j\in [n]}$, $\{ \ek_j\}_{j\in [n]}$, $(n, \TSthreshold)$ for $\TS$ scheme; and  $(n, \TSthreshold, \ltslayerMax, \{n_{\ltslayer}, \TSthreshold_{\ltslayer}, \ltslayerTotalNodes_{\ltslayer}\}_{\ltslayer=1}^{\ltslayerMax})$ for $\LTS$ scheme under the constraints:  $n = \prod_{\ltslayer=1}^{\ltslayerMax} n_{\ltslayer}$,  $\prod_{\ltslayer=1}^{\ltslayerMax} \TSthreshold_{\ltslayer} \geq \TSthreshold$, and $\ltslayerTotalNodes_{\ltslayer} := \prod_{\ltslayer'=1}^{\ltslayer} n_{\ltslayer'}$ for each  $\ltslayer \in [\ltslayerMax]$,  with $\ltslayerTotalNodes_{0} := 1$,  for some $\TSthreshold\in [t+1, n-t]$, and  $n\geq 3t+1$.  $\AVSSVDelay$ is a preset delay parameter.   $\LTSIndexBook$ is a book of  index mapping for $\LTS$ scheme, available at all nodes.     
\Statex Set $\PolyDegree:= \TSthreshold - 1$ and  $\PolyDegree_{\ltslayer} := \TSthreshold_{\ltslayer} - 1$ for $\ltslayer \in [\ltslayerMax]$.    
  Set     $\jshuffle:= \LTSIndexBook[\ltuple \IDProtocol, j \rtuple], \forall j \in [n]$.

 \Statex {\bf Private Inputs:}   $\skpki_i$ and $\dk_i$     
\State  initially set $\AVSSValidSet\gets \{\}$

 \Statex
 
\Statex   \emph{//   *****  Code run by Dealer~$\Dealer$  with input $\skinput$  *****}

\Statex   \emph{//   ***** for the $\TS$ scheme   *****}

\State  sample  a   $\PolyDegree$-degree random polynomial $\PolynomialFunction(\cdot) \in \FieldZ_{\FiniteFieldSize}[x]$  with $\PolynomialFunction(0)= \skinput$, where $\skinput \in \FieldZ_{\FiniteFieldSize}$ is an input secret    
 
\State randomly generate a witness $\CommitPolySecret \in  \FieldZ_{\FiniteFieldSize}$ such that $\CommitPolySecret\neq \defaultvalue$      \label{line:OciorASHVSSWitnessBegin}

\State   $ \CommitmentVector:=[\CommitmentSymbol_0,\CommitmentSymbol_1, \dotsc, \CommitmentSymbol_n] \gets \SPCCommit( \PolynomialFunction(\cdot), \CommitPolySecret, \PolyDegree, n) $   
 
 \State   $(\HashCommitVector:=[\HashCommit_1,\HashCommit_2, \dotsc, \HashCommit_n], \WitnessVectorR:=[\CommitPolySecret_1,\CommitPolySecret_2, \dotsc, \CommitPolySecret_n]) \gets \SPCWitnessCommit(   \CommitPolySecret, t, n) $

\State   $\skshareAddWitness_{j}  \gets \SPCOpen( \PolynomialFunction(\cdot), \CommitPolySecret, j)$, $\forall j \in [n]$

\State   $ \CommitPolySecret_{j} \gets \SPCWitnessOpen( \WitnessVectorR, j)$, $\forall j \in [n]$

  \Statex		
 \Statex   \emph{//   ***** for the $\LTS$ scheme   *****}

  \State		sample  $\PolyDegree_{\ltslayer}$-degree  random polynomials   $\LTSPolynomialFunction_{\ltslayer, \ltsgroup} (\cdot)\in \FieldZ_{\FiniteFieldSize}[x]$  for each $ \ltslayer\in[\ltslayerMax]$ and  $\ltsgroup \in [\ltslayerTotalNodes_{\ltslayer-1}]$, such that:    
  \[\LTSPolynomialFunction_{1,1}(0) =\skinput ,   \quad \text{and}\quad  \LTSPolynomialFunction_{\ltslayer,\ltsgroup}(0)  =\LTSPolynomialFunction_{\ltslayer -1, \ParentIndex_{\ltslayer-1, \ltsgroup} }(\IndexInParentGroup_{\ltslayer-1, \ltsgroup}), \quad  \forall  \ltslayer\in[2, \ltslayerMax], \ltsgroup \in [\ltslayerTotalNodes_{\ltslayer-1}] \] 
   where       $\ParentIndex_{\ltslayer, \ltsgroup}:=\lceil \ltsgroup/ n_{\ltslayer} \rceil $,   $\IndexInParentGroup_{\ltslayer, \ltsgroup}:=\ltsgroup - (\lceil \ltsgroup/ n_{\ltslayer} \rceil -1)  n_{\ltslayer}$

\State   $\LTSCommitmentVector_{\ltslayer, \ltsgroup} :=[\CommitmentSymbol_{\ltslayer, \ltsgroup,0},\CommitmentSymbol_{\ltslayer, \ltsgroup,1}, \dotsc, \CommitmentSymbol_{\ltslayer, \ltsgroup, n_{\ltslayer}}]  \gets \SPCCommit( \LTSPolynomialFunction_{\ltslayer, \ltsgroup} (\cdot), \CommitPolySecret, \PolyDegree_{\ltslayer}, n_{\ltslayer}) $,  $ \forall \ltslayer\in[\ltslayerMax]$,  $\forall\ltsgroup \in [\ltslayerTotalNodes_{\ltslayer-1}]$

\State    set $\skshareLTS_{j} := \LTSPolynomialFunction_{\ltslayerMax,    \ParentIndex_{\ltslayerMax, j}} (\IndexInParentGroup_{\ltslayerMax, j})  + \HashZ(\CommitPolySecret)$,   and  $\LTSCommitmentSymbol_{j}  :=\LTSCommitmentSymbol_{\ltslayerMax, \ParentIndex_{\ltslayerMax, j}, \IndexInParentGroup_{\ltslayerMax, j}}$,    $\forall    j\in [n]$   

 \Statex
		
\State   $\send$ $\ltuple \SHARE, \IDProtocol, \CommitmentSymbol_j,  \skshareAddWitness_{j},    \HashCommit_j,      \CommitPolySecret_{j}, \LTSCommitmentSymbol_{\jshuffle}, \skshareLTS_{\jshuffle}  \rtuple$  to  Node~$j$, $\forall    j\in [n]$

\State {\bf upon}  receiving  $\ltuple \ACK, \IDProtocol, \digitalsig_j  \rtuple$  from Node~$j$ for the first time  {\bf do}:         \label{line:AVSSRxFeedbacks}
\Indent

    		\If {$\PKIVerify( \pkpki_j,  \partialsig_j, \Hash\ltuple \CommitmentSymbol_{j}, \HashCommit_{j}, \LTSCommitmentSymbol_{\jshuffle}\rtuple ) =\true$  }   
			\State   $\AVSSValidSet\gets \AVSSValidSet\cup\{\ltuple j,  \partialsig_j\rtuple \}$
 	
		\EndIf

\EndIndent

\State {\bf upon}  $|\AVSSValidSet| = n-t$   {\bf do}:       
\Indent
  
				\State  $\wait$ for $\AVSSVDelay$ time     \quad   \emph{//  to allow more valid signatures to be included, if possible, within the limited delay time}
				\State  update $\AVSSValidSet$ to include all valid signatures received from distinct nodes			 
				\State  let $\AVSSValidSetIndex:= \{j \in [n] \mid (j, \digitalsig_j) \in \AVSSValidSet\}$    and  $\AVSSMissSetIndex:=[n]\setminus \AVSSValidSetIndex$    
				\State  $(\CiphertextVector^{(0)},  \NIZKProofVector^{(0)}) \gets \VEbEncProve( \AVSSMissSetIndex,  \{ \ek_j\}_{j\in \AVSSMissSetIndex},  \{ \skshareAddWitness_j\}_{j\in \AVSSMissSetIndex},  \{ \CommitmentSymbol_j\}_{j\in \AVSSMissSetIndex} )$		  	 
 				\State  $(\CiphertextVector^{(1)},  \NIZKProofVector^{(1)}) \gets \VEbEncProve( \AVSSMissSetIndex,  \{ \ek_j\}_{j\in \AVSSMissSetIndex},  \{ \skshareLTS_{\jshuffle}\}_{j\in \AVSSMissSetIndex},  \{ \LTSCommitmentSymbol_{\jshuffle}\}_{j\in \AVSSMissSetIndex} )$	 
    				\State   $\broadcast$ $\ltuple \RBCSEND, \IDProtocol, \CommitmentVector,   \{\LTSCommitmentVector_{\ltslayer, \ltsgroup} \}_{\ltslayer\in[\ltslayerMax], \ltsgroup \in [\ltslayerTotalNodes_{\ltslayer-1}]},    \HashCommitVector,  \AVSSMissSetIndex, \AVSSValidSet$, $\{\CiphertextVector^{(\AVSSindex)},  \NIZKProofVector^{(\AVSSindex)} \}_{\AVSSindex\in \{0,1\}} \rtuple$    with  a $\RBC$     \label{line:AVSSRBC}

\EndIndent

 \Statex
  
\Statex   \emph{//   *****  Code run by each node  *****}   
 		
\State {\bf upon}  receiving  $\ltuple \SHARE, \IDProtocol,  \CommitmentSymbol_{\thisnodeindex},  \skshareAddWitness_{\thisnodeindex}, \HashCommit_{\thisnodeindex},      \CommitPolySecret_{\thisnodeindex},  \LTSCommitmentSymbol_{\thisnodeindexshuffle}, \skshareLTS_{\thisnodeindexshuffle}   \rtuple$  from Dealer~$\Dealer$ for the first time {\bf do}:

\Indent

    		\If {$\SPCVerify( \CommitmentSymbol_{\thisnodeindex}, \skshareAddWitness_{\thisnodeindex}, \HashCommit_{\thisnodeindex},    \CommitPolySecret_{\thisnodeindex},  \thisnodeindex ) =\true$ and $\SPCVerify( \LTSCommitmentSymbol_{\thisnodeindexshuffle}, \skshareLTS_{\thisnodeindexshuffle}, \HashCommit_{\thisnodeindex},    \CommitPolySecret_{\thisnodeindex},  \thisnodeindex ) =\true$ }   
			\State  $\digitalsig_\thisnodeindex \gets \PKISign(\skpki_\thisnodeindex, \Hash\ltuple \CommitmentSymbol_{\thisnodeindex}, \HashCommit_{\thisnodeindex}, \LTSCommitmentSymbol_{\thisnodeindexshuffle}\rtuple)$   
   			\State   $\send$ $\ltuple \ACK, \IDProtocol, \digitalsig_\thisnodeindex  \rtuple$  to  Dealer~$\Dealer$   
		\EndIf
		
\EndIndent

\State {\bf upon}  outputting   $\ltuple \RBCSEND, \IDProtocol, \CommitmentVector,   \{\LTSCommitmentVector_{\ltslayer, \ltsgroup} \}_{\ltslayer\in[\ltslayerMax], \ltsgroup \in [\ltslayerTotalNodes_{\ltslayer-1}]},   \HashCommitVector, \AVSSMissSetIndex, \AVSSValidSet, \{\CiphertextVector^{(\AVSSindex)},  \NIZKProofVector^{(\AVSSindex)} \}_{\AVSSindex \in \{0,1\}} \rtuple$  from a  $\RBC$ {\bf do}:         \label{line:AVSSRBCcheckBegin}
\Indent
		\State   parse  $\CommitmentVector$  as  $\bigl[ \CommitmentSymbol_0, \CommitmentSymbol_1,  \dotsc, \CommitmentSymbol_n   \bigr]$;    
		   \  parse  $\LTSCommitmentVector_{\ltslayer, \ltsgroup}$ as $\bigl[ \LTSCommitmentSymbol_{\ltslayer, \ltsgroup,0}, \LTSCommitmentSymbol_{\ltslayer, \ltsgroup,1},  \dotsc, \LTSCommitmentSymbol_{\ltslayer, \ltsgroup, n_{\ltslayer}}  \bigr]$, $\forall \ltslayer\in[\ltslayerMax], \ltsgroup \in [\ltslayerTotalNodes_{\ltslayer-1}]$		
 		\State   parse  $\HashCommitVector$  as  $\bigl[ \HashCommit_1,  \HashCommit_2,  \dotsc, \HashCommit_n   \bigr]$;   \  parse  $\AVSSValidSet$  as  $\{(j,\digitalsig_j)\}_j$;  \   parse  $\CiphertextVector^{(\AVSSindex)}$  as  $\{(j,\CiphertextSymbol_{j}^{(\AVSSindex)})\}_j$ for $\AVSSindex\in \{0,1\}$, 
 		\State         set   $\LTSCommitmentSymbol_{j}  :=\LTSCommitmentSymbol_{\ltslayerMax, \ParentIndex_{\ltslayerMax, j}, \IndexInParentGroup_{\ltslayerMax, j}}$,    $\forall    j\in [n]$   
 		\State   $\checkNew$  $\PCDegCheck( \CommitmentVector, \PolyDegree) = \true$  \label{line:AVSSDegreecheck}
 		\State   $\checkNew$  $\PCDegCheck( \LTSCommitmentVector_{\ltslayer, \ltsgroup}, \PolyDegree_{\ltslayer}) = \true$, $\forall  \ltslayer\in[\ltslayerMax]$ and  $\ltsgroup \in [\ltslayerTotalNodes_{\ltslayer-1}]$    \label{line:AVSSDegreecheckLTS}

		\State    $\checkNew$    $\PKIVerify( \pkpki_j,  \partialsig_j, \Hash\ltuple \CommitmentSymbol_{j}, \HashCommit_{j}, \LTSCommitmentSymbol_{\jshuffle}\rtuple ) =\true$,    $\forall j \in \AVSSValidSetIndex  :=[n]\setminus \AVSSMissSetIndex$ 
		\State    $\checkNew$     $\VEbVerify( \AVSSMissSetIndex,  \{ \ek_j\}_{j\in \AVSSMissSetIndex},    \{ \CommitmentSymbol_j\}_{j\in \AVSSMissSetIndex},  \CiphertextVector^{(0)},  \NIZKProofVector^{(0)}) = \true$

		\State    $\checkNew$     $\VEbVerify( \AVSSMissSetIndex,  \{ \ek_j\}_{j\in \AVSSMissSetIndex},    \{ \LTSCommitmentSymbol_{\jshuffle}\}_{j\in \AVSSMissSetIndex},  \CiphertextVector^{(1)},  \NIZKProofVector^{(1)}) = \true$   \label{line:AVSSVEVerify2}

		\State   $\checkNew$ 				 $\CommitmentSymbol_0  =   \LTSCommitmentSymbol_{1, 1,0}$;   and $\LTSCommitmentSymbol_{\ltslayer,\ltsgroup,0}  =\LTSCommitmentSymbol_{\ltslayer -1, \ParentIndex_{\ltslayer-1, \ltsgroup}, \IndexInParentGroup_{\ltslayer-1, \ltsgroup}}$,  $\forall  \ltslayer\in[2, \ltslayerMax], \ltsgroup \in [\ltslayerTotalNodes_{\ltslayer-1}]$

    		\If {all the above \emph{parallel} checks pass  }   
			\If {this node $\thisnodeindex \in \AVSSMissSetIndex$}
				\State     $\skshareAddWitness_{\thisnodeindex} \gets \VEbDec( \dk_{\thisnodeindex},   \CiphertextSymbol_{\thisnodeindex}^{(0)})$;    \      $\skshareLTS_{\thisnodeindexshuffle} \gets \VEbDec( \dk_{\thisnodeindex},   \CiphertextSymbol_{\thisnodeindex}^{(1)})$; \  			$\CommitPolySecret_{\thisnodeindex}\gets \defaultvalue$	  
			\EndIf  
			\State    $\Output$   $(\{\CommitmentSymbol_{j}\}_{j\in [0,n]},   \{ \LTSCommitmentSymbol_{j}\}_{j\in [n]},   \{ \HashCommit_{j}\}_{j\in [n]},    \skshareAddWitness_{\thisnodeindex},   \skshareLTS_{\thisnodeindexshuffle},    \CommitPolySecret_{\thisnodeindex}  )$ and $\return$        \label{line:AVSSTerminate}
		\EndIf
\EndIndent

\end{algorithmic}
\end{algorithm}

\begin{algorithm}
\caption{$\OciorADKG$ protocol, with an identifier $\IDProtocol$.   Code is shown for Node~$\thisnodeindex \in [n]$. } \label{algm:OciorADKG}   
\begin{algorithmic}[1] 
\vspace{5pt}    
\footnotesize   

\Statex   \emph{//   ***** Public and Private Parameters     *****}

 \Statex {\bf Public Parameters:} $\pp = (\Group,  \randomgenerator)$;   $\{ \pkpki_j\}_{j\in [n]}$, $\{ \ek_j\}_{j\in [n]}$, $(n, \TSthreshold)$ for $\TS$ scheme; and  $(n, \TSthreshold, \ltslayerMax, \{n_{\ltslayer}, \TSthreshold_{\ltslayer}, \ltslayerTotalNodes_{\ltslayer}\}_{\ltslayer=1}^{\ltslayerMax})$ for $\LTS$ scheme under the constraints:  $n = \prod_{\ltslayer=1}^{\ltslayerMax} n_{\ltslayer}$,  $\prod_{\ltslayer=1}^{\ltslayerMax} \TSthreshold_{\ltslayer} \geq \TSthreshold$, and $\ltslayerTotalNodes_{\ltslayer} := \prod_{\ltslayer'=1}^{\ltslayer} n_{\ltslayer'}$ for each  $\ltslayer \in [\ltslayerMax]$,  with $\ltslayerTotalNodes_{0} := 1$,  for some $\TSthreshold\in [t+1, n-t]$, and  $n\geq 3t+1$.  $\AVSSVDelay$ is a preset delay parameter.   $\LTSIndexBook$ is a book of  index mapping for $\LTS$ scheme, available at all nodes.     
\Statex Set $\PolyDegree:= \TSthreshold - 1$ and  $\PolyDegree_{\ltslayer} := \TSthreshold_{\ltslayer} - 1$ for $\ltslayer \in [\ltslayerMax]$.    
Set $\missing:=0$. 
  Set     $\jshuffle:= \LTSIndexBook[\ltuple \IDProtocol, j \rtuple], \forall j \in [n]$.

 \Statex {\bf Private Inputs:}   $\skpki_i$ and $\dk_i$

   \State  Initially set     $\APBVAinputNew_i\gets [\missing]*n$;  $\WitnessShareValidSet\gets \{\}$; $\WitnessShareValidSet[j]\gets \{\}, \forall j\in [n]$;  $\BOneSet\gets \{\}$,  $\MarkSet\gets \{\}$

\Statex   \emph{//   ***********************************************************************************  } 
\Statex   \emph{//   ***** $\OciorASHVSS$ Phase     *****} 
\Statex   \emph{//   ***********************************************************************************  }

 \State randomly sample a secret $\skinput \in \FieldZ_{\FiniteFieldSize}$  	
 		
 \State run $\OciorASHVSS[\ltuple \IDProtocol, \thisnodeindex\rtuple](\skinput)$  	\quad \quad    \emph{//  also run  $\OciorASHVSS[\ltuple \IDProtocol, j\rtuple],  \forall j\in [n]\setminus \{\thisnodeindex\}$ in parallel with all other nodes}

 \Statex
 
\Statex   \emph{//   ***********************************************************************************  } 
\Statex   \emph{//   ***** $\APVA$ Selection Phase     *****} 
\Statex   \emph{//   ***********************************************************************************  }

\State {\bf upon}  $\OciorASHVSS[\ltuple \IDProtocol, j\rtuple]$ outputting  values, for $j\in [n]$  {\bf do}:     
	\Indent  
 
		\State  input $\APBVAinputNew_{\thisnodeindex}[j] = 1$ to $\APVA[\IDProtocol]$
	\EndIndent 

 \State  $\wait$ for  $\APVA[\IDProtocol]$  to output $\APBVAoutputNew$  such that $\sum_{j\in[n]}\APBVAoutputNew[j] \geq n-t$; and then let $\ABAOneSet:=\{j \in [n] \mid \APBVAoutputNew[j] = 1\}$        		    \label{line:keyderivationAPVAoutput}

\State  $\wait$ for $\OciorASHVSS[\ltuple \IDProtocol, \AVSSindex\rtuple]$ to output  $(\{\CommitmentSymbol_{j}^{(\AVSSindex)} \}_{j\in [0, n]},   \{ \LTSCommitmentSymbol_{j}^{(\AVSSindex)}\}_{j\in [n]},   \{ \HashCommit_{j}^{(\AVSSindex)}\}_{j\in [n]},    \skshareAddWitness_{\thisnodeindex}^{(\AVSSindex)},   \skshareLTS_{\thisnodeindexshuffle}^{(\AVSSindex)},    \CommitPolySecret_{\thisnodeindex}^{(\AVSSindex)})$,   $\forall \AVSSindex \in \ABAOneSet $   

\Statex

\Statex   \emph{//   ***********************************************************************************  } 
\Statex   \emph{//   ***** Witness Reconstruction Phase     *****} 
\Statex   \emph{//   ***********************************************************************************  } 

\State   $\send$ $\ltuple \WitnessSHARE, \IDProtocol,  \{(\AVSSindex,  \CommitPolySecret_{\thisnodeindex}^{(\AVSSindex)})\}_{\AVSSindex \in \ABAOneSet}  \rtuple$  to  all nodes

\State {\bf upon}  receiving $\ltuple \WitnessSHARE, \IDProtocol,  \{(\AVSSindex,  \CommitPolySecret_{j}^{(\AVSSindex)})\}_{\AVSSindex \in \ABAOneSet}  \rtuple$   from Node~$j$ for the first time {\bf do}:    
\Indent		
 
 \For {$\AVSSindex\in  \ABAOneSet$}   
 	\If { $ \HashZ(\CommitPolySecret_{j}^{(\AVSSindex)}) =  \HashCommit_{j}^{(\AVSSindex)}$}
		\State   $\WitnessShareValidSet[\AVSSindex]\gets \WitnessShareValidSet[\AVSSindex]\cup\{(j, \CommitPolySecret_{j}^{(\AVSSindex)})\}$
		 \If {$(|\WitnessShareValidSet[\AVSSindex]| \geq t+1)\AND (\AVSSindex\notin \MarkSet)$}
		 		\State $\CommitPolySecret^{(\AVSSindex)}  \gets \SPCWitnessReconstruct( \{(j, \HashCommit_{j})\}_{j\in[n]},    \WitnessShareValidSet[\AVSSindex]) $
		 		\State $\MarkSet  \gets \MarkSet\cup\{\AVSSindex\}  $
		 		\If {$\CommitPolySecret^{(\AVSSindex)}  \neq \defaultvalue$}
		 			  $\BOneSet\gets \BOneSet\cup\{\AVSSindex\}$
		 		\EndIf
		 \EndIf
	\EndIf
\EndFor

\EndIndent						

\State  $\wait$ for   $|\MarkSet| = |\ABAOneSet|$

\Statex

\Statex   \emph{//   ***********************************************************************************  } 
\Statex   \emph{//   ***** Key Derivation Phase     *****} 
\Statex   \emph{//   ***********************************************************************************  }

		\State  $\sk_{\thisnodeindex} \gets \sum_{\AVSSindex\in \BOneSet}( \skshareAddWitness_{\thisnodeindex}^{(\AVSSindex)} -\HashZ(\CommitPolySecret^{(\AVSSindex)}))$   

		\State  $\skl_{\thisnodeindexshuffle} \gets \sum_{\AVSSindex\in \BOneSet}(\skshareLTS_{\thisnodeindexshuffle}^{(\AVSSindex)} -\HashZ(\CommitPolySecret^{(\AVSSindex)}))$

		\State    $\pk_{j} \gets  \prod_{\AVSSindex\in \BOneSet} \CommitmentSymbol^{(\AVSSindex)}_{j} \cdot \randomgenerator^{-\HashZ(\CommitPolySecret^{(\AVSSindex)})}$,  \ $\forall j \in [0, n]$   	 

		\State    $\pkl_{j} \gets  \prod_{\AVSSindex\in \BOneSet} \LTSCommitmentSymbol^{(\AVSSindex)}_{j}\cdot \randomgenerator^{-\HashZ(\CommitPolySecret^{(\AVSSindex)})}$,  \ $\forall j \in [n]$    		
				
		\State    $\pk\gets \pk_{0}$;   $\pkl\gets \pk_{0}$   	 				
 
 		\State    $\Output$   $(\sk_{\thisnodeindex}, \skl_{\thisnodeindexshuffle},    \pk, \pkl, \{\pk_{j}\}_{j\in [n]}, \{\pkl_{j}\}_{j\in [n]})$

\Statex

\end{algorithmic}
\end{algorithm}

\appendices

\section{Lagrange Coefficient Sum Lemma}  \label{sec:degreecheck}

\begin{lemma}  \label{lm:degreecheck}
Let $\IndexSet = \{x_1, x_2, \dots, x_{\NumPoints}\} \subset \FieldZ_{\FiniteFieldSize}$  be a set of  $\NumPoints:=|\IndexSet|$ distinct points, and let $f(x) \in \FieldZ_{\FiniteFieldSize}[x]$ be a polynomial such that $\deg(f) \leq \NumPoints - 2$. Define the Lagrange coefficient at each point $x_i \in \IndexSet$ as:
\begin{align}
\LagrangeCoefficient_i := \frac{1}{\prod\limits_{j \in [\NumPoints], j \ne i} (x_i - x_j)}.   \label{eq:LagrangeCoefficient}
\end{align}
Then, the following identity holds true:
\[
\sum_{i=1}^{\NumPoints} f(x_i) \cdot \LagrangeCoefficient_i = 0.
\]
\end{lemma}

\begin{proof}
We will prove this lemma using the method of partial fraction decomposition. 
First, let $P(x)$ be the polynomial defined by the product of linear factors corresponding to the points in $\IndexSet$:
\[
P(x) = \prod_{\indexi=1}^{\NumPoints} (x - x_{\indexi}).
\]
It is true that the degree of $P(x)$ is $\NumPoints$.
The   derivative of $P(x)$ is expressed as 
\[
P'(x) = \sum_{\indexi=1}^{\NumPoints} \left( \prod_{\substack{j=1 \\ j \ne \indexi}}^{\NumPoints} (x - x_j) \right).
\]
By evaluating $P'(x)$ at a specific point $x_i \in \IndexSet$, we have
\begin{align}
P'(x_i) = \prod_{\substack{j=1 \\ j \ne i}}^{\NumPoints} (x_i - x_j).       \label{eq:evaluationPx}
\end{align}
Comparing \eqref{eq:evaluationPx} with the definition of the Lagrange coefficient $\LagrangeCoefficient_i$ in \eqref{eq:LagrangeCoefficient}, it holds true that
\[
\LagrangeCoefficient_i = \frac{1}{P'(x_i)}.
\]

Next, consider the rational function $R(x) = \frac{f(x)}{P(x)}$. Since   $\deg(f) \leq \NumPoints - 2$, and  $\deg(P) =\NumPoints$, it implies that as $x \to \infty$, the rational function $R(x)$ approaches 0:
\[
\lim_{x \to \infty} R(x) = \lim_{x \to \infty} \frac{f(x)}{P(x)} = 0.
\]
We can decompose the rational function $R(x)$ into partial fractions. Since the roots $x_1, x_2, \dots, x_{\NumPoints}$ of $P(x)$ are distinct, the partial fraction decomposition takes the form:
\begin{align}
\frac{f(x)}{P(x)} = \frac{f(x)}{\prod_{\indexi=1}^{\NumPoints} (x - x_{\indexi})} = \sum_{\indexi=1}^{\NumPoints} \frac{A_{\indexi}}{x - x_{\indexi}}.      \label{eq:partialFD}
\end{align}

By multiplying  both sides of the equation \eqref{eq:partialFD}  by $(x - x_{\indexi})$ and then taking the limit as $x \to x_{\indexi}$, we compute the coefficients $A_{\indexi}$ as:
\[
A_{\indexi} = \lim_{x \to x_{\indexi}} \frac{f(x)}{\prod_{\substack{j=1 \\ j \ne \indexi}}^{\NumPoints} (x - x_j)}.
\]
Since $f(x)$ is continuous and the denominator $\prod_{\substack{j=1 \\ j \ne \indexi}}^{\NumPoints} (x - x_j)$ is non-zero at $x=x_{\indexi}$, we can substitute $x=x_{\indexi}$:
\[
A_{\indexi} = \frac{f(x_{\indexi})}{\prod_{\substack{j=1 \\ j \ne \indexi}}^{\NumPoints} (x_{\indexi} - x_j)}.
\]
By the definition of $\LagrangeCoefficient_{\indexi}$, we have $A_{\indexi} = f(x_{\indexi}) \LagrangeCoefficient_{\indexi}$. 
By substituting $A_{\indexi}$ back into the partial fraction decomposition in \eqref{eq:partialFD}, we get: 
\begin{align}
\frac{f(x)}{P(x)} = \sum_{\indexi=1}^{\NumPoints} \frac{f(x_{\indexi}) \LagrangeCoefficient_{\indexi}}{x - x_{\indexi}}.   \label{eq:partialFDNew}
\end{align} 
Now, multiply both sides of  equation \eqref{eq:partialFDNew} by $x$, it gives 
\begin{align}
\frac{x f(x)}{P(x)} = \sum_{\indexi=1}^{\NumPoints} \frac{x f(x_{\indexi}) \LagrangeCoefficient_{\indexi}}{x - x_{\indexi}}.   \label{eq:partialFDNewX}
\end{align} 
We can rewrite the term on the right-hand side  of  equation \eqref{eq:partialFDNewX} as follows:
\[
\frac{x f(x_{\indexi}) \LagrangeCoefficient_{\indexi}}{x - x_{\indexi}} = \frac{(x - x_{\indexi} + x_{\indexi}) f(x_{\indexi}) \LagrangeCoefficient_{\indexi}}{x - x_{\indexi}} = \frac{(x - x_{\indexi}) f(x_{\indexi}) \LagrangeCoefficient_{\indexi}}{x - x_{\indexi}} + \frac{x_{\indexi} f(x_{\indexi}) \LagrangeCoefficient_{\indexi}}{x - x_{\indexi}} = f(x_{\indexi}) \LagrangeCoefficient_{\indexi} + \frac{x_{\indexi} f(x_{\indexi}) \LagrangeCoefficient_{\indexi}}{x - x_{\indexi}}.
\]
Then, the equation \eqref{eq:partialFDNewX} can be rewritten as:
\begin{align}
\frac{x f(x)}{P(x)} = \sum_{\indexi=1}^{\NumPoints} \left( f(x_{\indexi}) \LagrangeCoefficient_{\indexi} + \frac{x_{\indexi} f(x_{\indexi}) \LagrangeCoefficient_{\indexi}}{x - x_{\indexi}} \right) = \sum_{\indexi=1}^{\NumPoints} f(x_{\indexi}) \LagrangeCoefficient_{\indexi} + \sum_{\indexi=1}^{\NumPoints} \frac{x_{\indexi} f(x_{\indexi}) \LagrangeCoefficient_{\indexi}}{x - x_{\indexi}}.     \label{eq:partialFDNewXNew}
\end{align} 
In the following, let's take the limit as $x \to \infty$ for both sides of the equation \eqref{eq:partialFDNewXNew}. 
For the left-hand side of \eqref{eq:partialFDNewXNew},  given that  $\deg(f) \leq \NumPoints - 2$ and $\deg(P) = \NumPoints$,    the limit as $x \to \infty$ is 0:
\begin{align}
\lim_{x \to \infty} \frac{x f(x)}{P(x)} = 0.  \label{eq:limitLeft}
\end{align} 
Let us now look at the  second sum on the right-hand side of \eqref{eq:partialFDNewXNew}. 
For each term $\frac{x_{\indexi} f(x_{\indexi}) \LagrangeCoefficient_{\indexi}}{x - x_{\indexi}}$, as $x \to \infty$, the denominator $x - x_{\indexi}$ approaches $\infty$, while the numerator $x_{\indexi} f(x_{\indexi}) \LagrangeCoefficient_{\indexi}$ is a constant. Thus, each term $\frac{x_{\indexi} f(x_{\indexi}) \LagrangeCoefficient_{\indexi}}{x - x_{\indexi}}$ approaches 0:
\[
\lim_{x \to \infty} \frac{x_{\indexi} f(x_{\indexi}) \LagrangeCoefficient_{\indexi}}{x - x_{\indexi}} = 0.
\]
Therefore, second sum on the right-hand side of \eqref{eq:partialFDNewXNew} approaches 0:
\begin{align}
\lim_{x \to \infty} \sum_{\indexi=1}^{\NumPoints} \frac{x_{\indexi} f(x_{\indexi}) \LagrangeCoefficient_{\indexi}}{x - x_{\indexi}}   = 0.   \label{eq:limitRightSecond}
\end{align}  
By substituting the limits in \eqref{eq:limitLeft} and  \eqref{eq:limitRightSecond}  back into the   equation \eqref{eq:partialFDNewXNew}, we have :
\[
0 = \sum_{\indexi=1}^{\NumPoints} f(x_{\indexi}) \LagrangeCoefficient_{\indexi} + 0.
\]
This implies the desired identity:
\[
\sum_{\indexi=1}^{\NumPoints} f(x_{\indexi}) \LagrangeCoefficient_{\indexi} = 0
\]
which completes the proof.
\end{proof}



\end{document}